\documentclass[]{scrartcl}
\usepackage{graphicx}
\usepackage{enumerate, ulem}
\usepackage{xr-hyper} 
\usepackage{bm, amssymb, amsthm, amsmath, physics, booktabs, xcolor, subfig, adjustbox, todonotes}
\makeatletter
\let\c@author\relax
\makeatother
\usepackage[style=authoryear, sorting=nyt, backend=biber, url=false, 
doi = false, date=year, maxcitenames=1, uniquelist=false, natbib,
giveninits=true, isbn=false]{biblatex}
\usepackage[colorlinks,citecolor=blue,urlcolor=blue,filecolor=blue]{hyperref}
\appto{\bibsetup}{\raggedright}
\usepackage{multirow, array, colortbl} 
\newcolumntype{$}{>{\global\let\currentrowstyle\relax}}
\newcolumntype{^}{>{\currentrowstyle}}

\newcolumntype{i}{>{\columncolor{gray!70}}c}
\newcolumntype{a}{>{\columncolor{gray!20}}c}
\newcolumntype{g}{>{\columncolor{white}}c}

\newcommand{\tphi}{\tilde{\bm{\phi}}}
\newcommand{\Db}{\mathbf{D}}

\newcommand{\Ib}{\mathbf{I}}
\newcommand{\Mb}{\mathbf{M}}
\newcommand{\Cb}{\mathbf{C}}

\newcommand{\Lb}{\mathbf{L}}

\newcommand{\Sib}{\mathbf{\Sigma}}
\newcommand{\Wb}{\mathbf{W}}
\newcommand{\Qb}{\mathbf{Q}}

\newcommand{\Qbt}{\tilde{\mathbf{Q}}}
\newcommand{\Hb}{\mathbf{H}}

\newcommand{\tSig}{\tilde{\mathbf{\Sigma}}}

\newcommand{\one}{\mathbf{1}}

\newcommand{\yb}{\bm{y}}
\newcommand{\by}{\bm{y}}
\newcommand{\tby}{\tilde{\bm{y}}}

\newcommand{\thetab}{\bm{\theta}}
\newcommand{\bX}{\bm{X}}

\newcommand{\ttheta}{\tilde{\bm{\theta}}}
\newcommand{\ty}{\tilde{\bm{y}}}
\newcommand{\bphi}{\bm{\phi}}
\newcommand{\brho}{\bm{\rho}}

\newcommand\redsout{\bgroup\markoverwith{\textcolor{red}{\rule[0.5ex]{2pt}{0.4pt}}}\ULon}

\usepackage{url} 

\date{}
\numberwithin{equation}{section}
\theoremstyle{plain}
\newtheorem{thm}{Theorem}[section]
\newtheorem{lemma}{Lemma} 
\newtheorem{definition}{Definition} 
\newtheorem{corollary}{Corollary}

\title{Structure induced by a multiple membership transformation on the Conditional Autoregressive model}

\usepackage{authblk}
\author[1]{Marco Gramatica}
\author[1, 2]{Silvia Liverani}
\author[3]{Peter Congdon}

\affil[1]{School of Mathematical Sciences, Queen Mary University of London, London, UK}
\affil[2]{The Alan Turing Institute, The British Library, London, UK}
\affil[3]{School of Geography, Queen Mary University of London, London, UK}

\begin{document}
	\maketitle

		\begin{abstract}
			The objective of disease mapping is to model data aggregated at the areal level. In some contexts, however, (e.g. residential histories, general practitioner catchment areas) when data is arising from a variety of sources, not necessarily at the same spatial scale, it is possible to specify spatial random effects, or covariate effects, at the areal level, by using a multiple membership principle (MM) \parencite{Petrof2020, Gramatica2021}. A weighted average of conditional autoregressive (CAR) spatial random effects embeds spatial information for a \textit{spatially-misaligned} outcome and estimate relative risk for both frameworks (areas and memberships). In this paper we investigate the theoretical underpinnings of these application of the multiple membership principle to the CAR prior, in particular with regard to parameterisation, properness and identifiability. We carry out simulations involving different numbers of memberships as compared to number of areas and assess impact of this on estimating parameters of interest. Both analytical and simulation study results show under which conditions parameters of interest are identifiable, so that we can offer actionable recommendations to researchers. Finally, we present the results of an application of the multiple membership model to diabetes prevalence data in South London, together with strategic implications for public health considerations.
		\end{abstract}

	\section{Introduction}
	\label{sec:intro}
	
	The scope for geo-referenced data analysis is vast, though particular methodological paradigms are more common. In the context of disease mapping, the information available is typically aggregated at areal level, thus constraining the range of viable modelling strategies. Bayesian hierarchical models have long been a widely popular choice for areal disease mapping \parencite{Lawson2018, Congdon2019, Blangiardo2015}. Our interest focuses on the conditional autoregressive prior (CAR) \parencite{Besag1974} which has been employed extensively in the disease mapping literature to account for spatial autocorrelation which is a common occurrence in areal data \parencite{MacNab2010}. 
	
	The most common framework in this context arises from the observations of an outcome of interest $ \bm{y} = (y_1, ... , y_n) $ and some explanatory covariates, both observed on the same set of \textit{n} areas that partition a certain region of interest, modelled with a distribution that includes spatial random effects such as the CAR prior. In the present work we will investigate a different type of setup, where covariates are observed at the areal level but outcomes refer to a different spatial framework.
	
	An instance of this non-standard setup can be found in \citet{Gramatica2021}, where they investigate general practitioner (GP) level data of diabetes prevalence jointly modelled with areal level mortality data for the same condition. In this instance covariates of interest are available by area (i.e. census tracts) but not for practices. However, the patients of each practice do not necessarily reside in a unique set of areas. This gives rise to two misalignment problems: first, the population of each practice can be attributed to several non-contiguous areas and second, residents of a certain area can be registered at different practices. The former type of misalignment is more related to the (non) spatial structure of the data, given that spatial modelling strategies usually require continuity. On the other hand, the second misalignment problem concerns the \textit{imperfect nesting} of areas in GP practices.
	
	Another example of the same ilk can be found in \citet{Petrof2020}, where the impact of residential history of people  affected by a rare type of cancer caused by exposure to asbestos (mesothelioma), is investigated. For each patient the probability of contracting mesothelioma is modelled using age at death as a covariate together with a spatial random effect. Clearly, given that each individual might have resided in several different locations throughout the course of their life, their residential history has to be taken into account to correctly estimate the spatial risk. Hence, the random effect is  specified in terms of a weighted average of the random effects of the areas where the individual has resided. The weights correspond to the share of time spent by the individual living in a certain area, hence the multiple membership of each area to different individuals' residential history.

	Both of the examples above can be framed as applying the multiple membership principle (MM) \parencite{Browne2001} transformation to an areal framework. Hence, we will denote this new prior as CAR-MM. Given that this model is a weighted average of areal random effects, it can be seen as a linear transformation of a CAR prior. It is well known that any covariance matrix of a normal distribution can be represented as a CAR prior \parencite{VerHoef2018}, as well as that the normal distribution is closed under affine transformations \parencite{Rao1973}. Consequently, it is possible to use tools from Gaussian distribution theory to investigate the properties and assumptions underlying the CAR-MM prior, such as its identifiability.
	
	Despite both \citet{Petrof2020} and \citet{Gramatica2021} including simulation studies to verify the performance of the model, they did not investigate the identifiability and the structure of the prior. In this paper, we develop the underpinning theory for this methodology, and provide a more general framework able to encompass a wider range of applications. In the case of \citet{Gramatica2021} as well as this paper, the weights are obtained by computing the share of population of a certain GP that resides in a specific area. This approach can take into account the fact that a single area can contribute to more than one membership (where each GP practice constitutes a membership) as well that each membership is usually defined in terms of more than one area, thus addressing the misalignment issues mentioned above.
	
	Instead of a \textit{spatial misalignment} problem, defined as ``different spatial data layers are collected at different scales"  \parencite[p. 165]{Banerjee2014}, one could consider it as a \textit{change of support} (COSP) problem. For instance, \citet{Bradley2016} start from the subregion $ D \subset  \mathbb{R}^d $, the domain of interest, and partition it according to $ L\geq 1 $ potentially non-nesting areal frameworks called \textit{source supports} i.e. the spatial support for which data is available, as opposed to \textit{target supports} where inference should take place. Due to the uniqueness requirement, this approach is not be feasible given the structure of the data that we work with.  In fact, it would require the region of interest to be partitioned into smaller areas, each of which would have to be assigned uniquely to a specific GP. If one were to follow this approach, they would fail to take into account the actual practice population structure. Obviously, with individual level data this problem could be easily overcome, however these are rarely available in disease mapping applications.
	
	The paper is organised as follows: in Section \ref{sec:CAR} we introduce the CAR model and its variant, the ICAR. In Section \ref{section:car-mm} we review the CAR-MM model, its definition and distributional form, we compare the two proposals by \citet{Petrof2020} and \citet{Gramatica2021} and introduce a more compact notation in matrix form. The main contribution of this paper is in Section \ref{sec:rel_car_carmm} which (i) introduces a more general framework for the MM models and discusses results of both uniqueness of the covariance matrix and CAR specification for the CAR-MM prior, (ii) presents two different parameterisations for the CAR-MM model and (iii) uses these results to investigate the issue of identifiability in a Bayesian framework. Section \ref{sec:val_post_checks} is mostly concerned with introducing the methodologies used for the simulation study that will follow and the model fit and model comparison methods that will be used in the data analysis. In Section \ref{sec:sim_study} simulation results are introduced in reference to the considerations of Section \ref{sec:rel_car_carmm}. Lastly, in Section \ref{sec:data_an} we consider an example on a real South London diabetes prevalence dataset providing an application of the CAR-MM prior. Three approaches are compared: one where no spatial prior accounting for spatial autocorrelation is added, and two others employing the CAR-MM and ICAR-MM priors.

	\section{Generalised Linear Models (GLM) and CAR priors} 
	\label{sec:CAR}

	When modelling area referenced data, the region of interest is generally divided into \textit{n} areas for which the outcome of interest is observed.	For each area $ i = 1, ..., n $ a random variable $ Y_i $ is used to model the observed data, whose mean is generally specified as follows:
	\begin{align} \label{eq:GLM}
		g(E(Y_i | \thetab, \bX)) = \eta_i(\thetab_i, \bX) \quad \quad \forall i  = 1, ... , n,
	\end{align}
	
	where $ \eta_i $ is a linear function of the parameters $ \thetab_i $ and the \textit{p} covariates $ \bX $, and \textit{g} is a continuous monotonic function called the \textit{link-function} \parencite{mccullagh1983generalized}. The considerations in the remainder of this paper are valid for all models such as in (\ref{eq:GLM}) however, given that for areal data of prevalence or mortality counts Poisson distribution is the most commonly assumed, for ease of presentation,  we will work under the following framework:
	
	\begin{align} \label{eq:glm_poi}
		\begin{split}
			Y_i & \sim Poisson(E_i \rho_i) \\ 
			\log(E_i \rho_i) & =  \log{E_i} + \gamma  + \bm{x}_i^T \bm{\beta} + \phi_i,
		\end{split}
	\end{align}
	
	where the link function in (\ref{eq:GLM}) is the natural logarithm, $ E_i $ is the expected number of cases for area \textit{i} and the relative risk $ \rho_i $ is modelled as a linear function of the parameters $ \thetab_i = (\gamma, \bm{\beta}, \phi_i) $. In turn, the random effects $ \phi_i $ are modelled using a CAR prior which is a random vector  $ \bm{\phi} = \left(\phi_1 , ... , \phi_n \right)^T \in \mathbb{R}^n $ defined on an undirected graph $ \mathcal{G} = (\mathcal{V}, \mathcal{E}) $.

	More specifically, using matrices $ \Mb $ and $ \Cb $ that are specified below in terms of the $ \mathcal{G} $, we define the prior as:
	\begin{align}  \label{eq:cov_car_general}
		\bm{\phi} \sim \mathcal{N}\left(\bm{0}, \Sib = (\Ib- \Cb)^{-1}\Mb \right).
	\end{align}
	\citet{VerHoef2018} state four conditions to ensure the positive definiteness of the covariance matrix $ \Sib $
	\begin{enumerate}
		\item[\textbf{C1}:] $ (\Ib - \Cb) $ has positive eigenvalues
		\item[\textbf{C2}:] $ \Mb $ is diagonal with positive diagonal elements
		\item[\textbf{C3}:] $ (\Cb)_{ii} = 0 $ for all $ i = 1, ..., n $
		\item[\textbf{C4}:] $ (\Cb)_{ij}/(\Mb)_{ii}  = (\Cb)_{ji}/(\Mb)_{jj} $ for all $ i,j = 1, ..., n $
	\end{enumerate}
	
	In practice, a very common specification that fulfils \textbf{C1} - \textbf{C4}, imposes $ \Cb = \alpha\Db^{-1}\Wb $, with $ \alpha \in \left[0, 1\right) $ and $ \Wb $ the adjacency matrix such that $ (\Wb)_{ij} = 1 $ \textit{iff} \textit{i} and \textit{j} are neighbours and 0 otherwise, and $ \Mb = \tau^{-1}\Db^{-1} $  with $ \tau \in \mathbb{R}^+ $ and $ (\Db)_{ii} = \sum_{j=1}^{n} (\Wb)_{ij}  $ for all \textit{i} thus indicating the number of neighbours of area \textit{i}. Finally, we can write the CAR prior as:
	\begin{align} \label{eq:car_prec}
		\bm{\phi} \sim & CAR(\alpha, \tau, \Wb) \quad , \quad \Sib^{-1} = \Qb = \tau(\Db - \alpha \Wb),
	\end{align}
	
	\citet{VerHoef2018} note that the representation of the CAR covariance matrix in (\ref{eq:cov_car_general}) is rather general, and in fact they prove it can represent any positive-definite covariance matrix:
	\begin{thm}
		\label{thm:verhoef_2}
		Any positive-definite covariance matrix $\boldsymbol{\Sigma}$ can be expressed as the covariance matrix of a CAR model $(\mathbf{I} - \mathbf{C})^{-1} \mathbf{M}$, for a unique pair of matrices $\mathbf{C}$ and $\mathbf{M}$.
	\end{thm}
	
	For a proof of this result see \citet{VerHoef2018}
	
	A special case of the CAR prior in equation \ref{eq:car_prec}, arises by setting the parameter $ \alpha = 1 $. This prior is called \textit{intrinsic conditional autoregressive} (ICAR) \parencite{Besag1991} and is widely used in practical applications. We will employ this prior for the data analysis in Section \ref{sec:data_an}. Its precision matrix $ \Qb $ is singular consequently leading to an improper density. In our implementation of the ICAR prior we impose the constraint $ \sum_i\phi_i = 0 $ at every MCMC iteration, as is commonly done in the literature \parencite{Eberly2000} to improve convergence and avoid intercept identifiability issues \parencite{Goicoa2018}.
	
	\section{Multiple membership and the CAR-MM prior}
	\label{section:car-mm}
	
	Both \citet{Petrof2020} and \citet{Gramatica2021} proposed a GLM model with a CAR spatial prior and MM to deal with misaligned data. We will now lay down the model that they developed independently.
	
	Let us consider a region partitioned in \textit{n} areas and define a CAR prior on it as per equation (\ref{eq:car_prec}) and let $ \tilde{\bm{y}} = (\tilde{y}_1, ... , \tilde{y}_m) $ be an outcome observed on \textit{m} units called \textit{memberships}. Typically no areal spatial framework would be available for the memberships, as the main goal of this modelling approach is to use the \textit{n} areas as a latent spatial structure.
	
	For each membership \textit{j} we define a set of \textit{n} weights $0 \leq h_{ji} \leq 1$, corresponding to each area, such that $ \sum_{i=1}^{n} h_{ji} = 1 $ for all $ j = 1,..,m $. For example, the $h_{ji}$ might be the proportions of the population of GP practice \textit{j} living in areas \textit{i}.
	We can thus re-define the relative risk term of the GLM in (\ref{eq:glm_poi}) for outcome $ \tilde{\bm{Y}} $ so that it incorporates the latent areal structure we defined earlier:
	\begin{align} \label{eq:glm_poi_mm}
		\begin{split}
			\tilde{Y}_j \sim  Poisson(&\tilde{E}_j \tilde{\rho}_j),\\
			\log(\tilde{\rho}_j) = \sum_{i=1}^{n}  h_{ji} \log(\rho_i) = 
			\sum_{i=1}^{n}& h_{ji}  (\gamma + \bm{X}_i^T \bm{\beta} + \phi_i),
		\end{split}
	\end{align}
	where $ \tilde{E} $ indicate membership level offsets. This model defines each membership level relative risk $ \tilde{\rho}_j $ as a weighted average of the areal log relative risks $\log(\bm{\rho}) $. In this case each area can contribute to multiple memberships with different weights.
	
	We can more compactly represent the log relative risks using matrix notation, for a $ m \times n $ matrix $ \Hb $ we define it element-wise using the weights in (\ref{eq:glm_poi_mm}) so that $ h_{ji} = (\Hb)_{ji} $. Additionally, we assume $ \Hb $ to be of full column rank and that for each area \textit{i} there is at least one membership \textit{j} for which $ h_{ji} > 0 $, i.e. all areas contribute to at least one membership.
	
	The second line in equation (\ref{eq:glm_poi_mm}) can then be rewritten as:
	\begin{align} \label{eq:l_RR_mm}
		\log(\tilde{\bm{\rho}}) = \Hb\log(\bm{\rho}) = \gamma\Hb\one_m + \Hb\bm{X}\beta + \Hb\bm{\phi} = \gamma\one_m + \tilde{\bm{X}}\bm{\beta} + \tphi,
	\end{align}
	in line with \citet{Gramatica2021} memberships corresponded to General Practitioner (GP) practices, on which an outcome of interest was observed, and areas to English census areal subdivisions (Middle Super-Output Area; MSOA) where the covariates of interest were recorded. The MM matrix $ \Hb $ was obtained by computing the share of patients registered at a GP practice that resided in a specific MSOA. \citet{Petrof2020} instead defined individual patients as memberships and used membership level covariates, while the MM weights corresponded to the share of time that patient \textit{j} resided in a particular area \textit{i}. Nevertheless, both approaches can be represented with equation (\ref{eq:l_RR_mm}), by simply observing that the intercept term $ \gamma $ is unchanged by the linear transformation $ \Hb $ and so are coefficients $ \bm{\beta} $.
	
	
	\section{Relationship between CAR and CAR-MM models} \label{sec:rel_car_carmm}

	We will now focus on the CAR-MM prior and study its properties. Firstly, we can observe that since the CAR-MM is a linear transformation of normally distributed terms $ \bm{\phi} $, $ \bm{\tilde{\phi}} = \Hb\bm{\phi} $, it is also normally distributed \citep{Rao1973}. In fact:	
	\begin{align} \label{eq:phi_mm_dist}
		\bm{\tilde{\phi}} \sim \mathcal{N}_m (\bm{0}, \tSig =	\Hb\boldsymbol{\Sigma}\Hb^T),
	\end{align}
	where $\Sib $ is as in equation (\ref{eq:car_prec}). Additionally, in the case where $ m \leq n$, since we assumed full column rank for $ \Hb $, the resulting covariance matrix $ \tilde{\mathbf{\Sigma}} $, with the corresponding precision matrix $ \Qbt $, will also be full rank (see Corollary 8.3.3 in \citet{harville2008matrix}). Instead, when $ m > n $, $ \tilde{\mathbf{\Sigma}} $ will be of rank $ n $ and thus singular.

	In the remainder of this section we explore the corresponding uniqueness properties of the CAR-MM. As mentioned in Theorem \ref{thm:verhoef_2}, Section \ref{sec:CAR}, \citet{VerHoef2018} show that any zero-mean Gaussian distribution on a finite set of points, $\mathbf{Z} \sim N(0, \mathbf{\Sigma})$, can be written with a covariance matrix parameterised as a CAR model.

	Firstly, let us establish three conditions that will be common to the results presented below:
	
	\begin{enumerate}
		\item[\textbf{M1}:] Let $ \Sib $ be an $ n \times n $ covariance matrix as defined in Equation (\ref{eq:car_prec})
		\item[\textbf{M2}:] Let $\Hb$ be a $m\times n$ matrix such that:
		\begin{enumerate}
			\item for each row $j=1,...,m $ of $\Hb$: $\sum_{i=1}^n (H)_{ji}=1$, $ 0 \leq (H)_{ji} \leq 1 $ for all \textit{i} and \textit{j}
			\item it is of full rank
			\item there is no column \textit{i} such that $ (\Hb)_{\cdot i} = \bm{0} $
		\end{enumerate}
		\item[\textbf{M3}:] Let $ \tSig $ be the covariance matrix of the corresponding CAR-MM model such that $ \tSig = \Hb \Sib \Hb^T  $
	\end{enumerate}
	where $ (\Hb)_{\cdot i} = \bm{0} $ indicates the \textit{i-th} column of matrix $ \Hb $.
	
	We remark that, although our approach is mainly focused on the CAR prior, the following considerations directly apply to any Gaussian model. Therefore, the spatial random effects in equation (\ref{eq:glm_poi_mm}) could also be modelled through Gaussian distributions with arbitrary covariance structure, for instance one could use a Mátern covariance function using the Euclidean distances between areas centroids \parencite[Example 6.1]{MartinezBeneito2019}.

	Above, following Theorem \ref{thm:verhoef_2}, we hinted at the fact that the CAR-MM can be seen as a CAR prior. We now introduce a corollary (proved in the Supplemental material) to that theorem, that draws a parallel between the two matrices $ \Mb $ and $ \Cb $ that determine the CAR prior structure (\textbf{C1}-\textbf{C4} in Section \ref{sec:CAR}), and the CAR-MM.

	\begin{corollary} \label{corollary:covariance}
		Let $\Sib, \, \Hb, \, \tSig $  be as defined in \textbf{M1}-\textbf{M3}, then iff $ m \leq n $, $ \tSig $ can be expressed as the covariance matrix of a CAR model $ (\Ib - \tilde{\Cb})^{-1}\tilde{\Mb} $ for a unique pair of $ m \times m $ matrices $ \tilde{\Cb} $ and $ \tilde{\Mb} $.
	\end{corollary}
	
	Since $ \Cb $ contains the off-diagonal non-zero elements of the CAR precision matrix, it determines the partial correlation structure of the prior, i.e. its zero entries correspond to partially uncorrelated components \parencite[Theorem 2.2]{Rue2005}. Therefore, the same could be said for $ \tilde{\Cb} $. However, this matrix is not straightforward to obtain. In fact, the MM transformation $ \Hb $ is applied to the covariance matrix $ \Sib $ so that to obtain the precision matrix $ \tilde{\Qb} $ we could assume $ m = n $, \textbf{M1}-\textbf{M3} and, starting from (\ref{eq:car_prec}) we could obtain: 
	\begin{align}
		\Qbt = (\Hb(\Ib - \Cb)^{-1}\Mb\Hb^T)^{-1} = (\Hb^T)^{-1}\Mb^{-1}(\Ib - (\Cb\Hb^{-1} - \Hb^{-1} + \Ib)) .
	\end{align}
	
	Given that the resulting matrices $ \tilde{\Mb} $ and $ \tilde{\Cb} $ would have to fulfil \textbf{C1-C4} (Section \ref{sec:CAR}) it is not immediately clear how they would look in terms of $ \Mb $, $ \Cb $ and $ \Hb $. In fact, if we took for example $ \tilde{\Mb}^{-1} =  (\Hb^T)^{-1}\Mb^{-1}$ and $ \tilde{\Cb} =  \Cb\Hb^{-1} - \Hb^{-1} + \Ib$ we would not have any guarantee that $ \tilde{\Mb}^{-1} $ would be diagonal as per \textbf{C2}, or $ \tilde{\Cb} $ fulfils \textbf{C3} (null diagonal elements) or \textbf{C4} $(\tilde{\Cb})_{ij}/(\tilde{\Mb})_{ii}  = (\tilde{\Cb})_{ji}/(\tilde{\Mb})_{jj} $ for all $ i,j = 1, ..., n $. Therefore, differently from the spatial CAR model, we have no way of obtaining analytically the partial correlation structure of the CAR-MM prior.

	Despite these difficulties regarding the precision matrix, we can consider the covariance matrix, together with its transformation $ \Hb $, and prove a crucial result for model identifiability. The following lemma will help evaluate under which circumstances the MM transformation of the CAR prior is injective, and will also be used in the subsequent proofs.

	\begin{lemma} \label{lemma:L_unique}
		Let $ \Hb $ be a matrix as defined in \textbf{M1}-\textbf{M3}. $ \Hb $ has a left inverse $ \Lb $ iff $ m\geq n $, and this matrix is unique iff $ m = n $.
	\end{lemma}
	
	It is worth mentioning that any left inverse $ \Lb $ of $ \Hb $ will satisfy the same constraint, i.e the columns adding up to the $ \one $ vector, in fact: $ \Hb\one_n = \one_m \, \Leftrightarrow \, \Lb\Hb\one_n=\Lb\one_m \, \Leftrightarrow \, \one_n = \Lb\one_m $ (proof in Section 1 of the Supplemental material.).
	
	\begin{thm} \label{thm:unique_Sig}
		Let $\Sib, \, \Hb, \, \tSig $  be as defined in \textbf{M1}-\textbf{M3}, then given $ \tSig $ and $ \Hb $, $ \Sib $ is uniquely defined iff $ m \geq n $.
	\end{thm}
	
	See Section 1 in the Supplemental material for a proof of both results.

	
	\subsection{Parameterisation} \label{subsec:param}
	
	In light of the results presented above, we can devise two different parameterisations for the random effects in equation (\ref{eq:glm_poi_mm}): the \textit{post} parameterisation and the \textit{inverse} parameterisation. In fact vectors $ \bm{\phi} $ and $ \tphi $ are linked deterministically by the linear transformation $ \Hb $ (see equation (\ref{eq:l_RR_mm})), so the model has to specify a prior only for one of the two.
	
	The \textit{post}-transform parameterisation consists of using the CAR prior directly to sample random effects $ \bm{\phi}^* $ from the spatial structure (equation (\ref{eq:car_prec})), and then multiply them by $ \Hb $ to obtain the MM equivalents using $ \Hb \bm{\phi}^* = \bm{\tilde{\phi}}^* $.
	
	The \textit{inverse} parameterisation is defined as follows. We specify a prior on $ \tphi^\dagger $ directly from the model in equation (\ref{eq:phi_mm_dist}) and obtain the corresponding spatial random effects using the inverse MM transformation  $ \Hb^{-}\bm{\tilde{\phi}}^\dagger = \bm{\phi}^\dagger $
	
	In terms of notation, all quantities associated with the \textit{post}-transform parameterisation will be denoted with an asterisk ($ * $) superscript, while for the \textit{inverse} parameterisation a dagger ($ \dagger $) will be used.
	
	In order to be able to specify a proper CAR prior, we need the covariance matrix to be positive definite. This means that due to Theorem \ref{thm:unique_Sig}, the \textit{inverse} parameterisation is available only in the case where $ m \leq n $. Since when \textit{m} is strictly smaller than \textit{n} by Lemma \ref{lemma:L_unique} there is no left inverse that would allow us to obtain $ \Hb^{-}\bm{\tilde{\phi}}^\dagger = \bm{\phi}^\dagger $. Instead, the \textit{post} parameterisation can be implemented in all cases, given that it is only necessary to specify a proper CAR prior on the areal structure as in Equation (\ref{eq:car_prec}). 
	
	It should be noted that whenever $ m \neq n $ the specification of the CAR-MM prior can be seen as a case of Low-rank representation of a spatial process \parencite[Chapter 8]{gelfand2010handbook}, and specifically \textit{fixed rank kriging} \parencite{Cressie2008}. In fact, if we consider the \textit{post} parameterisation for instance, given that $ \Hb \bm{\phi}^* = \bm{\tilde{\phi}}^* $, for $ m > n $ $ \bm{\phi}^* $ is a lower rank representation of $ \tphi^* $, and vice versa for $ m < n $, as one random effect can be specified in terms of a deterministic transformation of a lower rank distribution. 
	
	Low-rank representations can be particularly useful for geostatistical models as they can significantly reduce the computational burden that is generally required for spatial prediction. In this context matrix $ \Hb $ takes the name of \textit{basis} and its selection, and in some cases estimation, is generally arbitrary and depends on the problem at hand. This is different from the CAR-MM case where the MM matrix $ \Hb $ is entirely determined by the population structure of the areas and memberships. Additionally, contrary to most geostatistical applications in this context where the lower rank process is simply used as a latent factor, for the CAR-MM both the membership and the areal level, i.e. both the lower and higher rank process can be of inferential interest.

	\subsection{Identifiability} \label{subsec:ident}
	
	
	We can now investigate the identifiability of the model by introducing two definitions of non-identifiability in a Bayesian context and evaluate each one in the case of CAR-MM random effects. We will not provide explicit considerations for the hyperparameters $ \alpha $ and $ \tau $, as they are generally not parameters of interest.

	Firstly, according to \citet{Dawid1979}:
	\begin{definition}[Likelihood non-identifiability] \label{def:non-id-like}
		For a Bayesian model $ M(\bm{\theta}, \bm{y}) $ with data $\bm{y}$, likelihood $ L $ and prior $ f $, if the parameter vector $\bm{\theta}$ is partitioned into $\bm{\theta} = \left(\bm{\theta}_1, \bm{\theta}_2\right)$ then $\bm{\theta}_2$ is non-identifiable if:
		\begin{align}
			f(\bm{\theta}_2 \mid \bm{\theta}_1, \bm{y}) = f(\bm{\theta}_2 \mid \bm{\theta}_1),
		\end{align}
	\end{definition}
	\citet{Gelfand1999} remark that Definition \ref{def:non-id-like} is equivalent to the classical notion of non-identifiability of the likelihood (see Definition 1 Section 3.1.1 in \cite{cole2020parameter}) since:
	\begin{align}
		f(\bm{\theta}_2 \mid \bm{\theta}_1, \bm{y}) \propto 
		L( \bm{\theta}_1, \bm{\theta}_2 \mid \bm{y})
		f(\bm{\theta}_2 \mid \bm{\theta}_1) f(\bm{\theta}_1),
	\end{align} 
	is not identifiable according to Definition \ref{def:non-id-like} \textit{iff} $ L(\bm{\theta}_1, \bm{\theta}_2 \mid  \bm{y} ) $ can be re-parameterised so that it becomes free of $ \bm{\theta}_2 $.
	
	In addition, it is worth recalling that  \citet{Catchpole1997} define \textit{parameter redundancy} as
	\begin{definition}[Parameter redundancy] \label{def:par_red}
		A model $ M(\bm{\theta}, \bm{y}) $  is parameter redundant if we can write $ \bm{\theta} $ as a function of just $ \bm{\beta} $, where $ \bm{\beta} = g(\bm{\theta}) $ and $ dim(\bm{\beta}) < dim(\bm{\theta})$, where dim is length of the vector.
	\end{definition}
	They also prove that parameter redundancy implies non-identifiability in the likelihood which is then equivalent to Definition \ref{def:non-id-like} (Theorem 4 in \cite{Catchpole1997}). We refer back to the model in equation (\ref{eq:phi_mm_dist}) and conditions \textbf{M1}-\textbf{M3}  imposed on $ \Hb $ in Section \ref{sec:rel_car_carmm} and evaluate the two parameterisations under the definitions we just presented.

	In the case of a \textit{post}-parameterised model, it is evident that when $m > n$ the MM random effects $ \tphi^* $ are parameter redundant as they can be written as a function of the sole areal random effects $ \bm{\phi}^*  $ which are \textit{n} dimensional. More specifically, there is an injective function $ g $ such that $ g(\bm{\phi}^* ) = \Hb \bm{\phi}^*  = \tphi^*  $, where $ g: \mathbb{R}^n \to \mathbb{R}^m $.

	This can also be seen more explicitly, by partitioning the model parameters in equation (\ref{eq:glm_poi_mm}) into $ (\bm{\phi}, \tphi, \bm{\theta}) $, where $ \bm{\theta} = (\gamma, \bm{\beta}, \alpha, \tau) $, we can write the posterior distribution, up to a constant, as
	\begin{align} \label{eq:post_id1}
		\begin{split}
			L(\tby | \tphi^*, \bm{\phi}, \thetab, \Hb, \bX)f(\tphi^*, \bm{\phi}^*, \thetab) & = 
			L(\tby | g(\bm{\phi}^*), \bm{\phi}, \thetab, \Hb, \bX)f(g(\bm{\phi}^*) \bm{\phi}^*, \thetab) \\
			& = L(\tby | g(\bm{\phi}^*), \thetab, \Hb, \bX)f(\bm{\phi}^*, \thetab) 
		\end{split}
	\end{align}
	where the second equality comes from the fact that the $ \tphi^* $ are a deterministic function exclusively of the $ \bm{\phi}^*  $ and thus do not depend on either $ \thetab  $ or $ \{\Hb, \bX\} $. Consequently, in (\ref{eq:post_id1}) we can see that the MM random effects $ \tphi $ are non-identifiable in the likelihood (Definition \ref{def:non-id-like}). In the case of the \textit{post} parameterisation, this holds true for all \textit{m} and \textit{n}.
	
	With regard to the \textit{inverse} parameterisation, where a prior is specified on the MM random effects $ \tphi^\dagger $, in the case $ m \geq n $ by Lemma \ref{lemma:L_unique} there is a left inverse such that $ \bm{\phi}^\dagger = \Lb \tphi^\dagger $. Thus the same considerations hold as for the \textit{post} parameterisation made above. 
	However in this case the areal and MM random effects exchange statuses in terms of identifiability and redundancy. To see this, suffices to swap  $ \bm{\phi}^* $ for $ \tphi^\dagger $  and \textit{k} for \textit{g} in (\ref{eq:post_id1}), where $ \Lb\tphi^\dagger = \bm{\phi}^\dagger = k(\tphi^\dagger)$ and $ k:\mathbb{R}^m \to \mathbb{R}^n$. In the $ m < n $ case for the \textit{inverse} parameterisation, by Theorem \ref{thm:unique_Sig} the areal random effects covariance matrix $ \Sib^\dagger $ is not uniquely defined and thus the $ \bm{\phi}^\dagger $ are not identifiable in the likelihood.

	
	Despite these issues with the likelihood, we can adopt a different perspective following a different definition, by \citet{cole2020parameter} (Definition 12):
	\begin{definition}[Posterior non-identifiability] \label{def:non-id-post}
		A Bayesian model $ M(\bm{\theta}, \yb) $ described by the posterior distribution, is globally identifiable if $ f(\bm{\theta}_a \mid \yb) = f(\bm{\theta}_b \mid \yb) $ implies $ \bm{\theta}_a = \bm{\theta}_b $. A model is locally identifiable if there exists an open neighbourhood of any $ \bm{\theta} $ such that this is true. Otherwise the posterior is non-identifiable.
	\end{definition}
	
	Under this perspective we can see that both membership and areal random effects are identifiable in the posterior for the post parameterisation when $ m\geq n $. In fact, by setting $ \bm{\Delta} = \{ \tilde{\yb}, \Hb, \bm{X} \} $, we can investigate identifiability by writing down the posterior distribution as:
	\begin{align} \label{eq:post_id2}
		\begin{split}
			f(g(\bm{\phi}^*_a) \mid \bm{\phi}^*_a, \bm{\theta}_a , \bm{\Delta}) 
			f(\bm{\phi}^*_a, \bm{\theta}_a \mid \bm{\Delta}) & = 
			f(g(\bm{\phi}^*_b) \mid \bm{\phi}^*_b, \bm{\theta}_b , \bm{\Delta}) 
			f(\bm{\phi}^*_b, \bm{\theta}_b \mid \bm{\Delta}).
		\end{split}
	\end{align}
	In the $ m \geq n $ case, we can easily see that since the function $ g $ is injective, we have that if and only if $ \bm{\phi^*}_a = \bm{\phi^*}_b $ then $ g(\bm{\phi}^*_a) = g(\bm{\phi}^*_b) $. In turn, this entails that for the purpose of verifying Definition \ref{def:non-id-post}, (\ref{eq:post_id2}) reduces to:
	\begin{align} \label{eq:post_id2_cont}
		f(\bm{\phi}^*_a, \bm{\theta}_a \mid \bm{\Delta})  = 
		f(\bm{\phi}^*_b, \bm{\theta}_b \mid \bm{\Delta})
	\end{align}
	where the terms $ \bphi^* $ are simple CAR random effects.
	
	Therefore, for $ m \geq n $ Equation \ref{eq:post_id2_cont} shows that a post parameterised CAR-MM model as in (\ref{eq:glm_poi_mm}), is identifiable if and only if the underlying CAR model (\ref{eq:glm_poi}) is as well. Conversely, we remark that in the $ m < n $ case, given that \textit{g} is not injective, (\ref{eq:post_id2}) does not hold in general, thus making the random effects non-identifiable in the posterior.
	
	In the case of the \textit{inverse} parameterisation for $ m > n $, Equation \ref{eq:post_id2} remains the same with simply the $ \bm{\phi}^* $ swapped for $ \tphi^\dagger $ and \textit{k} for \textit{g}. However, in this instance the function \textit{k} is not injective, thus leading to non-identifiability. Furthermore, for $ m < n $ there is no function \textit{k}, given that $ \Hb $ has no left inverse. Only the $ m = n $ case is identifiable under both parameterisations.
	
	With regard to the covariates coefficients $ \bm{\beta} $, it is apparent from equation \ref{eq:l_RR_mm} that the MM transformation of the log relative risks $ \log(\bm{\rho}) $ does not affect these parameters. Therefore, the question of identifiability of both the areal and membership relative risks hinges on the random effects $ \bm{\phi} $ and $ \tphi $ as discussed above.

	To conclude, we recommend the use of the CAR-MM model only when a $ m \geq n $ dataset is available. Only in this case, in fact, parameters of interest, such as the relative risks as well as risk factor coefficients will be identifiable. 
	
	
	\section{Model validation and posterior checks}
	\label{sec:val_post_checks}
	\subsection{Simulation based Calibration (SBC)}
	
	\label{subsec:sbc}
	In order to both check the software implementation of our posterior sampler and verify that the model is able to correctly recover the true parameters, we devised a simulation study following a simulation based calibration (SBC) approach \parencite{Talts2018}. In an SBC study, instead of comparing posterior samples to some fixed \textit{ground truth} parameter, the procedure assesses the overall calibration of the posterior samples with respect to the joint distribution $ \pi(\by, \thetab) $. In this way, it is possible to verify the performance of the posterior sampler for multiple \textit{ground truths} $ \ttheta $ generated from the prior distribution.
	
	In summary, the SBC procedure can be written as follows:
	\begin{align}
		\ttheta \sim \pi(\thetab), \quad \ty \sim l(\by | \ttheta), \quad \{\thetab^{(1)}, ... , \thetab^{(B)}\} \sim \pi(\thetab | \ty).
	\end{align}
	Since we have $ \pi(\by, \thetab) = \pi(\thetab)l(\by | \thetab)$, by sampling a parameter vector $ \ttheta $ from the prior and then generate new data from the likelihood $ \ty \sim l(\by | \ttheta) $, this would be equivalent to drawing a sample from the joint distribution $ (\ty, \ttheta) \sim \pi(\by, \thetab) $. Therefore, whenever we implement an algorithm that correctly samples from a posterior distribution, the $ \{\thetab^{(1)}, ... , \thetab^{(B)}\} $ are going to have the same marginal distribution as $ \tilde{\thetab} $ sampled from the prior \parencite{Cook2006}.

	On the basis of this, we can straightforwardly check that the implemented algorithm performs inference correctly by verifying that prior and posterior samples are distributed according to the same distribution \parencite{Talts2018}: for any one-dimensional random variable $ f:\bm{\Theta} \to \mathbb{R} $, the \textit{rank statistic} $ r(\cdot) $ of the prior sample relative to the posterior samples:
	\begin{align} \label{eq:rank_stat}
		r(\left\{f(\thetab^{(1)}), ..., f(\thetab^{(B)})\right\}, f(\tilde{\thetab})) =
		\sum_{b=1}^{B} \mathbb{I} \left[f(\thetab^{(b)})  < f(\tilde{\thetab})\right] \in \left[0, B\right],
	\end{align}
	will be uniformly distributed across the integers $ \left[0, B\right] $. Therefore, by sampling from the Bayesian joint distribution \textit{N} times, and for each sample drawing \textit{B} posterior samples, computing the rank statistic in equation (\ref{eq:rank_stat}) and finally checking it is actually uniformly distributed, we can build a validation procedure for any Bayesian model. We mention that in practice, given that we have both \textit{N} and \textit{B} very large, we divide the sampled ranks by $ B + 1 $ and verify calibration with the Uniform distribution on $ [0, 1] $
	
	\subsection{Posterior checks and model comparison measures}
	Before illustrating the simulation results we introduce the methodologies used to assess both model fits and comparisons. Posterior predictive p-values (PPP) are a widely used tool in Bayesian analysis to compare the estimated posterior distribution with the observed data \parencite{Gelman2013a}. In the case of marginal predictive checks, model calibration can be verified in two steps: first generating new observations $ \bm{y}^{rep} $ using posterior samples and secondly by computing the marginal PPP's, for each observation $ \bm{y} = \left(y_1, ..., y_m\right)^T$:
	\begin{align}
		p_j = P\left(T( y_j^{rep}) < T(y_j) \mid \bm{y}\right)  +
		\frac{1}{2}P\left(T( y_j^{rep}) = T(y_j) \mid \bm{y}\right),
	\end{align}
	if $ \bm{y} $ is a discrete outcome, where $ T(\cdot) $ is a generic test quantity; in our case we will take the identity function $ T(y_j) = y_j $.

	A computationally inexpensive alternative is the one from \citet{Marshall2003} where they propose a method to approximate this ``leave-one-out" cross-validatory model fit measure: $ p_{j|-\bm{j}} = P\left(y_j^{rep} \leq y_j \mid \bm{y}_{-j}\right) $, which for continuous data will have a uniform distribution if the model is well calibrated. Referring to Equation \ref{eq:glm_poi}, we can simulate a new random effect for area \textit{i} using $ \phi_i \sim p(\phi^{rep}_i \mid \alpha, \tau, \bm{\phi}_{-i}) $ using posterior samples of $ \alpha $ and $ \tau $ and thus approximate the $ p_{j|-\bm{j}} $'s above. This approach can be directly extended to the CAR-MM model, by simply post-multiplying the quantities in the relative risk by the MM matrix $ \Hb $ as per equation (\ref{eq:glm_poi_mm}).
	
	With regard to model comparison, in the data analysis Section we compare different models fit using three criteria. The first one is Pareto smoothed importance-sampling leave-one-out cross-validation (PSIS-LOO) \parencite{Vehtari2017}.
	Similarly to the approach described above for PPP, this method make use of a conditional independence assumption of data given model parameters $ \thetab $, which in our case is implicit in Equations (\ref{eq:glm_poi}) and (\ref{eq:glm_poi_mm}). The PSIS-LOO approach allows to compute the \textit{expected log pointwise predictive density} (elpd) to estimate out-of-sample predictive fit: $elpd_{loo} = \sum_{j=1}^{m} \log p(y_j \mid \bm{y}_{-j})$. Without the need to fit the model \textit{m} times, estimates of $ elpd_{loo} $ can be obtained using
	
	\begin{align}
		\widehat{elpd}_{psis-loo} = \sum_{j=1}^{m} \log \left(
		\frac{\sum_{s=1}^{S} w_j^s p(y_j \mid \bm{\theta}^s)}{\sum_{s=1}^{S} w_j^s}
		\right),
	\end{align} 
	where $ s = 1,..., S $ indexes the posterior sample iterations, and $ w_j^s $ the Pareto smoothed importance weights.
	Model comparison can be carried out by computing the difference in $ \widehat{elpd}_{loo} $ between the two models, so for models A and B: $\text{diff}(\widehat{elpd}_{loo}^{A,B}) = 
	\widehat{elpd}_{loo}^A - \widehat{elpd}_{loo}^B $, where a larger $ \widehat{elpd}_{loo} $ indicates better estimated predictive accuracy of the model.
	The standard error of the difference is thus:
	\begin{align} \label{eq:se_elpd_diff}
		se(\text{diff}(\widehat{elpd}_{loo-psis}^{A,B})) = 
		\sqrt{m V_{j=1}^m \left(\text{diff}(\widehat{elpd}_{loo-psis}^{A,B})\right)}.
	\end{align}
	where $ V_{j=1}^m(a_j) $ is the sample variance of \textit{m} values with mean $ \bar{a} $: $V_{s=1}^S(a_s) =  \frac{1}{m-1} \sum_{j=1}^{m} (a_j -\bar{a})^2   $. The paired estimate of equation (\ref{eq:se_elpd_diff}) relies on the two models being fitted on the same set of \textit{m} data points. 
	
	The $ \widehat{elpd}_{loo} $ (we drop ``psis" in the subscript for conciseness) is easily computable using the R package \verb|loo| \parencite{loo} in conjunction with the RStan package \parencite{Carpenter2017} for sampling from the posterior distribution. The authors suggest checking the estimated shape parameters $ \hat{k} $ of the Pareto weights for all observation to assess reliability of estimates: when $ \hat{k} > 0.7 $ there might be stability issues.

	Additionally, we also considered the mean Ranked Probability Score for count data ($ \overline{RPS} $)  \parencite{Czado2009}
	\begin{align} \label{eq:rps}
		\overline{RPS} = \frac{1}{m} \sum_{j=1}^{m} \left[\frac{1}{B} \sum_{i=1}^{B} \abs{y^{rep(i)}_j - y_j} - \frac{1}{B} \sum_{i=1}^{B/2} \abs{y^{rep(i)}_j - y^{rep(i+B/2)}_j}\right],
	\end{align}
	where $ y^{rep(i)}_j $ is the \textit{i-th} replicate of observation $ y_j $. Finally,  the Dawid-Sebastiani Score ($ \overline{DSS} $) \parencite{Dawid1999} was also included in our comparisons
	\begin{align}
		\overline{DSS} = \frac{1}{m} \sum_{j=1}^{m}\left[
		\left(\frac{y_j - \overline{y^{rep}_j}}{\sigma_{y_j^{rep}}}\right) ^ 2 +
		2\log \sigma_{y_j^{rep}}
		\right],
	\end{align}
	where $ \overline{y^{rep}_j} $ and $ \sigma_{y_j^{rep}} $ are respectively the mean and standard deviations of the replicates for the \textit{j-th} observation. For both $ \overline{RPS} $ and $ \overline{DSS} $ a lower value indicates better fit for the model.
	
	
	\section{Simulation study}
	\label{sec:sim_study}

	We now turn to a simulation study designed to assess the considerations on identifiability of Section \ref{sec:rel_car_carmm}. We designed the simulations to evaluate the two parameterisations (\textit{post} and \textit{inverse}) and the number of memberships relative to areas, i.e. $ m < n $, $ m = n $ and $ m > n $. We will first describe the implementation of the posterior sampler for each model, then the simulation setup in detail and finally the results.
	
	\subsection{Model Implementation} 
	\label{subsec:implementation}
	
	All models are implemented in RStan with significant differences between the two parameterisations. For the \textit{post} parameterised covariance matrices, we followed \citet{Joseph2016}'s sparse representation combined with the  approach suggested by \citet{Jin2005} for the determinant. The former relies on a sparse representation of the matrix $ \Wb $ whereby the terms corresponding to non-neighbouring areas are  from the computation. With regard to the determinant,  following \citet{Jin2005} we can avoid inverting the entire precision matrix $ \Qb $ at each iteration. Taking $ \tau = 1 $, without loss of generality, it can be shown that: $ \Qb = \Db - \alpha \Wb =  \Db^{1/2}(\Ib - \alpha\Db^{-1/2}\Wb\Db^{-1/2})\Db^{1/2} $ then $ \det(\Qb) = \det(\Db)\prod_{i=1}^{n} (1 - \alpha\lambda_i) $, where the $ \lambda_i $'s are the eigenvalues of $ \Db^{-1/2}\Wb\Db^{-1/2} $. This representation significantly speeds up the computation of $ det(\Qb) $ since the  $ \lambda_i $'s do not depend on either $ \alpha $ and $ \tau $ and thus have to be calculated only once.
	
	Conversely, for the \textit{inverse} parameterisation we specified the entire $ m\times m $ precision matrix $ \Qbt $. We remark that in this instance, the sampler can only be applied to the $ m \leq n $ case, as it is the only case where the covariance matrix is proper. Under this parameterisation, in order to obtain the areal random effects $ \bphi^\dagger $ we used the Moore-Penrose inverse of the $ \Hb $ matrix.
	
	The code for both the simulation study and the RStan models can be found at \url{https://github.com/markgrama/CAR-MM}.

	With regard to the model, we used the following  priors (see  equation (\ref{eq:glm_poi_mm})): a uniform on $ (0,1) $ for $ \alpha $, a $ Gamma(2, 0.2) $, parameterised in terms of shape and rate, for the marginal precision $ \tau $, while for $ \gamma $, $ \beta_1 $ and $ \beta_2 $ three independent $ \mathcal{N}(0, 0.7^2) $ distributions were used.

	In terms of computational time, with the same number of iterations (see Section \ref{subsec:ex_comp_par} for details), the \textit{post} parameterisation performed significantly better than the \textit{inverse}, with an average runtime of 2 minutes across membership sizes, against the 10 minutes for the $ m = 70 $ case and 29 minutes for $ m = 100 $ (Table 1 in the Supplemental Material). Our hypothesis is that the \textit{post} parameterisation runs more quickly due to the fact that the sampling of the CAR random effect is done through the sparse representation mentioned above, and the linear transformation $ \Hb $ is only evaluated in the likelihood with RStan optimised matrix operations.
	
	Conversely, the \textit{inverse} parameterisation requires the evaluation of the entire precision matrix, since no adjacency matrix is directly available. This is due to the problem in identifying the partial correlation structure of the CAR-MM mentioned in Section \ref{sec:rel_car_carmm}. Additionally, even if we were able to determine partial independence relationships between components, there is no guarantee that the resulting precision matrix will be sparse.

	\subsection{Setup}
	\label{subsec:ex_comp_par}

	In this simulation study we use the \textit{post} and \textit{inverse} parameterisations to both generate the data and sample from the posterior, thus obtaining 4 different scenarios. In order to refer to a particular scenario we might write \textit{post}-\textit{inverse}, where the first parameterisation refers to the data generation (\textit{post}), while the second indicates the MCMC sampler parameterisation (\textit{inverse}). For each scenario we consider three membership sizes ($ 70, 100, 130 $) so that we can correspondingly evaluate the three cases $ m < n $, $ m = n $, and $ m > n $. However, as we remark in Section \ref{subsec:param}, the \textit{inverse} parameterisation is available only for $ m \leq n $, therefore we simulate the $m=130$ case only for the \textit{post} parameterisation.

	The spatial structure used in the study is a $ 10 \times 10 $ grid, hence $ n = 100 $. Expected rates for each memberships ($ \tilde{E}_j $ in equation (\ref{eq:glm_poi_mm})) were simulated from $ m $ independent $ Poisson(20) $ distributions. For both the MM matrix and the $ \tilde{\bm{E}} $ vector, we simulated the case $ m = 130 $ and removed the extra rows or elements to obtain the other membership sizes.
	
	With regard to the MM matrix $\Hb$, this was simulated starting from each area's neighbourhood structure, encoded in the adjacency matrix $\Wb$, up to second order neighbours. For each area a new membership was created by first generating random weights from an uniform distribution to all first and second order neighbours. Then those quantities were normalised so that the area itself would amount to 50\% of the total weights, the sum of all first order neighbours would be attributable 35\% of the total weight, and the remaining 15\% would be taken up by the second order neighbours. More formally, the simulated $m \times n$ matrix $\Hb $ is obtained by first randomly drawing samples:
	
	\begin{align} \label{eq:mm_mat_1}
		\begin{split}
			(G)_{ji} \overset{iid}{\sim} \mbox{Uniform}(0,1) \quad &\mbox{for}\quad  i=1,...,n\text{ and } j = 1,...,m  \\
			\mbox{s.t.} \quad \{w_{ji}=1\} &\vee \{w_{il}=1 \wedge w_{lj}=1\},
		\end{split}
	\end{align}
	for $l = 1, ..., n $. Therefore the elements of the matrix $\Hb$ are defined as follows
	\begin{align} \label{eq:mm_mat_2}
		(H)_{ji} =  \begin{cases}
			0.5 & \mbox{if } i=j,\\
			\frac{(G)_{ji}}{\sum_{j:w_{ji}=1}(G)_{ji}}\cdot 0.35 & \mbox{if } w_{ji}=1,\\
			\frac{(G)_{ji}}{\sum_{j:\{w_{il}=1 \wedge w_{lj}=1\}}(G)_{ji}}\cdot 0.15 & \mbox{if } w_{il}=1 \mbox{ and } w_{lj}=1.
		\end{cases}
	\end{align}
	
	In order to obtain a 130 membership matrix, the extra rows were created by simply repeating the same process and generating different random weights. Before deleting the bottom rows to obtain the varying number of memberships, in order to ensure that all areas were represented as uniformly as possible in the final matrix its rows were randomly shuffled. Finally, the obtained 70 membership matrix was checked to ensure that each area was represented with at least one weight. In our case the area with the least weights had 2, so it contributed to 2 memberships.

	A total of $ 10^4 $ datasets were simulated for the  memberships structure using the \textit{post} parameterisation, each membership size was obtained by simply truncating the simulated outcome vector accordingly, thus increasing the total of datasets used to 9000. Covariates $ \bm{X} $ were kept fixed across all datasets and they were generated using two independent standard normal samples and min-max normalised. Both the MM matrix and offsets were also the same in all simulated datasets. 
	
	Following the approach described in Section \ref{subsec:sbc}, hyperparameters for the CAR prior where sampled from the same priors described in Section \ref{subsec:implementation} as well as \textit{Hamiltonian Monte Carlo} (HMC) \parencite{Betancourt2018} sampling which was run using RStan \parencite{Carpenter2017} on Queen Mary's Apocrita High Performance Computing facility, supported by QMUL Research-IT \parencite{king_thomas_2017_438045}, using five $ 10^4 $ iterations chains and 50\% warm-up. To reduce autocorrelation, as suggested in \citet{Talts2018}, samples were thinned by 5 leading to $ L = 5000 + 1 $ posterior samples (see Section \ref{subsec:sbc}). 
	
	Additionally, in order to ensure convergence of the posterior sampler, we only keep those simulations for which no parameter report a $ \hat{R} $ statistic over $ 1.01 $, following \citet{Vehtari2020}. In Table 2 in the Supplemental material we report the number of simulations that were excluded. Convergence was satisfactory overall, since we always retained at least 98\% of the simulations.

	Finally we evaluate bias and root mean square error (RMSE) across simulations by computing the following quantities. For a generic parameter $ \theta $ with true value $ \hat{\theta}_k $ for simulated dataset $ k $, $ k = 1,...,10^4 $ we compute 
	\begin{align} \label{eq:bias}
		\text{Bias}_k = \frac{1}{B} \sum_{b=1}^{B} \theta^{(b)}_k - \hat{\theta}_k \, \, ; \, \,
		\text{Abs bias}_k = \frac{1}{B} \sum_{b=1}^{B} \abs{\theta^{(b)}_k  - \hat{\theta}_k}
		\, \, ; \, \,
		\text{RMSE}_k = \sqrt{ \frac{1}{B} \sum_{b=1}^{B} (\theta^{(b)}_k  - \hat{\theta}_k)^2 } ,
	\end{align}
	where $  \theta^{(b)}_k  $ is the $ b^{th} $ HMC sample for the $ k^{th} $ dataset and $ B = 2.5 \cdot 10^4 $.

	\subsection{Results}
	
	\begin{table}[ht]
		\centering
		\begin{tabular}{lgai|ga|ga|ga}
			\toprule
			&\multicolumn{5}{c}{\textbf{Data generation: Post}} & \multicolumn{4}{c}{\textbf{Data generation: Inverse}} \\
			\hline
			&\multicolumn{3}{c}{\textbf{MCMC: Post}} & \multicolumn{2}{c}{\textbf{MCMC: Inverse}} &
			\multicolumn{2}{c}{\textbf{MCMC: Post}} & \multicolumn{2}{c}{\textbf{MCMC: Inverse}} \\
			\hline
			&$70$  & $100$  & $130$ & $70$& $100$& $70$  & $100$  & $70$ & $100$  \\ 
			\midrule
			$\alpha$ & 87.35 & 66.62 & 80.61 & 37.75 & 73.45 & 43.16 & 76.44 & 95.80 & 71.72 \\ 
			$\tau$ & 68.52 & 93.80 & 94.43 & 8.92 & 93.36 & 30.82 & 90.85 & 91.47 & 91.54 \\ 
			$\gamma$ & 98.38 & 99.42 & 98.60 & 99.70 & 99.23 & 99.96 & 99.52 & 99.95 & 99.57 \\ 
			$\beta_1$ & 96.15 & 96.14 & 83.34 & 99.79 & 97.76 & 99.99 & 96.41 & 99.96 & 97.25 \\ 
			$\beta_2$ & 99.97 & 99.96 & 99.75 & 99.96 & 99.77 & 99.95 & 99.52 & 99.81 & 99.46 \\ 
			$\bm{\phi}$ & 94.63 & 94.33 & 94.22 & 50.72 & 94.53 & 57.72 & 94.53 & 93.83 & 94.54 \\ 
			$\bm{\rho}$ & 94.45 & 94.37 & 93.51 & 26.35 & 94.68 & 59.30 & 94.46 & 27.95 & 94.17 \\ 
			$\tilde{\bm{\rho}}$ & 95.37 & 95.78 & 94.72 & 94.32 & 95.78 & 94.38 & 95.72 & 95.16 & 95.88 \\
			\bottomrule
		\end{tabular}
		\caption{Coverage of the $ 95\% $ expected variation intervals of the SBC obtained posterior samples of the CAR-MM parameters for each scenario and membership size}
		\label{tab:calibration_res}
	\end{table}

	In Table \ref{tab:calibration_res} we report the coverage of the $ 95\% $ intervals of the expected variation of the SBC rank statistics against the uniform distribution \footnote{The plots for all parameters can be found in Section 2.2 in the Supplemental material}. For the $ m = 100, 130 $ cases, we observe high coverage for all the parameters in all scenarios, with the exception of the $ \alpha $ parameter, which is generally not of inferential interest. We believe that the relatively lower coverage of $ \beta_1 $ in the $ m = 130 $ scenario is not a cause of concern for two reasons: first, as will be seen below, posterior means are unbiased across the simulation study, and second, this fact does not seem to hinder the good calibration of the relative risks $ \brho $ and $ \tilde{\brho} $.
	
	Overall, in the $ m = 70 $ case, we have satisfactory coverage for the linear term parameters $ \gamma $, $ \beta_1 $ and $ \beta_2 $ in all scenarios. Despite this, we observe poor calibration for the spatial random effects $ \bphi $ and the areal relative risks $ \brho $, the only exception being the \textit{post-post} case. These results are consistent with the considerations on identifiability that we discussed in Section \ref{sec:rel_car_carmm} where we remarked the identifiability only of the $ m = n $ case under both parameterisation and the $ m > n $ for the \textit{post} one. Although theoretically the $ m=70 $ should not be identifiable, the calibration in the \textit{post}-\textit{post} case could be due to the coincidence between data generation and the MCMC parameterisation. Conversely, in the \textit{inverse}-\textit{inverse} case, despite the two parameterisations being the same, we remark that the generalised inverse used to obtain areal relative risks ($ \brho $) from the sampled membership risks ($ \tilde{\brho} $) is not a right inverse (Lemma \ref{lemma:L_unique}). This leads to the poor calibration of the $ \brho $

	All in all, these simulations suggest that in the \textit{post} parameterisation case the posterior samples are well calibrated and support the considerations on identifiability of the CAR-MM model discussed above. Overall these results lead to us recommend using the \textit{post} parameterisations in all applications, not only due to it computational efficiency as mentioned in Section \ref{subsec:implementation}, but also due to its overall better calibration.
	
	However, despite the positive results for the $ m < n $ case, we remind the reader that in Section \ref{subsec:ident} we showed its non identifiability in the posterior, under both parameterisations. The good calibration results in the \textit{post}-\textit{post} (both data and MCMC sampler are \textit{post} parameterised) case could simply an artefact of the simulation setup. In fact, in that scenario the function that maps the areal relative risks to the membership risks, i.e. the MM matrix $ \Hb $ is the same both when simulating the data and sampling from the posterior distribution. Conversely, whenever an inverse parameterisation is involved, either in the data simulation stage or during MCMC sampling, given that a left inverse does not exist (Lemma \ref{lemma:L_unique}), the data does not contain enough information to correctly recover the spatial random effects at the areal level. Finally, high calibration for all parameters across different parameterisations for the two $ m\geq n $ cases is evidence of the reliability of the CAR-MM model.

	\section{Data analysis}
	\label{sec:data_an}
	\subsection{Data structure and modelling}
	We present an application of this methodology to primary care catchment areas. As argued above, the CAR-MM approach is of utility in situations where data are observed for non-spatial units (memberships) but where an underlying spatial framework can be conceived, and where an area to membership population cross-referencing is available \parencite{patients2019}. In particular, UK primary care data on disease prevalence (which is recorded at GP practice level) offers a suitable application. Specifically, we consider diabetes prevalence for a set of GP practices in South East London (SEL), and to the MSOAs (small areas) defining the catchment areas of these GP practices.

	The data pertain to 153 GP practices which together constitute what are known as Clinical Commissioning Groups or CCGs: the 153 practices are those for Bexley, Lewisham, Greenwich and Bromley CCGs. Based on the choice of GPs, a total of 152 MSOAs were included so that at least 90\% coverage of each GPs population was reached.
	
	Prevalence data for GPs refers to October 2015, as does the cross reference information providing the number of each MSOA residents registered at each GP practice \parencite{qof2015}. Age structure information was also available at GP level, so that expected rates could be easily computed using England and Wales national rates of diabetes prevalence.
	
	For MSOA level covariates, two were included: from the 2011 Census the share of South Asian residents $(\bm{X})_{\cdot1}$ i.e. Indian, Pakistani and Bangladeshi, see Table KS201EW in the NOMIS database \parencite{ons_census}, and the 2015 Index of Multiple deprivation (IMD) $(\bm{X})_{\cdot2}$. The two variables were included as a control of socio-economic status and due to the higher prevalence of diabetes among the South Asian population \parencite{Hall2010}. Both variables were min-max normalised. 
	
	\subsection{Results}
	
	Initial runs using the Poisson likelihood model (Eq. \ref{eq:glm_poi_mm}) returned concerning PSIS-LOO weights $ \hat{k} $ which might be a symptom of misspecification of the model. In fact, \citet{Vehtari2017} note that high $ \hat{k} $ mean that the full posterior distribution is a poor approximation of the leave-one-out posterior.
	
	Since the possible misspecification might be attributable to overdispersion, i.e. when the variance of the outcome is larger than its average \parencite{Coly2019} we therefore turned to a Negative Binomial likelihood for the model in equation \ref{eq:glm_poi_mm}, that allows for direct modelling of overdispersion.
	
	The density of the parameterisation we used is the following:	 
	\begin{align} \label{eq:dens_neg_bin}
		f_Y(y_j ; \mu_j, \psi) = 
		\begin{pmatrix}
			y_j + \psi - 1 \\
			y_j
		\end{pmatrix} 
		\left(\frac{\mu_j}{\mu_j + \psi}\right)^{y_j} \left(\frac{\psi}{\mu_j + \psi}\right)^\psi,
	\end{align}
	where $ y_j \in \mathbb{R} $ for all $ j = 1, ..., m $; $ E(Y_j) = \mu_j \in \mathbb{R}^+ $, $ \psi \in \mathbb{R}^+ $ and $ Var(Y_j) = \mu_j + \mu_j^2/\psi $. Therefore the mean of each $ Y_j $ can be modelled directly through the $ \mu_j $'s. Equation \ref{eq:glm_poi_mm} then becomes:
	\begin{align} \label{eq:glm_nb_mm}
		\tilde{Y}_j \sim  NegBin(\tilde{E}_j \tilde{\rho}_j, \psi) \quad , \quad \log(\tilde{\bm{\rho}}) = \log(\bm{\Hb\rho}) = \Hb \left(\gamma\one_n + \bm{X} \bm{\beta} + \bm{\phi}\right),
	\end{align}
	so the overdispersion parameter $ \psi $ is assumed common across all GP practices (i.e. memberships). The priors used were the same as for the simulation study in Section \ref{sec:sim_study}.
	
	In the analysis we compared three models: the CAR-MM as described above, the ICAR-MM and a GLM-MM that simply consists in removing the spatial random effects $ \bm{\phi} $ from equation \ref{eq:glm_nb_mm}. We include the ICAR since it is one of the most widely used prior in the disease mapping literature \parencite{MartinezBeneito2019}. All models were coded in RStan (see Section \ref{subsec:implementation} and \cite{Morris2019} for the ICAR prior) and run for five $ 10^5 $ iterations chains with 50\% warm-up. With regard to the prior for the added overdispersion parameter $ \psi $ for the Negative Binomial likelihood, $ Gamma(2, 0.2) $, parameterised in terms of shape and rate, was used.
	
	\begin{table}[ht]
		\centering
		\caption{Posterior results for SEL data analysis comparing Negative Binomial GLM-MM models (white columns), CAR-MM (light grey columns) and ICAR-MM (dark grey columns). $ \beta_1 $ is the coefficient, associated to the share of South Asian population, and $ \beta_2 $  for the IMD (both min-max normalised)}
		\label{tab:data_res}
		\centering
		\begin{adjustbox}{max width=\textwidth}
			\begin{tabular}{l|gai|gai|gai|gai|gai|gai}
				\toprule
				&\multicolumn{3}{c}{$\alpha$} & \multicolumn{3}{c}{$\tau$} & \multicolumn{3}{c}{$\gamma$} & \multicolumn{3}{c}{$\beta_{1}$} & \multicolumn{3}{c}{$\beta_{2}$} & \multicolumn{3}{c}{$\psi$} \\ 
				\midrule
				mean & - & 0.52 & - & - & 18.42 & 26.02 & -0.75 & -0.74 & -0.69 & 0.75 & 0.69 & 0.54 & 1.02 & 1.02 & 0.99 & 31.61 & 33.85 & 34.71 \\ 
				sd & - & 0.28 & - & - & 7.89 & 9.01 & 0.05 & 0.06 & 0.07 & 0.14 & 0.17 & 0.20 & 0.07 & 0.09 & 0.11 & 4.05 & 4.59 & 4.63 \\ 
				$2.5\%$ & - & 0.03 & - & - & 7.13 & 12.26 & -0.84 & -0.86 & -0.83 & 0.49 & 0.34 & 0.13 & 0.88 & 0.85 & 0.78 & 24.35 & 25.73 & 26.35 \\ 
				$50\%$ & - & 0.54 & - & - & 17.04 & 24.71 & -0.75 & -0.74 & -0.69 & 0.75 & 0.69 & 0.54 & 1.02 & 1.02 & 0.99 & 31.37 & 33.55 & 34.46 \\ 
				$97.5\%$ & - & 0.96 & - & - & 37.64 & 47.06 & -0.66 & -0.62 & -0.56 & 1.02 & 1.02 & 0.94 & 1.16 & 1.19 & 1.19 & 40.20 & 43.68 & 44.58 \\ 
				width CI & - & 0.93 & - & - & 30.51 & 34.80 & 0.18 & 0.24 & 0.28 & 0.53 & 0.68 & 0.80 & 0.28 & 0.34 & 0.41 & 15.85 & 17.94 & 18.23 \\ 
				$\hat{R}$ & - & 1.00 & - & - & 1.00 & 1.00 & 1.00 & 1.00 & 1.00 & 1.00 & 1.00 & 1.00 & 1.00 & 1.00 & 1.00 & 1.00 & 1.00 & 1.00 \\  
				
				ESS  & - & 3414 & - & - & 14048 & 13034 & 10821 & 7586 & 5325 & 12741 & 9850 & 6496 & 12570 & 14058 & 5937  & 14984 & 27963 & 28552 \\ 
				\bottomrule
			\end{tabular}
		\end{adjustbox}
	\end{table}

	Table \ref{tab:data_res} shows posterior estimates for the three models estimated on the SEL dataset. In general estimates are very similar across the three models, with the exceptions of marginal CAR precision $ \tau $ and $ \beta_1 $ coefficient, the former larger and the latter smaller in the ICAR-MM. As expected, given the fact that no term is directly accounting for residual spatial autocorrelation, the GLM-MM model estimates slightly smaller credible intervals for $ \beta_1 $ and $ \beta_2 $ than the two CAR based models.
	
	Model comparisons in Table \ref{tab:data_post_pred} do not seem to show substantial differences between model fits: the ICAR shows the highest $ \widehat{elpd}_{loo} $ although with no significant difference from the other models. Despite the GLM-MM having the lowest $ \overline{RPS} $ and $ \overline{DSS} $, the small differences and the result from other model comparison criteria, there is not sufficient evidence to claim it provides a better fit. In fact, in Figure 10 (Supplemental material), we can see no discernible difference between posterior calibration of marginal PPP's. Therefore, we focus on the CAR-MM model as it is the one on which we focussed our attention in this paper.

	\begin{table}[h!]
		\caption{Comparison of model check measures. ICAR-MM model is used as reference (model A) when computing differences  $ \text{diff}(\widehat{elpd}_{loo}^{A, B}) $ }
		\label{tab:data_post_pred}
		\begin{tabular}{lcccccc}
			\toprule
			& $ \widehat{elpd} $ & $ se(\widehat{elpd}_{loo})  $  &
			$ \text{diff}(\widehat{elpd}_{loo}^{A, B}) $ & 
			$ se\left[\text{diff}(\widehat{elpd}_{loo}^{A,B})\right] $ & $ \overline{RPS} $ & $ \overline{DSS} $ \\
			\midrule
			ICAR-MM & -847.8 & 14.2 & 0 & 0 & 35.93 & 9.39 \\
			CAR-MM & -849.0 & 13.8 &  -1.2  & 1.1 & 35.11 &  9.32 \\
			GLM-MM & -849.1 & 13.0 &  -1.3   &  2.5   & 34.86 & 9.28 \\
			\bottomrule
		\end{tabular}
	\end{table}
	
	\begin{figure}[h!]
		\centering
		\caption{(a) Posterior means for each areal relative risk ($ \bm{\rho} $ in equation \ref{eq:glm_nb_mm}) for the CAR-MM model (b) Probabilities $ P(\rho_i > 1| \yb) $ that the relative risk in each exceeds 1 among posterior samples under the CAR-MM prior
		}
		
		\subfloat[][]{\includegraphics[width=.5\linewidth]{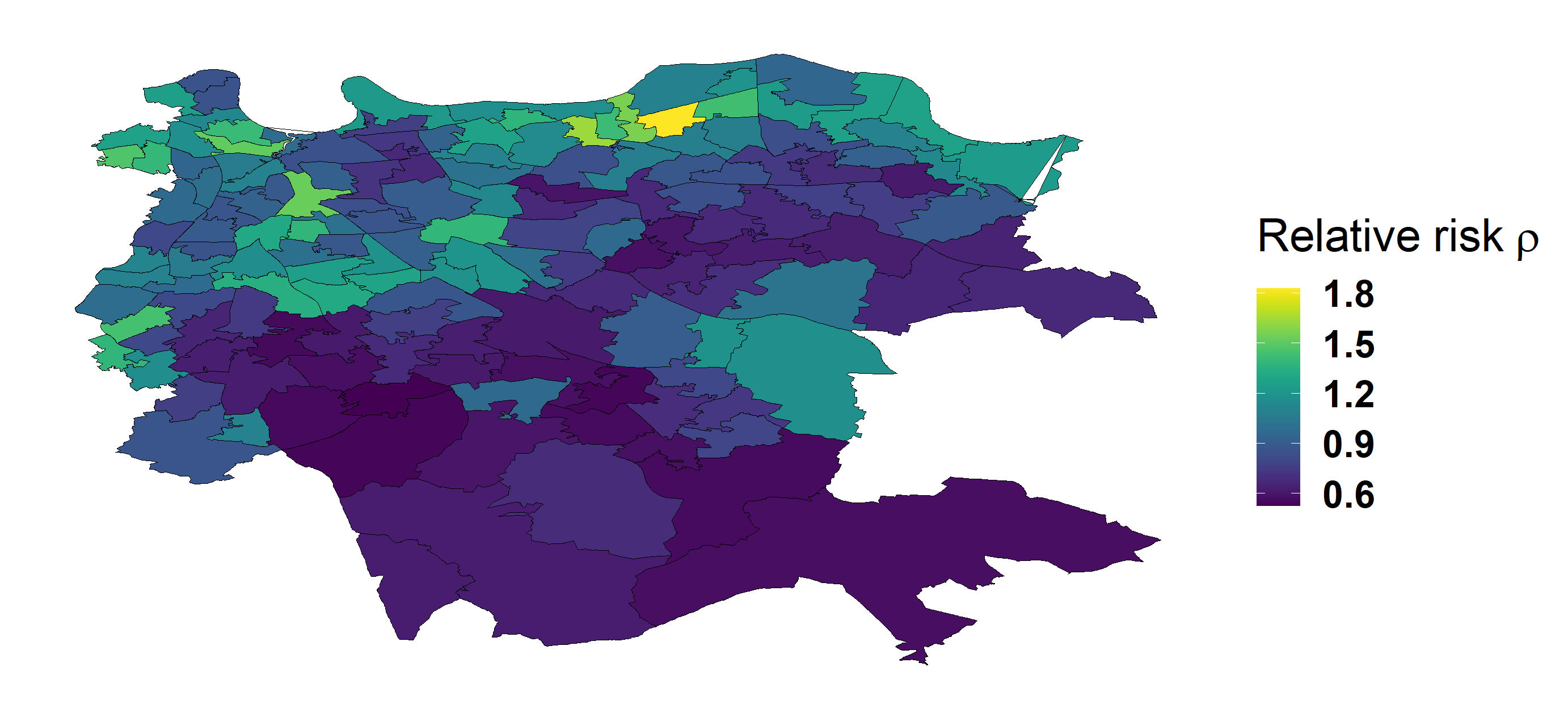}   \label{fig:rr_map_sel_car_nb}}
		\subfloat[][]{\includegraphics[width=.5\linewidth]{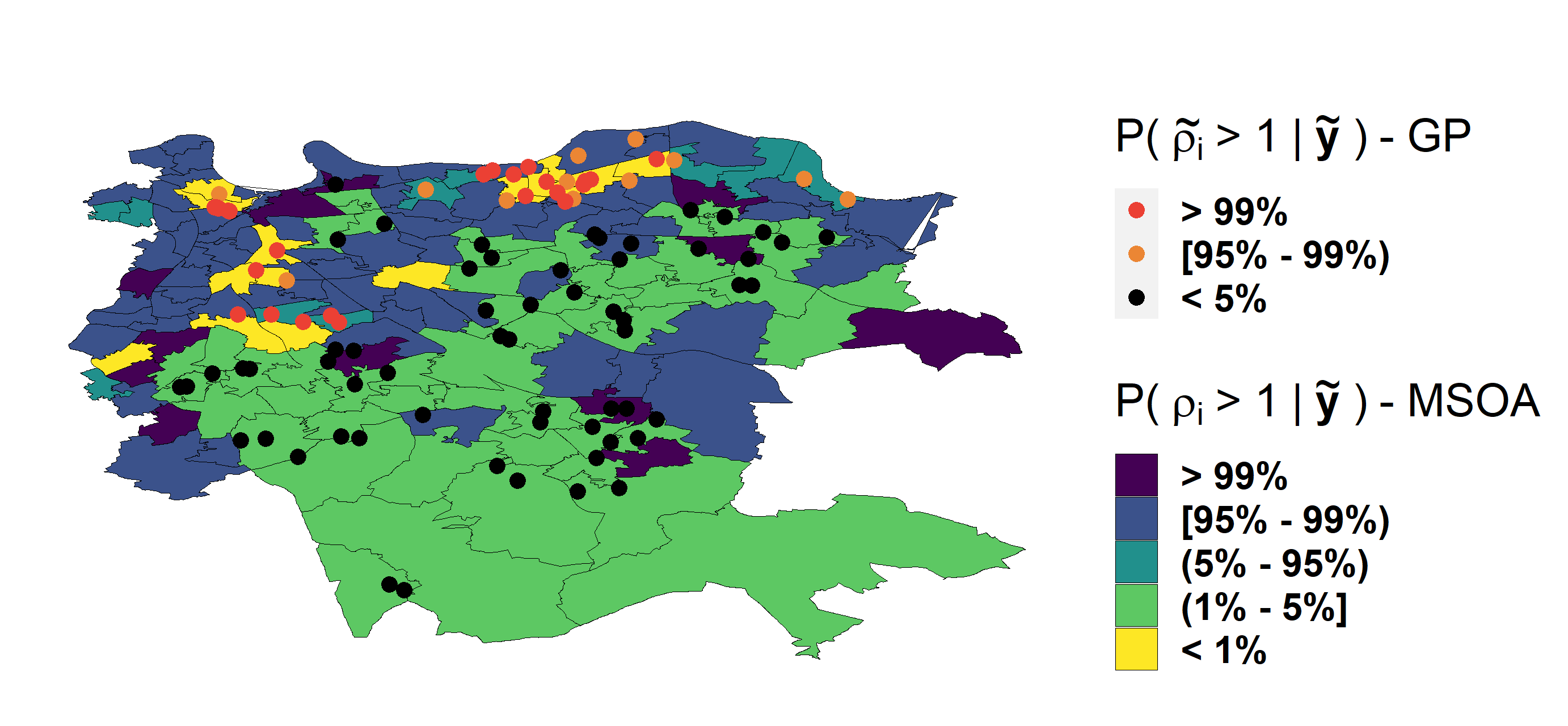}   \label{fig:rr_over_car}}
	\end{figure}
	
	\subsection{Strategic implications}
	
	Comparison of posterior relative risk estimates between the 152 MSOAs in the catchment area shows a wide contrast in estimated prevalence and hence in the need for diabetes care. Such contrasts show substantial ecological variation – that is, variation due to area context. Such variations support locality based or place-based interventions and  allocation of resources for tackling diabetes \parencite{nice_2014}.	
	
	Considering the CAR-MM model, the ratio of the $ 90^{th} $ percentile relative risk to the $ 10^{th} $ percentile relative risk has posterior mean ($ 95\% $ Credible Interval) of $ 2.31 $ $  (2.06, 2.55) $, an over 2-fold variation in risk.
	
	Figure \ref{fig:rr_over_car} maps out the probabilities $ P(\rho_i > 1| \tilde{\yb}) $ that the relative risk exceeds 1 under the CAR-MM in the posterior samples \parencite{Richardson2004}. There are $ 24 $ areas with over $ 95\% $ probability of elevated risk, indicating highly elevated diabetes prevalence, and $ 14 $ areas with this probability exceeding $ 99\% $. By contrast, a considerably larger minority of $ 61 $ areas have under $ 5 \% $ probability of elevated risk, with $ 49 $ areas having exceedance probabilities below $ 0.01 $, indicating significantly low prevalence. 
	
	Figures \ref{fig:rr_map_sel_car_nb} and \ref{fig:rr_over_car} also show the spatial features of elevated diabetes levels in a minority of areas, particularly a spatial concentration in the north of the region, reinforcing findings such as \parencite{Green2003}. The posterior mean ($ 95\% $ CI) of the Moran’s I for the relative risks is $ 0.44 $ $ (0.35, 0.49) $. The clustering of high risk reflects residential concentration in observed area risk factors, namely deprived communities and in the South Asian community, as well as spatial clustering in unobserved influences on diabetes risk. Hence both area risk factors should figure in any area rankings for locality based intervention.

	Table \ref{tab:risk_profiles} disaggregates the 152 MSOAs into quintiles according to their deprivation scores and their percentages South Asian. If we consider the mean of posterior relative risks in the 30 MSOAs which are most deprived we see an average ($ 95\% $ CI) of $ 1.32 $ $ (1.15, 1.55) $, as compared to a much lower relative risk in the least deprived MSOAs. The posterior relative risk in the 30 MSOAs with the relatively largest Asian communities has mean ($ 95\% $ CI) $ 1.07 $ $ (0.64, 1.68) $ as compared to significantly lower relative risks in MSOAs with relatively small Asian communities. Such findings add to existing UK-based findings that area deprivation and ethnic mix significantly affects area diabetes prevalence and incidence variations \parencite{Moody2016, Nishino2015}.
	
	\begin{table}[h!]
		\begin{tabular}{lcccccc}
			\toprule
			& \multicolumn{3}{c}{\textbf{IMD}} & \multicolumn{3}{c}{\textbf{South Asian population}}\\
			\midrule
			& mean & $ 2.5 \% $ & $ 97.5 \% $ & mean & $ 2.5\% $ & $ 97.5 \% $\\
			\hline
			$ 1^{st} $ & 0.62&	0.54&	0.79	&0.92&	0.54&	1.39\\
			$ 2^{nd} $ & 0.74&	0.62&	0.89	&0.91&	0.56&	1.38\\
			$ 3^{rd} $ & 0.90&	0.75&	1.10&	0.87&	0.55&	1.45\\
			$ 4^{th} $ & 1.13&	0.96&	1.64&	0.93&	0.58&	1.45\\
			$ 5^{th} $ & 1.32&	1.15&	1.55&	1.07&	0.64&	1.68\\
			\bottomrule
		\end{tabular}
		\caption{Relative risk profiles of posterior means by covariate quintiles}
		\label{tab:risk_profiles}
	\end{table}
	
	Finally, so far we analysed the areal relative risks $ \bm{\rho} $ however, as mentioned in Section \ref{sec:rel_car_carmm}, the membership relative risks $ \tilde{\bm{\rho}} $ are also identifiable in the posterior (since $ m > n $ in our case) and can be of inferential interest in many applications. In fact, in our case each $ \tilde{\rho_j} $ corresponds to the relative risk of a single GP practice and can be of significant epidemiological and public health interest.
	
	Therefore, we can replicate the analysis of the probabilities of a relative risk exceeding 1 as we did above. In the case of GP practices, we can see that: $ 41 $ practices have over $ 95\% $  probability, $ 29 $ with over $ 99\% $, and $ 68 $ are under $ 5\% $, with $ 62 $ under $ 1\% $. In Figure \ref{fig:rr_over_car} we also show the locations of those practices on the map.

	As expected, in the SEL region, these results are similar to the ones observed for areal relative risks, given that we have a similar number of practices and areas. However, these membership level relative risk can give insight in which GP's could be facing significant public health challenges. In fact, although most of the practices with either high or low probability are in accordance located in areas with similar risk profiles, there are some GP's for which the opposite is true, i.e. high or low risk practices in respectively low or high risk areas. A purely areal analysis of prevalence would have not been able to identify these practices.

	\section{Discussion}
	\label{sec:disc}
	In this work we have studied and evaluated the theoretical underpinnings of the CAR-MM model devising a general framework to evaluate its structure and implications to spatial and relative risk modelling. We compared our approach to that of \citet{Petrof2020}, provided a more general framework able to encompass a potentially wide range of applications, and showed how this could apply to the parameterisation and evaluation of a CAR-MM model. Building on that, we considered identifiability in both the likelihood and the posterior for different membership sizes compared to areas and both parameterisations.
	
	The simulation study highlights how this approach leads to well calibrated posterior samples and small bias for estimated parameters. We showed how the two parameterisations, although theoretically identical in the $ m =n $ case, can lead to significantly different results in terms of both computational time and identifiability when $ m \neq n $. More specifically, we recommend the use of the \textit{post} parameterisation in all cases although we warn of the possible identifiability problems in the $ m < n $ case. In this regard, the present paper supports the previous analyses and findings of both \citet{Gramatica2021} and \citet{Petrof2020}, where $ m > n $ datasets were considered, and where the models were \textit{post} parameterised.

	Finally, the data analysis allowed us to show how this approach can be used in a real world application where the MM principle and CAR frameworks merge naturally, thus allowing for a direct implementation of the CAR-MM prior. Datasets from different institutions, in this case the National Health Service (NHS) and the Census, can be directly integrated leading to a modelling strategy that not only is able to pinpoint areas with increased relative risks, but also specific practices, by making use of the membership level modelling. These evaluations can be of interest for both epidemiological research and public health considerations.

	\printbibliography
	
\end{document}


\maketitle

	\section{Proofs of main results}
	\label{sec:proofs}
	
	\begin{corollary} \label{corollary:covariance}
		Let $\Sib, \, \Hb, \, \tSig $  be as defined in \textbf{M1}-\textbf{M3}, then iff $ m \leq n $, $ \tSig $ can be expressed as the covariance matrix of a CAR model $ (\Ib - \tilde{\Cb})^{-1}\tilde{\Mb} $ for a unique pair of $ m \times m $ matrices $ \tilde{\Cb} $ and $ \tilde{\Mb} $.
	\end{corollary}
	
	\begin{proof}
		To prove the sufficiency part, since Theorem 2.1 states that the statement is true for any positive-definite covariance matrix, we only need to prove that $ \tSig $ is in fact positive definite. This can be achieved by simply noting that: $ \tSig = \Hb \Sib \Hb^T = \Hb^T \Sib \Hb $ and both $ \Sib $ and $ \Hb $ are of full rank, so by applying (3) in \textit{Theorem 14.2.9} by \citet{harville2008matrix} we obtain the positive definiteness of $ \tSig $.
		
		To prove necessity, we observe that when $ m > n $  $rank(\tSig) = n $  making it rank deficient, thus we conclude that it cannot be also positive definite.
	\end{proof}
	
	\begin{lemma} \label{lemma:L_unique}
		Let $ \Hb $ be a matrix as defined in \textbf{M1}-\textbf{M3}. $ \Hb $ has a left inverse $ \Lb $ iff $ m\geq n $, and this matrix is unique iff $ m = n $.
	\end{lemma}
	
	\begin{proof}
		In the case $ m = n $ it is obvious that $ \Lb $ exists, given that $ \Hb $ is of full rank by \textbf{M3}, and in fact coincides with $ \Hb^{-1} $ guaranteeing uniqueness.
		Since an $ m \times n $ matrix has a left inverse iff it is of full column rank, it follows that when $ m < n $, $ \Hb $ will not have any such inverse. 
		
		Moving on to the $ m > n $ case, following Theorem 9.2.7 in \citet{harville2008matrix}: for a generic $ m \times n $ matrix $ \Hb $ and any of its generalised inverses $ \Gb $, an $ n \times m $ matrix $ \Gb^* $ is also a generalised inverse of $ \Hb $ iff
		\begin{align} \label{eq:gen_inv}
			\Gb^* = \Gb + \Zb - \Gb \Hb \Zb \Hb \Gb,
		\end{align}
		for some $ n\times m $ matrix $ \Zb $. We can verify that in our case this is true for \textit{any} $ \Zb $. If we take $ \Gb = \left(\Hb^T \Hb\right)^{-1}\Hb^T $ as a generalised inverse: first, we note that $ \Gb $ is evidently a left inverse, then we can see that so is $ \Gb^* $, in fact for any $ \Zb $ we have:
		\begin{align}
			\Gb^*\Hb = \Gb\Hb + \Zb\Hb + \Gb\Hb\Zb\Hb\Gb = \Ib_n + \Zb\Hb - \Zb\Hb = \Ib_n,
		\end{align}
		thus proving non uniqueness of $ \Hb $'s left inverse.
	\end{proof}
	
	\begin{lemma} \label{lemma:L_constraint}
		Any left inverse $ \Lb $ as in Lemma \ref{lemma:L_unique} will satisfy the same constraint \textbf{M2-1} as $ \Hb $, i.e the columns adding up to the $ \one $ vector
	\end{lemma}
	\begin{proof}
		$ \Hb\one_n = \one_m \, \Leftrightarrow \, \Lb\Hb\one_n=\Lb\one_m \, \Leftrightarrow \, \one_n = \Lb\one_m $.
	\end{proof}
	
	\begin{thm} \label{thm:unique_Sig}
		Let $\Sib, \, \Hb, \, \tSig $  be as defined in \textbf{M1}-\textbf{M3}, then given $ \tSig $ and $ \Hb $, $ \Sib $ is uniquely defined iff $ m \geq n $.
	\end{thm}
	\begin{proof}
		In the case $ m \geq n $ it is evident that there is only one $ \Sib $ for a given $ \Hb $ and $\tSig $. This can be proved by contradiction because, if we assume that there was not, it would mean that there existed a covariance matrix $ \Sib^* \neq \Sib $ such that
		\begin{align} \label{eq:in_theo3}
			\Hb \Sib^* \Hb^T = \Hb \Sib \Hb^T,
		\end{align}
		but, by Lemma \ref{lemma:L_unique}, we have that a left inverse exists thus leading to $ \Sib^* = \Sib $ which contradicts the initial assumption.
		
		For $ m < n $ it suffices to look at a simple case $ m = 2 $ and $ n = 3 $ case. Assuming 
		$ \Hb = \begin{pmatrix} 0.5 & 0 & 0.5 \\ 0 & 1 & 0 \end{pmatrix}$, then in order to find a $ 3 \times 3 $ covariance matrix $ \Sib $ for a fixed $ 2 \times 2 $ $ \tSig $, we would need to solve a 4 equation linear systems with 9 unknowns, which is clearly underdetermined, thus making $ \Sib $ not uniquely defined.
	\end{proof}

	\section{SBC results}
	\label{sec:res_calibration}
	
	\subsection{Convergence and computational time}
	
	\begin{table}[ht]
		\centering
		\begin{tabular}{lrrrrrrrrr}
			\toprule
			&\multicolumn{5}{c}{\textbf{Data: Post}} & \multicolumn{4}{c}{\textbf{Data: Inverse}} \\
			\hline
			&\multicolumn{3}{c}{\textbf{MCMC: Post}} & \multicolumn{2}{c}{\textbf{MCMC: Inverse}} &
			\multicolumn{2}{c}{\textbf{MCMC: Post}} & \multicolumn{2}{c}{\textbf{MCMC: Inverse}} \\
			\hline
			&$70$  & $100$  & $130$ & $70$& $100$& $70$  & $100$  & $70$ & $100$  \\ 
			\midrule
			mean & 2.22 & 2.32 & 2.99 & 11.30 & 29.20 & 2.31 & 2.34 & 11.30 & 28.84 \\ 
			sd & 4.28 & 0.57 & 0.81 & 2.52 & 5.18 & 4.77 & 0.58 & 2.53 & 5.31 \\ 
			median & 1.84 & 2.27 & 2.83 & 11.26 & 28.84 & 1.85 & 2.28 & 11.26 & 28.60 \\ 
			$2.5\%$ & 1.20 & 1.44 & 1.79 & 7.02 & 19.42 & 1.21 & 1.45 & 7.03 & 18.79 \\ 
			$97.5\%$ & 3.07 & 3.58 & 4.71 & 17.89 & 41.06 & 3.10 & 3.60 & 17.88 & 40.88 \\ 
			\bottomrule
		\end{tabular}
		\caption{Summary statistics for the posterior sampler time (in minutes) for each scenario and membership size}
		\label{tab:sbc_time}
	\end{table}
	
	\begin{table}[ht]
		\centering
		\begin{tabular}{rrrrrrrrr}
			\toprule
			\multicolumn{5}{c}{\textbf{Data: Post}} & \multicolumn{4}{c}{\textbf{Data: Inverse}}\\
			\hline
			\multicolumn{3}{c}{\textbf{MCMC: Post}} & \multicolumn{2}{c}{\textbf{MCMC: Inverse}} &
			\multicolumn{2}{c}{\textbf{MCMC: Post}} & \multicolumn{2}{c}{\textbf{MCMC: Inverse}} \\
			\hline
			$70$  & $100$  & $130$ & $70$& $100$& $70$  & $100$  & $70$ & $100$  \\ 
			\midrule
			168 & 150 & 183 & 18 & 172 & 86 & 160 & 15 & 165 \\
			1.68 \% & 1.50 \% & 1.83 \% &0.18 \% &1.72 \% &0.86 \% &1.60 \% &0.15 \% & 1.65\% \\
			\bottomrule
		\end{tabular}
		\caption{Number of removed simulations due to parameters reporting a $ \hat{R} > 1.01 $, from the simulation study in Section 6}
		\label{tab:sim_removed}
	\end{table}
	
	\clearpage
	\subsection{Bias, absolute bias and RMSE}
	\label{subsec:res_bias}
	\vspace{-.4cm}
	
	\begin{table}[h!]
		\centering
		\resizebox{\columnwidth}{!}{%
			\begin{tabular}{lr}
				\begin{tabular}{c}
					\multirow{7}{*}{mean} \\ \\  \\   \\   \\   \\   \\  
					\multirow{10}{*}{sd} \\ \\  \\   \\   \\   \\   \\  
					\multirow{11}{*}{median} \\ \\  \\   \\   \\   \\   \\  
					\multirow{14}{*}{$ 2.5\% $} \\ \\  \\   \\   \\   \\   \\  
					\multirow{17}{*}{$ 97.5\% $} \\ \\  \\   \\   \\   \\   \\  
				\end{tabular} &
				\rowcolors{2}{Lightgray}{white}
				\begin{tabular}{lrrrrrrrrr}
					\toprule
					\multicolumn{10}{c}{\textbf{Bias}}\\
					\hline
					& \multicolumn{5}{c}{\textbf{Data: Post}} & \multicolumn{4}{c}{\textbf{Data: Inverse}}\\
					& \multicolumn{3}{c}{\textbf{MCMC: Post}} & \multicolumn{2}{c}{\textbf{MCMC: Inverse}} &
					\multicolumn{2}{c}{\textbf{MCMC: Post}} & \multicolumn{2}{c}{\textbf{MCMC: Inverse}} \\
					\hline
					& $70$  & $100$  & $130$ & $70$& $100$& $70$  & $100$  & $70$ & $100$  \\ 
					\midrule
					$\alpha$ & 0.00 & 0.00 & 0.00 & 0.01 & 0.00 & 0.00 & 0.00 & 0.01 & 0.00 \\ 
					$\tau$ & 0.15 & 0.14 & 0.13 & -0.04 & 0.14 & 0.15 & 0.13 & -0.04 & 0.14 \\ 
					$\gamma$ & -0.00 & -0.00 & -0.00 & -0.00 & -0.00 & -0.00 & -0.00 & -0.00 & -0.00 \\ 
					$\beta_1$ & 0.00 & 0.00 & 0.00 & 0.00 & 0.00 & 0.00 & 0.00 & 0.00 & 0.00 \\ 
					$\beta_2$ & 0.00 & -0.00 & -0.00 & 0.00 & -0.00 & 0.00 & -0.00 & 0.00 & -0.00 \\ 
					$\bm{\phi}$ & 0.00 & 0.00 & 0.00 & 0.00 & 0.00 & 0.00 & 0.00 & -0.00 & 0.00 \\ 
					$\bm{\rho}$ & -0.00 & -0.00 & -0.00 & -0.04 & -0.00 & -0.00 & -0.00 & -0.04 & -0.00 \\ 
					$\tilde{\bm{\rho}}$ & -0.00 & -0.00 & -0.00 & -0.00 & -0.00 & -0.00 & -0.00 & -0.00 & -0.00 \\ 
					\midrule
					$\alpha$ & 0.26 & 0.25 & 0.24 & 0.26 & 0.25 & 0.26 & 0.25 & 0.26 & 0.25 \\ 
					$\tau$ & 5.54 & 5.37 & 5.16 & 5.54 & 5.38 & 5.54 & 5.36 & 5.54 & 5.38 \\ 
					$\gamma$ & 0.21 & 0.19 & 0.18 & 0.21 & 0.19 & 0.21 & 0.19 & 0.21 & 0.19 \\ 
					$\beta_1$ & 0.29 & 0.27 & 0.25 & 0.29 & 0.27 & 0.29 & 0.27 & 0.29 & 0.27 \\ 
					$\beta_2$ & 0.33 & 0.30 & 0.28 & 0.33 & 0.30 & 0.33 & 0.30 & 0.33 & 0.30 \\ 
					$\bm{\phi}$ & 0.22 & 0.20 & 0.20 & 0.22 & 0.20 & 0.22 & 0.20 & 0.22 & 0.20 \\ 
					$\bm{\rho}$ & 0.49 & 0.37 & 0.34 & 1.01 & 0.35 & 0.49 & 0.37 & 1.02 & 0.35 \\ 
					$\tilde{\bm{\rho}}$ & 0.17 & 0.16 & 0.15 & 0.17 & 0.16 & 0.17 & 0.16 & 0.17 & 0.16 \\ 
					\midrule
					$\alpha$ & -0.00 & -0.00 & -0.00 & 0.00 & -0.00 & -0.00 & -0.00 & 0.00 & -0.00 \\ 
					$\tau$ & 0.84 & 0.71 & 0.63 & 0.61 & 0.73 & 0.84 & 0.71 & 0.61 & 0.74 \\ 
					$\gamma$ & -0.00 & -0.00 & -0.00 & -0.00 & -0.00 & -0.00 & -0.00 & -0.00 & -0.00 \\ 
					$\beta_1$ & 0.00 & 0.00 & 0.00 & 0.00 & 0.00 & 0.00 & 0.00 & 0.00 & 0.00 \\ 
					$\beta_2$ & 0.00 & 0.00 & 0.00 & 0.00 & 0.01 & 0.00 & 0.00 & 0.00 & 0.00 \\ 
					$\bm{\phi}$ & -0.00 & -0.00 & -0.00 & -0.00 & -0.00 & -0.00 & -0.00 & -0.00 & -0.00 \\ 
					$\bm{\rho}$ & 0.01 & 0.01 & 0.01 & 0.01 & 0.01 & 0.01 & 0.01 & 0.01 & 0.01 \\ 
					$\tilde{\bm{\rho}}$ & 0.00 & 0.00 & 0.00 & 0.00 & 0.00 & 0.00 & 0.00 & 0.00 & 0.00 \\ 
					\midrule
					$\alpha$ & -0.46 & -0.45 & -0.45 & -0.44 & -0.45 & -0.46 & -0.45 & -0.44 & -0.45 \\ 
					$\tau$ & -14.32 & -13.96 & -13.24 & -14.42 & -13.96 & -14.31 & -13.94 & -14.42 & -13.95 \\ 
					$\gamma$ & -0.41 & -0.37 & -0.35 & -0.42 & -0.37 & -0.41 & -0.37 & -0.41 & -0.37 \\ 
					$\beta_1$ & -0.58 & -0.54 & -0.51 & -0.58 & -0.54 & -0.58 & -0.54 & -0.58 & -0.54 \\ 
					$\beta_2$ & -0.67 & -0.60 & -0.57 & -0.67 & -0.60 & -0.67 & -0.60 & -0.67 & -0.60 \\ 
					$\bm{\phi}$ & -0.43 & -0.41 & -0.39 & -0.44 & -0.40 & -0.43 & -0.40 & -0.44 & -0.40 \\ 
					$\bm{\rho}$ & -0.80 & -0.73 & -0.69 & -1.68 & -0.72 & -0.81 & -0.73 & -1.68 & -0.72 \\ 
					$\tilde{\bm{\rho}}$ & -0.36 & -0.34 & -0.32 & -0.36 & -0.34 & -0.36 & -0.34 & -0.36 & -0.34 \\ 
					\midrule
					$\alpha$ & 0.45 & 0.45 & 0.44 & 0.46 & 0.45 & 0.45 & 0.45 & 0.46 & 0.45 \\ 
					$\tau$ & 8.44 & 8.39 & 8.26 & 8.36 & 8.39 & 8.43 & 8.39 & 8.36 & 8.39 \\ 
					$\gamma$ & 0.42 & 0.38 & 0.36 & 0.42 & 0.38 & 0.42 & 0.38 & 0.42 & 0.38 \\ 
					$\beta_1$ & 0.59 & 0.54 & 0.51 & 0.59 & 0.54 & 0.59 & 0.54 & 0.59 & 0.54 \\ 
					$\beta_2$ & 0.66 & 0.58 & 0.56 & 0.66 & 0.58 & 0.65 & 0.58 & 0.66 & 0.58 \\ 
					$\bm{\phi}$ & 0.44 & 0.41 & 0.40 & 0.44 & 0.41 & 0.44 & 0.41 & 0.44 & 0.41 \\ 
					$\bm{\rho}$ & 0.71 & 0.63 & 0.61 & 1.02 & 0.62 & 0.71 & 0.64 & 1.03 & 0.62 \\ 
					$\tilde{\bm{\rho}}$ & 0.33 & 0.32 & 0.29 & 0.33 & 0.31 & 0.33 & 0.32 & 0.33 & 0.31 \\ 
					\bottomrule
				\end{tabular}
			\end{tabular}
		}
		\caption{Descriptive statistics of the bias (Equation 6.3) across all simulated datasets for all parameters in Equation 3.2. Individual values of vector parameters, i.e. ($ \bm{\phi}, \bm{\rho}, \bm{\tilde{\rho}} $) are pooled together and treated as belonging to a single parameter}
		\label{tab:bias_relative_tot} 
	\end{table}

	\begin{table}[ht]
		\centering
		\begin{tabular}{lr}
			\begin{tabular}{c}
				\multirow{7}{*}{mean} \\ \\  \\   \\   \\   \\   \\  
				\multirow{10}{*}{sd} \\ \\  \\   \\   \\   \\   \\  
				\multirow{11}{*}{median} \\ \\  \\   \\   \\   \\   \\  
				\multirow{14}{*}{$ 2.5\% $} \\ \\  \\   \\   \\   \\   \\  
				\multirow{17}{*}{$ 97.5\% $} \\ \\  \\   \\   \\   \\   \\  
			\end{tabular} &
			\rowcolors{2}{Lightgray}{white}
			\begin{tabular}{lrrrrrrrrr}
				\toprule
				\multicolumn{10}{c}{\textbf{Absolute bias}}\\
				\hline
				& \multicolumn{5}{c}{\textbf{Data: Post}} & \multicolumn{4}{c}{\textbf{Data: Inverse}}\\
				& \multicolumn{3}{c}{\textbf{MCMC: Post}} & \multicolumn{2}{c}{\textbf{MCMC: Inverse}} &
				\multicolumn{2}{c}{\textbf{MCMC: Post}} & \multicolumn{2}{c}{\textbf{MCMC: Inverse}} \\
				\hline
				& $70$  & $100$  & $130$ & $70$& $100$& $70$  & $100$  & $70$ & $100$  \\ 
				\midrule
				$\alpha$ & 0.29 & 0.28 & 0.27 & 0.29 & 0.28 & 0.29 & 0.28 & 0.29 & 0.28 \\ 
				$\tau$ & 5.57 & 5.30 & 5.05 & 5.49 & 5.33 & 5.56 & 5.30 & 5.49 & 5.32 \\ 
				$\gamma$ & 0.23 & 0.20 & 0.20 & 0.23 & 0.20 & 0.23 & 0.20 & 0.23 & 0.20 \\ 
				$\beta_1$ & 0.32 & 0.29 & 0.28 & 0.32 & 0.29 & 0.32 & 0.29 & 0.32 & 0.29 \\ 
				$\beta_2$ & 0.37 & 0.33 & 0.31 & 0.37 & 0.33 & 0.37 & 0.33 & 0.37 & 0.33 \\ 
				$\bm{\phi}$ & 0.23 & 0.22 & 0.21 & 0.21 & 0.22 & 0.23 & 0.22 & 0.21 & 0.22 \\ 
				$\bm{\rho}$ & 0.32 & 0.29 & 0.27 & 0.42 & 0.28 & 0.32 & 0.29 & 0.43 & 0.28 \\ 
				$\tilde{\bm{\rho}}$ & 0.15 & 0.14 & 0.13 & 0.15 & 0.14 & 0.15 & 0.14 & 0.15 & 0.14 \\ 
				\midrule
				$\alpha$ & 0.09 & 0.09 & 0.09 & 0.09 & 0.09 & 0.09 & 0.09 & 0.09 & 0.09 \\ 
				$\tau$ & 3.51 & 3.51 & 3.45 & 3.57 & 3.50 & 3.52 & 3.51 & 3.57 & 3.50 \\ 
				$\gamma$ & 0.10 & 0.10 & 0.09 & 0.10 & 0.09 & 0.10 & 0.10 & 0.10 & 0.09 \\ 
				$\beta_1$ & 0.14 & 0.13 & 0.13 & 0.14 & 0.13 & 0.14 & 0.13 & 0.14 & 0.13 \\ 
				$\beta_2$ & 0.16 & 0.14 & 0.14 & 0.16 & 0.14 & 0.16 & 0.14 & 0.16 & 0.14 \\ 
				$\bm{\phi}$ & 0.13 & 0.12 & 0.11 & 0.14 & 0.11 & 0.13 & 0.12 & 0.14 & 0.11 \\ 
				$\bm{\rho}$ & 0.49 & 0.33 & 0.30 & 0.93 & 0.31 & 0.49 & 0.33 & 0.95 & 0.31 \\ 
				$\tilde{\bm{\rho}}$ & 0.13 & 0.13 & 0.12 & 0.13 & 0.12 & 0.13 & 0.13 & 0.13 & 0.12 \\ 
				\midrule
				$\alpha$ & 0.27 & 0.26 & 0.25 & 0.27 & 0.26 & 0.27 & 0.26 & 0.27 & 0.26 \\ 
				$\tau$ & 5.27 & 5.04 & 4.77 & 5.20 & 5.05 & 5.26 & 5.04 & 5.20 & 5.05 \\ 
				$\gamma$ & 0.20 & 0.18 & 0.17 & 0.20 & 0.18 & 0.20 & 0.18 & 0.20 & 0.18 \\ 
				$\beta_1$ & 0.28 & 0.26 & 0.25 & 0.28 & 0.26 & 0.28 & 0.26 & 0.28 & 0.26 \\ 
				$\beta_2$ & 0.33 & 0.30 & 0.27 & 0.33 & 0.30 & 0.33 & 0.30 & 0.33 & 0.30 \\ 
				$\bm{\phi}$ & 0.19 & 0.18 & 0.18 & 0.17 & 0.18 & 0.19 & 0.18 & 0.17 & 0.18 \\ 
				$\bm{\rho}$ & 0.21 & 0.20 & 0.19 & 0.22 & 0.20 & 0.21 & 0.20 & 0.22 & 0.20 \\ 
				$\tilde{\bm{\rho}}$ & 0.11 & 0.11 & 0.10 & 0.11 & 0.11 & 0.11 & 0.11 & 0.11 & 0.11 \\ 
				\midrule
				$\alpha$ & 0.13 & 0.11 & 0.11 & 0.13 & 0.12 & 0.11 & 0.11 & 0.13 & 0.12 \\ 
				$\tau$ & 0.54 & 0.43 & 0.40 & 0.51 & 0.46 & 0.53 & 0.43 & 0.51 & 0.47 \\ 
				$\gamma$ & 0.11 & 0.10 & 0.09 & 0.11 & 0.10 & 0.11 & 0.10 & 0.11 & 0.10 \\ 
				$\beta_1$ & 0.15 & 0.14 & 0.13 & 0.15 & 0.14 & 0.15 & 0.14 & 0.15 & 0.14 \\ 
				$\beta_2$ & 0.18 & 0.15 & 0.15 & 0.18 & 0.16 & 0.18 & 0.15 & 0.18 & 0.16 \\ 
				$\bm{\phi}$ & 0.11 & 0.10 & 0.10 & 0.05 & 0.10 & 0.11 & 0.10 & 0.05 & 0.10 \\ 
				$\bm{\rho}$ & 0.04 & 0.04 & 0.04 & 0.04 & 0.04 & 0.04 & 0.04 & 0.04 & 0.04 \\ 
				$\tilde{\bm{\rho}}$ & 0.03 & 0.03 & 0.02 & 0.03 & 0.03 & 0.03 & 0.03 & 0.03 & 0.03 \\ 
				\midrule
				$\alpha$ & 0.50 & 0.49 & 0.49 & 0.50 & 0.49 & 0.50 & 0.49 & 0.50 & 0.49 \\ 
				$\tau$ & 14.69 & 14.33 & 13.81 & 14.77 & 14.32 & 14.69 & 14.33 & 14.77 & 14.32 \\ 
				$\gamma$ & 0.50 & 0.45 & 0.44 & 0.50 & 0.45 & 0.50 & 0.45 & 0.50 & 0.45 \\ 
				$\beta_1$ & 0.70 & 0.65 & 0.61 & 0.70 & 0.65 & 0.70 & 0.65 & 0.70 & 0.65 \\ 
				$\beta_2$ & 0.79 & 0.71 & 0.67 & 0.79 & 0.72 & 0.79 & 0.72 & 0.79 & 0.72 \\ 
				$\bm{\phi}$ & 0.56 & 0.52 & 0.50 & 0.55 & 0.52 & 0.57 & 0.52 & 0.55 & 0.51 \\ 
				$\bm{\rho}$ & 1.24 & 1.06 & 1.01 & 2.19 & 1.04 & 1.25 & 1.07 & 2.20 & 1.04 \\ 
				$\tilde{\bm{\rho}}$ & 0.50 & 0.48 & 0.44 & 0.50 & 0.47 & 0.50 & 0.48 & 0.50 & 0.47 \\ 
				\bottomrule
			\end{tabular}
		\end{tabular}
		\caption{Descriptive statistics of the absolute bias (Equation 6.3) across all simulated datasets for all parameters in Equation 3.2. Individual values of vector parameters, i.e. ($ \bm{\phi}, \bm{\rho}, \bm{\tilde{\rho}} $) are pooled together and treated as belonging to a single parameter}
		\label{tab:abs_relative_tot} 
	\end{table}
	
	\begin{table}[ht]
		\centering
		\begin{tabular}{lr}
			\begin{tabular}{c}
				\multirow{7}{*}{mean} \\ \\  \\   \\   \\   \\   \\  
				\multirow{10}{*}{sd} \\ \\  \\   \\   \\   \\   \\  
				\multirow{11}{*}{median} \\ \\  \\   \\   \\   \\   \\  
				\multirow{14}{*}{$ 2.5\% $} \\ \\  \\   \\   \\   \\   \\  
				\multirow{17}{*}{$ 97.5\% $} \\ \\  \\   \\   \\   \\   \\    
			\end{tabular} &
			\rowcolors{2}{Lightgray}{white}
			\begin{tabular}{lrrrrrrrrr}
				\toprule
				\multicolumn{10}{c}{\textbf{RMSE}}\\
				\hline
				& \multicolumn{5}{c}{\textbf{Data: Post}} & \multicolumn{4}{c}{\textbf{Data: Inverse}}\\
				& \multicolumn{3}{c}{\textbf{MCMC: Post}} & \multicolumn{2}{c}{\textbf{MCMC: Inverse}} &
				\multicolumn{2}{c}{\textbf{MCMC: Post}} & \multicolumn{2}{c}{\textbf{MCMC: Inverse}} \\
				\hline
				& $70$  & $100$  & $130$ & $70$& $100$& $70$  & $100$  & $70$ & $100$  \\ 
				\midrule
				$\alpha$ & 0.35 & 0.34 & 0.33 & 0.36 & 0.34 & 0.35 & 0.34 & 0.35 & 0.34 \\ 
				$\tau$ & 7.05 & 6.69 & 6.36 & 6.92 & 6.73 & 7.04 & 6.69 & 6.92 & 6.72 \\ 
				$\gamma$ & 0.27 & 0.25 & 0.23 & 0.27 & 0.25 & 0.27 & 0.25 & 0.27 & 0.25 \\ 
				$\beta_1$ & 0.38 & 0.35 & 0.33 & 0.38 & 0.35 & 0.38 & 0.35 & 0.38 & 0.35 \\ 
				$\beta_2$ & 0.44 & 0.40 & 0.37 & 0.44 & 0.40 & 0.44 & 0.40 & 0.44 & 0.40 \\ 
				$\bm{\phi}$ & 0.28 & 0.26 & 0.25 & 0.24 & 0.26 & 0.28 & 0.26 & 0.24 & 0.26 \\ 
				$\bm{\rho}$ & 0.39 & 0.35 & 0.33 & 0.47 & 0.34 & 0.40 & 0.35 & 0.47 & 0.34 \\ 
				$\tilde{\bm{\rho}}$ & 0.18 & 0.17 & 0.16 & 0.18 & 0.17 & 0.18 & 0.17 & 0.18 & 0.17 \\ 
				\midrule
				$\alpha$ & 0.10 & 0.10 & 0.10 & 0.10 & 0.10 & 0.10 & 0.10 & 0.10 & 0.10 \\ 
				$\tau$ & 3.75 & 3.80 & 3.78 & 3.82 & 3.78 & 3.75 & 3.80 & 3.82 & 3.78 \\ 
				$\gamma$ & 0.11 & 0.10 & 0.10 & 0.11 & 0.10 & 0.11 & 0.10 & 0.11 & 0.10 \\ 
				$\beta_1$ & 0.15 & 0.14 & 0.14 & 0.15 & 0.14 & 0.15 & 0.14 & 0.15 & 0.14 \\ 
				$\beta_2$ & 0.17 & 0.15 & 0.15 & 0.17 & 0.15 & 0.17 & 0.15 & 0.17 & 0.15 \\ 
				$\bm{\phi}$ & 0.14 & 0.12 & 0.12 & 0.14 & 0.12 & 0.14 & 0.12 & 0.14 & 0.12 \\ 
				$\bm{\rho}$ & 0.80 & 0.38 & 0.35 & 0.95 & 0.36 & 0.80 & 0.38 & 0.97 & 0.36 \\ 
				$\tilde{\bm{\rho}}$ & 0.15 & 0.15 & 0.13 & 0.15 & 0.14 & 0.15 & 0.15 & 0.15 & 0.14 \\ 
				\midrule
				$\alpha$ & 0.33 & 0.32 & 0.31 & 0.34 & 0.32 & 0.33 & 0.32 & 0.34 & 0.32 \\ 
				$\tau$ & 7.05 & 6.73 & 6.34 & 6.95 & 6.74 & 7.04 & 6.73 & 6.95 & 6.74 \\ 
				$\gamma$ & 0.25 & 0.22 & 0.21 & 0.25 & 0.22 & 0.25 & 0.22 & 0.25 & 0.22 \\ 
				$\beta_1$ & 0.35 & 0.32 & 0.30 & 0.35 & 0.32 & 0.35 & 0.32 & 0.35 & 0.32 \\ 
				$\beta_2$ & 0.41 & 0.36 & 0.34 & 0.41 & 0.36 & 0.41 & 0.36 & 0.41 & 0.36 \\ 
				$\bm{\phi}$ & 0.24 & 0.23 & 0.22 & 0.21 & 0.23 & 0.24 & 0.23 & 0.21 & 0.23 \\ 
				$\bm{\rho}$ & 0.26 & 0.24 & 0.24 & 0.26 & 0.24 & 0.26 & 0.25 & 0.26 & 0.24 \\ 
				$\tilde{\bm{\rho}}$ & 0.14 & 0.13 & 0.13 & 0.14 & 0.13 & 0.14 & 0.13 & 0.14 & 0.13 \\ 
				\midrule
				$\alpha$ & 0.19 & 0.16 & 0.15 & 0.19 & 0.16 & 0.17 & 0.16 & 0.18 & 0.16 \\ 
				$\tau$ & 0.70 & 0.56 & 0.51 & 0.65 & 0.59 & 0.69 & 0.56 & 0.65 & 0.60 \\ 
				$\gamma$ & 0.14 & 0.12 & 0.11 & 0.13 & 0.12 & 0.14 & 0.12 & 0.13 & 0.12 \\ 
				$\beta_1$ & 0.19 & 0.17 & 0.16 & 0.19 & 0.17 & 0.19 & 0.17 & 0.19 & 0.17 \\ 
				$\beta_2$ & 0.22 & 0.19 & 0.18 & 0.22 & 0.19 & 0.22 & 0.19 & 0.22 & 0.19 \\ 
				$\bm{\phi}$ & 0.14 & 0.13 & 0.13 & 0.06 & 0.13 & 0.14 & 0.13 & 0.06 & 0.13 \\ 
				$\bm{\rho}$ & 0.06 & 0.05 & 0.05 & 0.05 & 0.05 & 0.06 & 0.05 & 0.05 & 0.05 \\ 
				$\tilde{\bm{\rho}}$ & 0.03 & 0.03 & 0.03 & 0.03 & 0.03 & 0.03 & 0.03 & 0.03 & 0.03 \\ 
				\midrule
				$\alpha$ & 0.56 & 0.55 & 0.55 & 0.57 & 0.55 & 0.56 & 0.55 & 0.57 & 0.55 \\ 
				$\tau$ & 15.66 & 15.36 & 14.94 & 15.73 & 15.33 & 15.65 & 15.36 & 15.73 & 15.32 \\ 
				$\gamma$ & 0.56 & 0.51 & 0.49 & 0.55 & 0.51 & 0.56 & 0.51 & 0.55 & 0.51 \\ 
				$\beta_1$ & 0.78 & 0.73 & 0.68 & 0.79 & 0.72 & 0.78 & 0.72 & 0.79 & 0.72 \\ 
				$\beta_2$ & 0.88 & 0.80 & 0.75 & 0.87 & 0.80 & 0.87 & 0.80 & 0.87 & 0.80 \\ 
				$\bm{\phi}$ & 0.63 & 0.58 & 0.56 & 0.60 & 0.58 & 0.64 & 0.58 & 0.60 & 0.58 \\ 
				$\bm{\rho}$ & 1.47 & 1.25 & 1.19 & 2.29 & 1.22 & 1.49 & 1.26 & 2.30 & 1.22 \\ 
				$\tilde{\bm{\rho}}$ & 0.58 & 0.56 & 0.51 & 0.58 & 0.55 & 0.58 & 0.56 & 0.59 & 0.55 \\ 
				\bottomrule
			\end{tabular}
		\end{tabular}
		\caption{Descriptive statistics of the RMSE (Equation 6.3) across all simulated datasets for all parameters in Equation 3.2. Individual values of vector parameters, i.e. ($ \bm{\phi}, \bm{\rho}, \bm{\tilde{\rho}} $) are pooled together and treated as belonging to a single parameter}
		\label{tab:rmse_relative_tot} 
	\end{table}
	
	\clearpage
	\subsection{Calibration plots} 
	\label{subsec:cal_plots}
	In this Section we report the Q-Q plot of the standardised, i.e. divided by their max $ 5001 $, SBC (see Section 5.1) obtained rank statistics against the ordered rank statistics of a discrete uniform $ [0, 1] $ distribution $ r / (L + 1) $ with $ r = 1, ..., L $ with $ L = 5001 $. Red lines indicate the $ 0.025$ and $ 0.975 $ quantiles of the ordered rank statistic, which correspond to the 95\% expected variation intervals (See also Table 1 for a more compact summary). Quantities on both axes are elevated to the $ 20^{th} $ power for improved readability.

	\clearpage	
	
	\begin{sidewaysfigure}
			\includegraphics[width=\wdIm\linewidth]{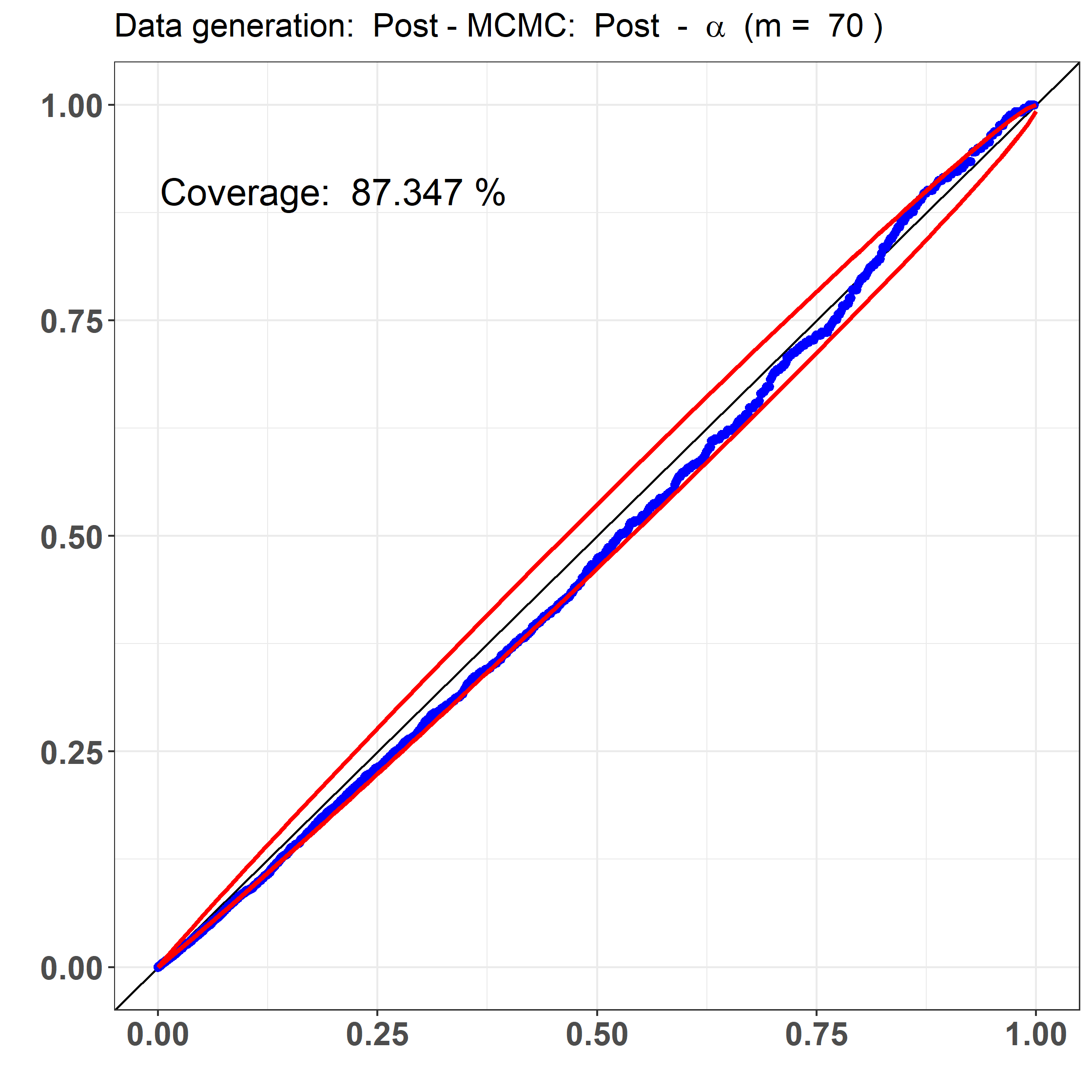}%
			\includegraphics[width=\wdIm\linewidth]{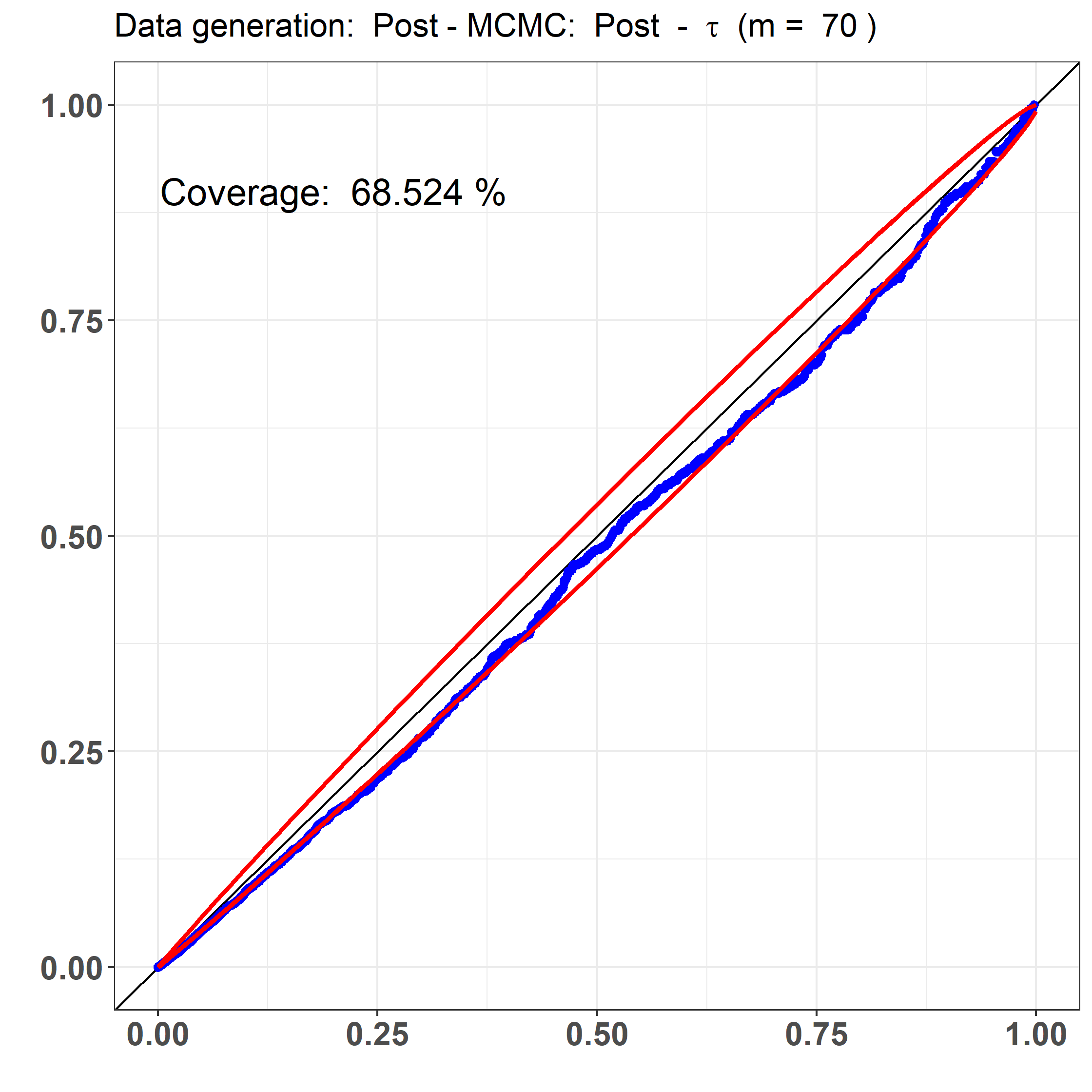}%
			\includegraphics[width=\wdIm\linewidth]{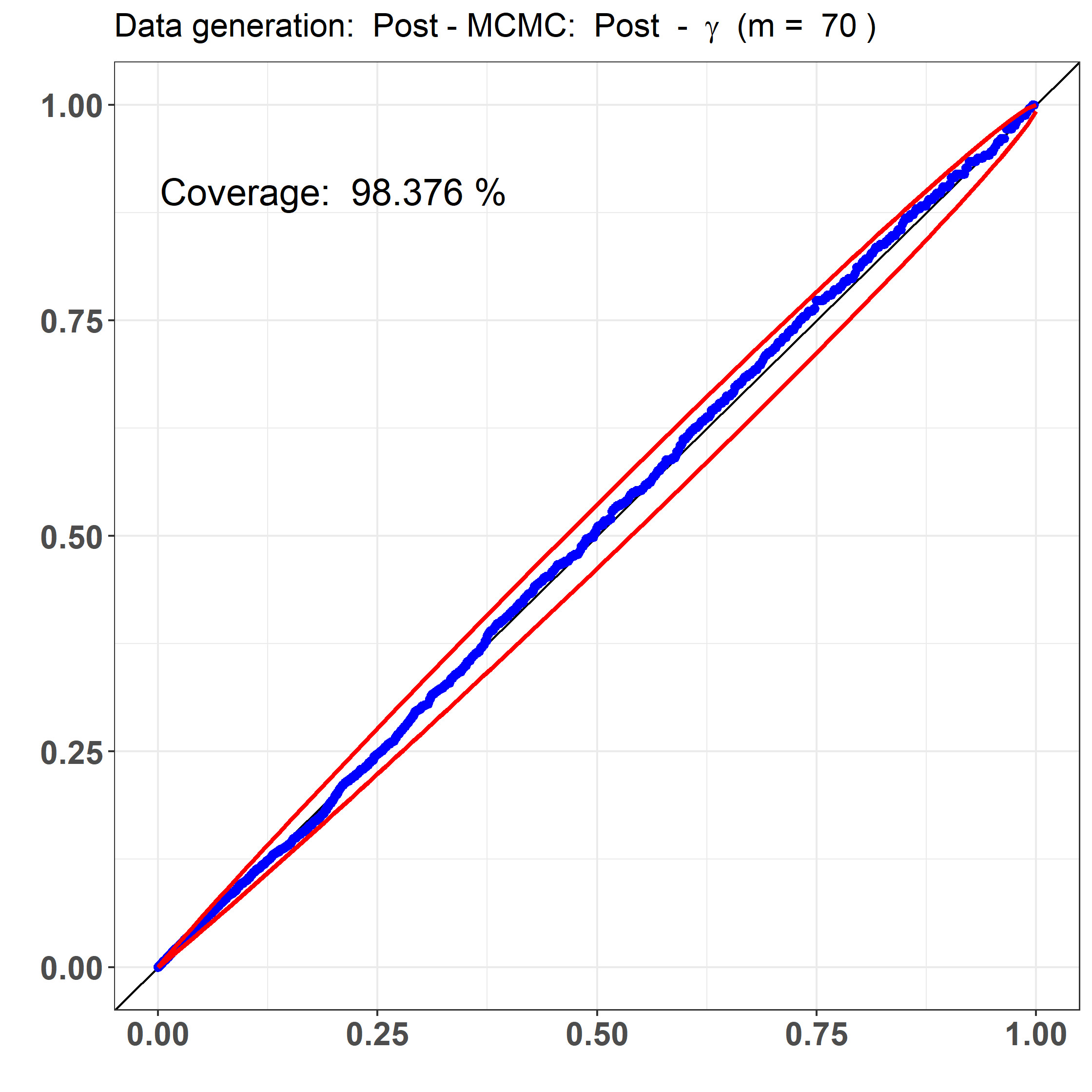} %
			\includegraphics[width=\wdIm\linewidth]{fig/ch5/sbc_post_beta[1]_70}\\%
			\includegraphics[width=\wdIm\linewidth]{fig/ch5/sbc_post_beta[2]_70}%
			\includegraphics[width=\wdIm\linewidth]{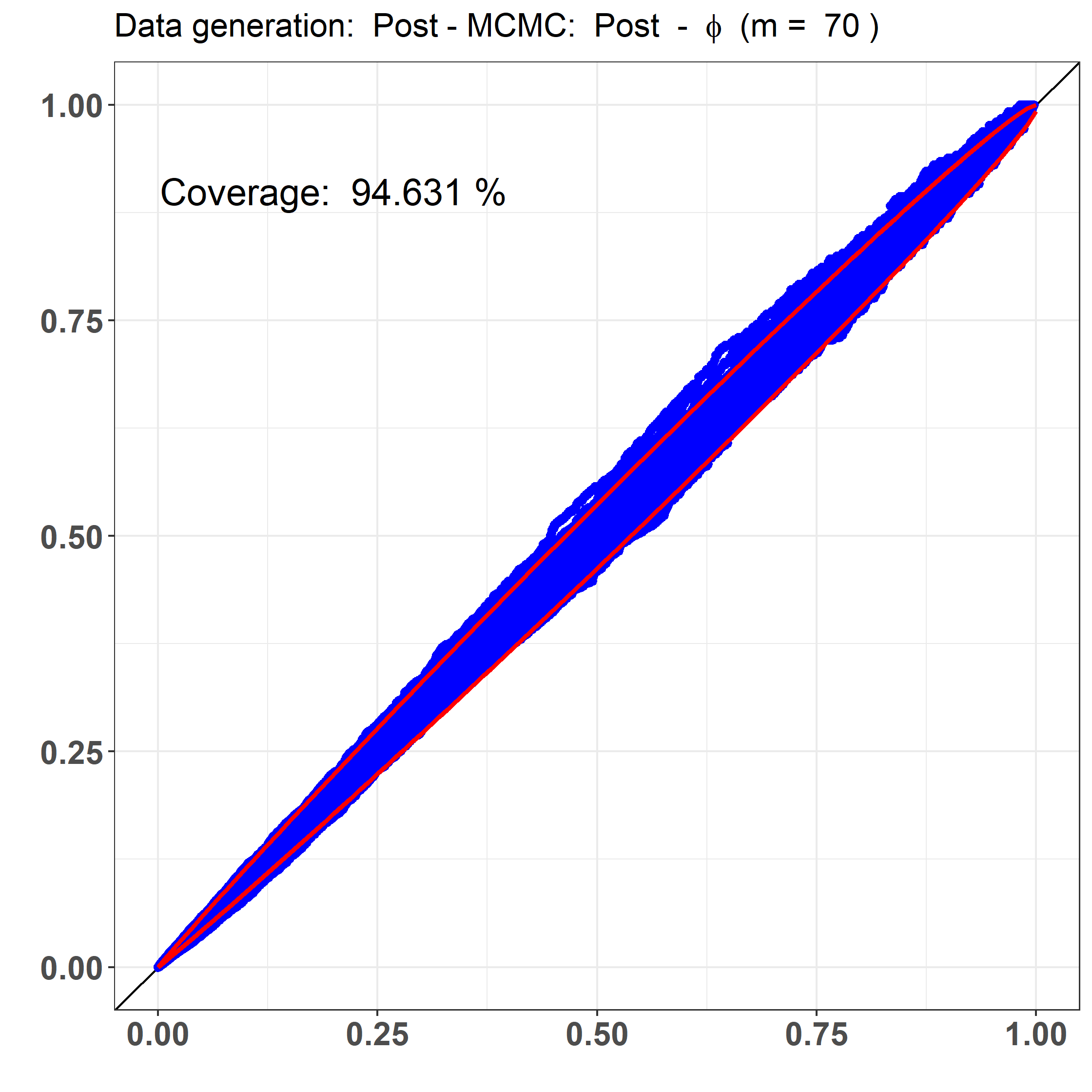}%
			\includegraphics[width=\wdIm\linewidth]{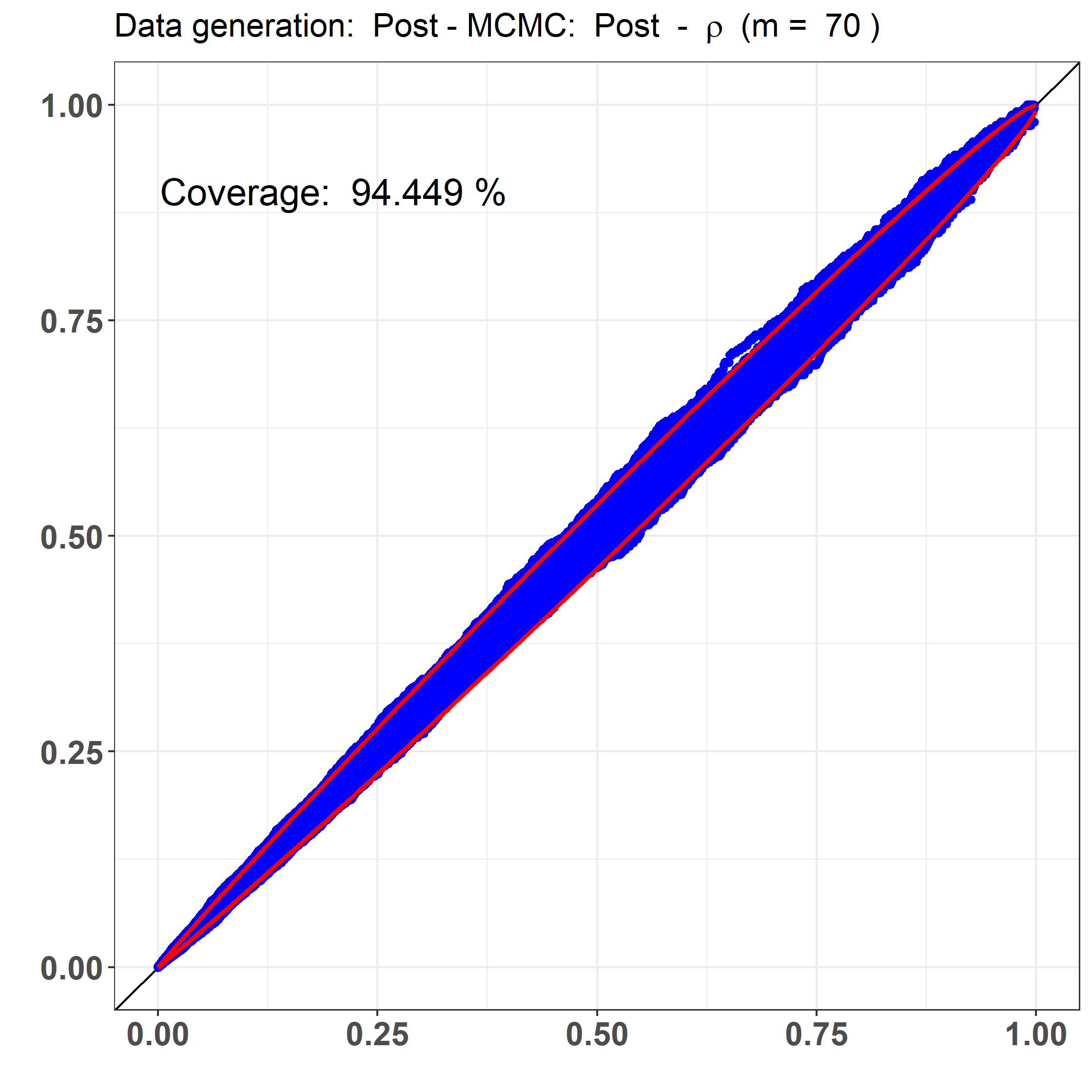}%
			\includegraphics[width=\wdIm\linewidth]{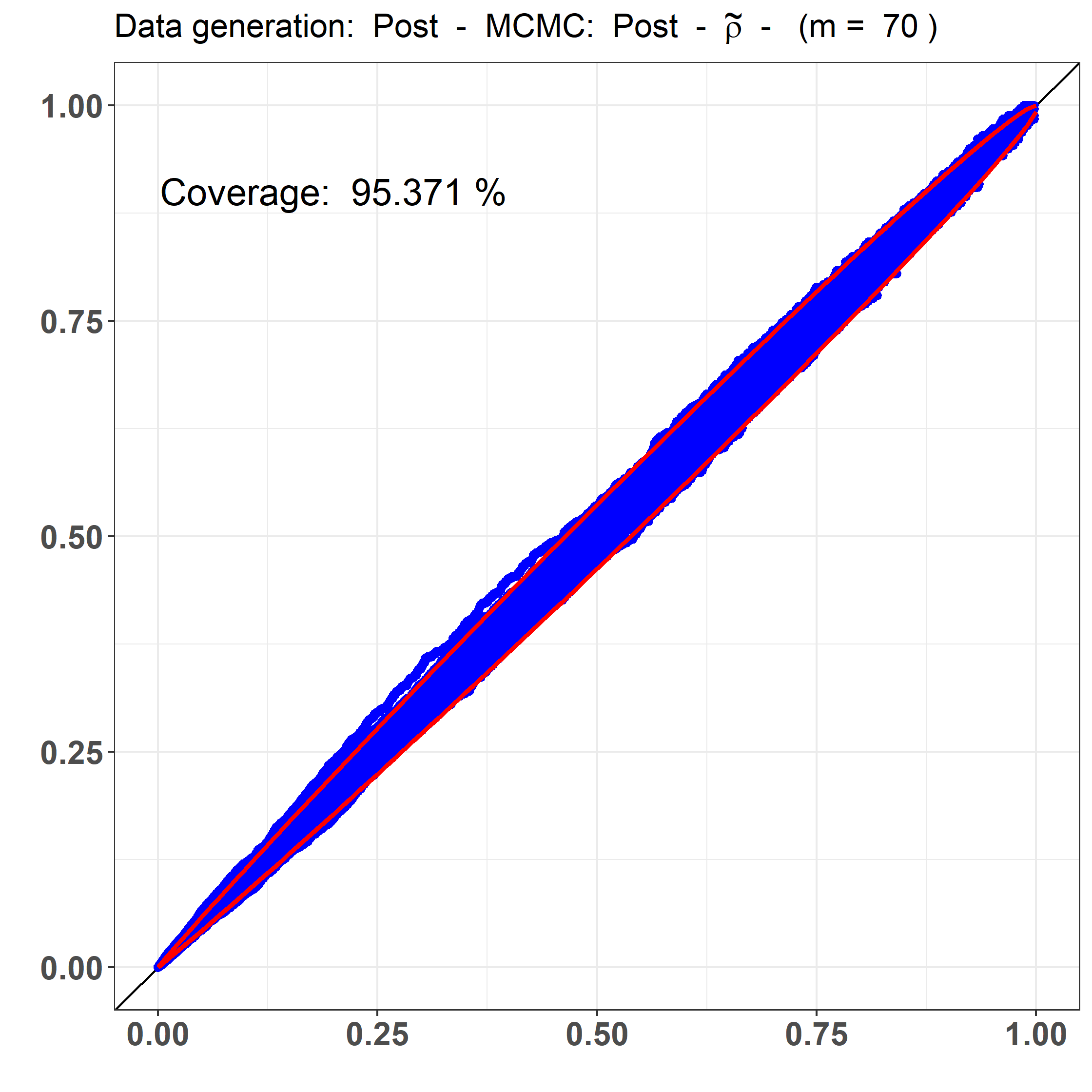}%
		\caption{Data generation: Post - MCMC: Post - $ m = 70 $}
	\end{sidewaysfigure}
	\newpage
	
	\begin{sidewaysfigure}
			\includegraphics[width=\wdIm\linewidth]{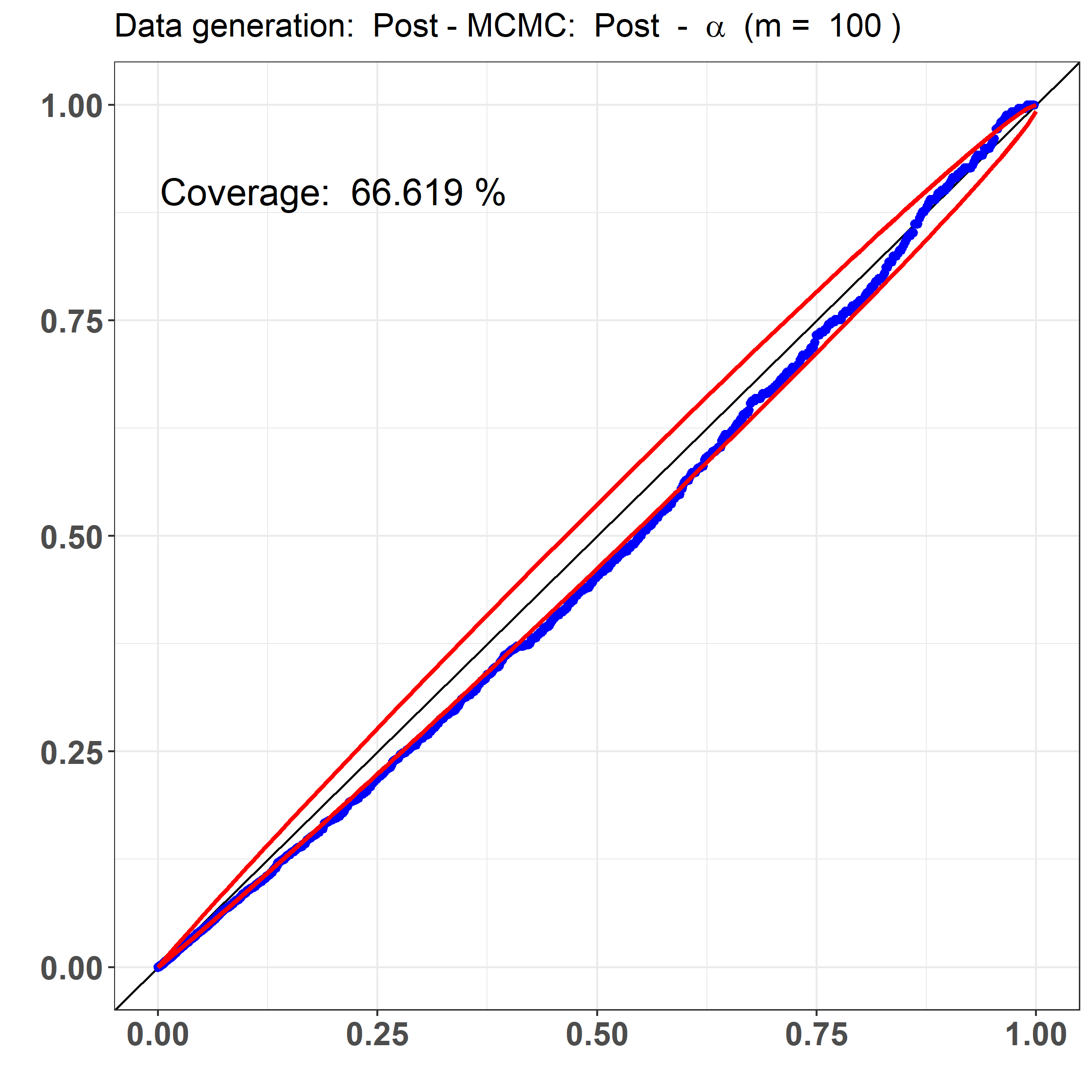}%
			\includegraphics[width=\wdIm\linewidth]{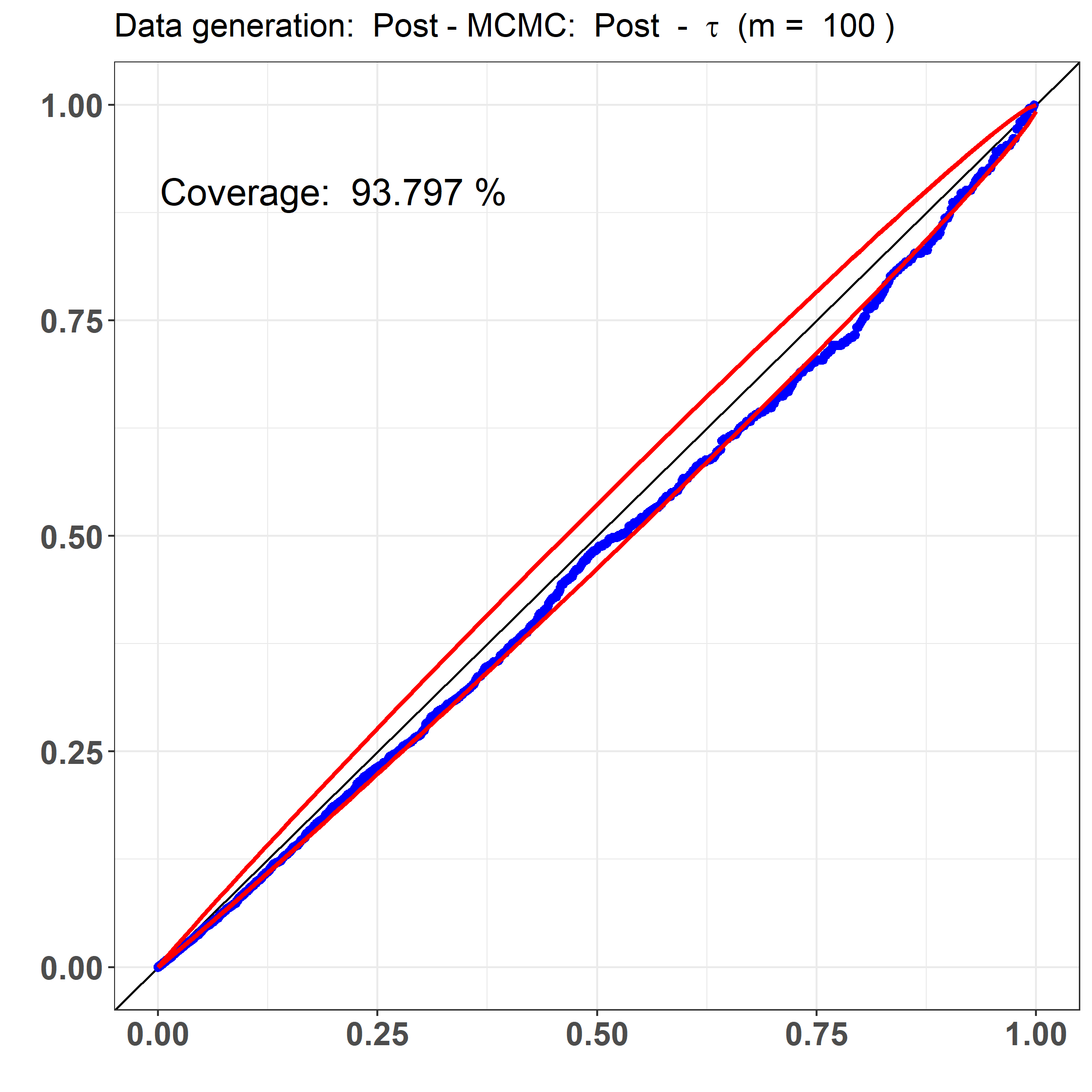}%
			\includegraphics[width=\wdIm\linewidth]{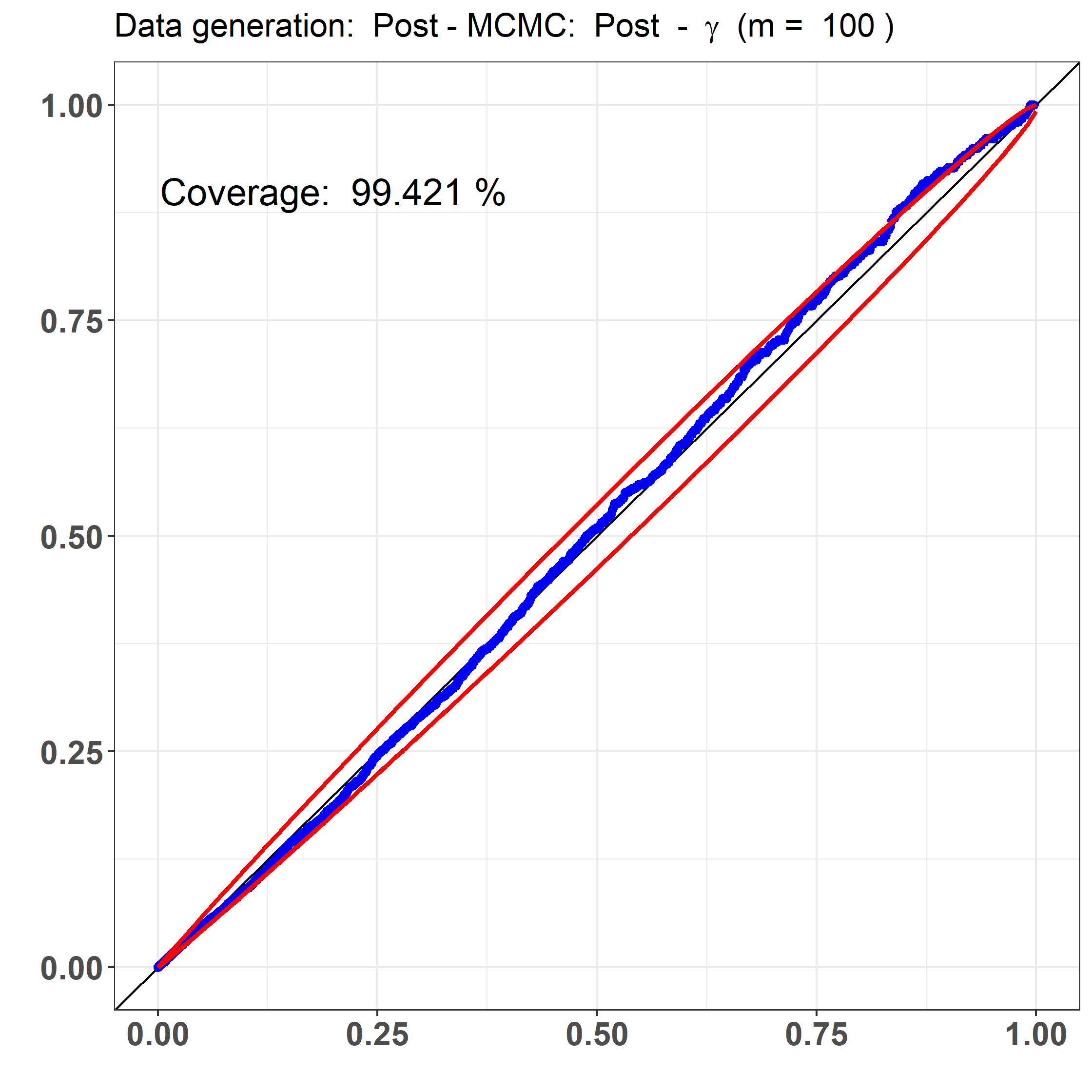} %
			\includegraphics[width=\wdIm\linewidth]{fig/ch5/sbc_post_beta[1]_100}\\%
			\includegraphics[width=\wdIm\linewidth]{fig/ch5/sbc_post_beta[2]_100}%
			\includegraphics[width=\wdIm\linewidth]{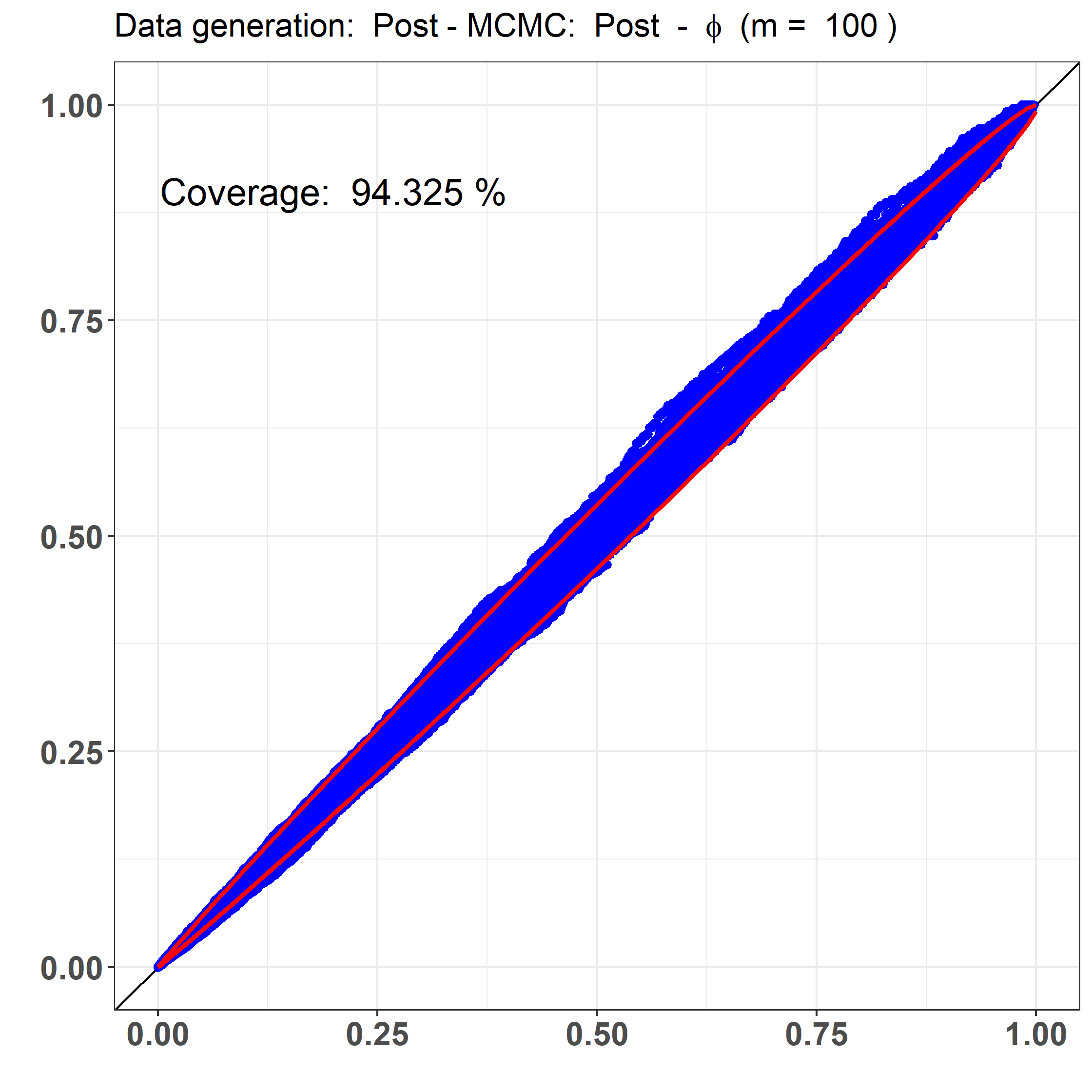}%
			\includegraphics[width=\wdIm\linewidth]{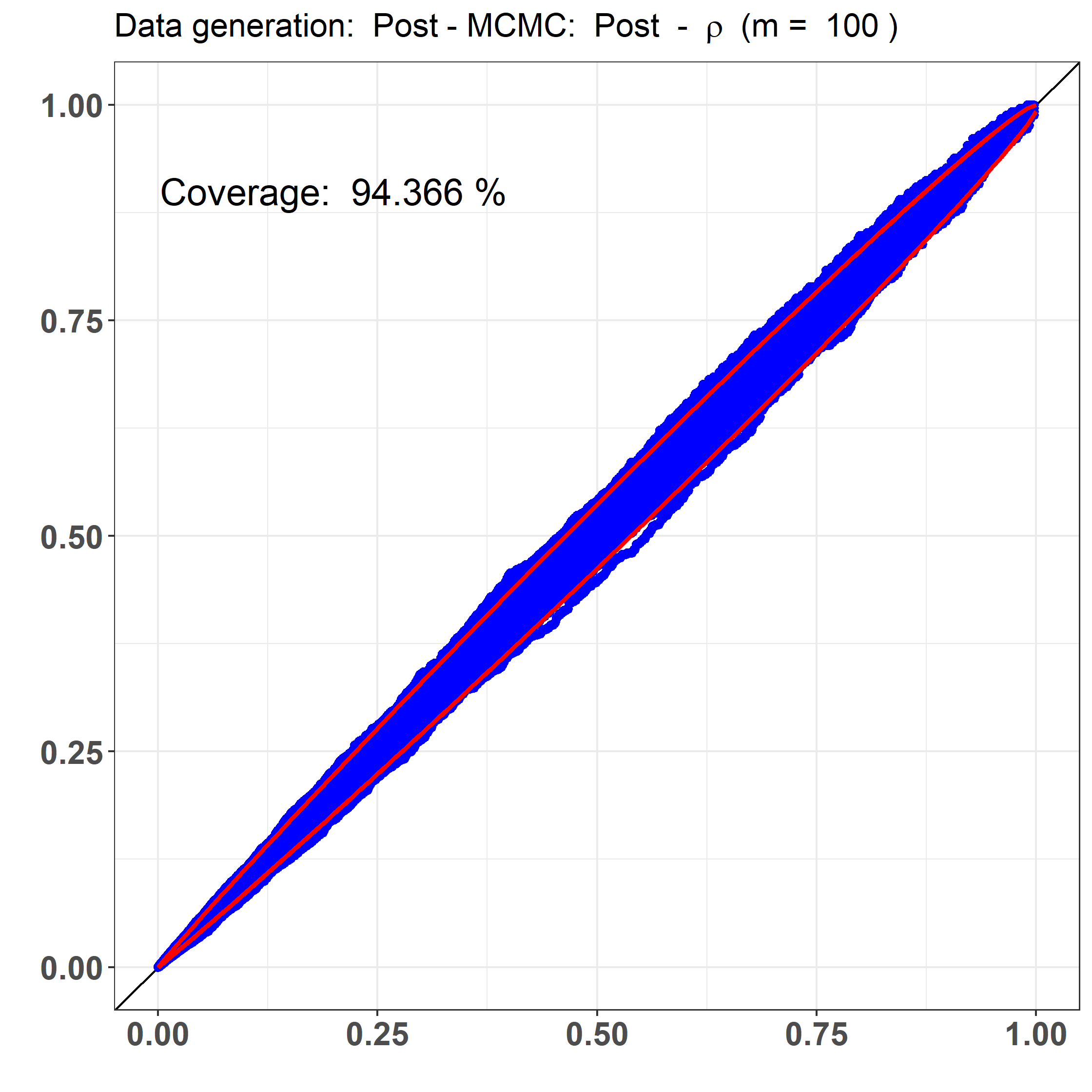}%
			\includegraphics[width=\wdIm\linewidth]{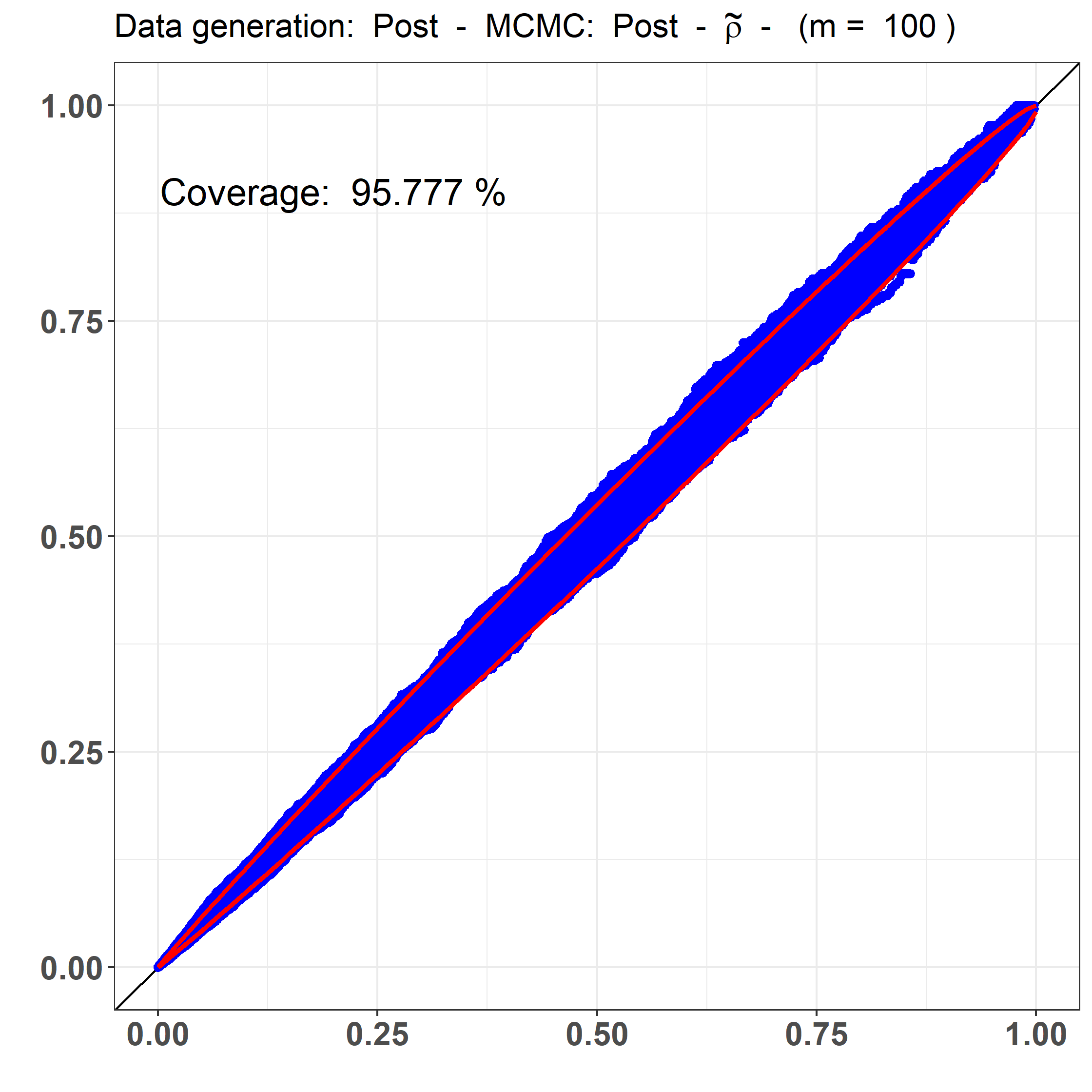}%
		\caption{Data generation: Post - MCMC: Post - $ m = 100 $}
	\end{sidewaysfigure}
	\newpage
	
	\begin{sidewaysfigure}
			\includegraphics[width=\wdIm\linewidth]{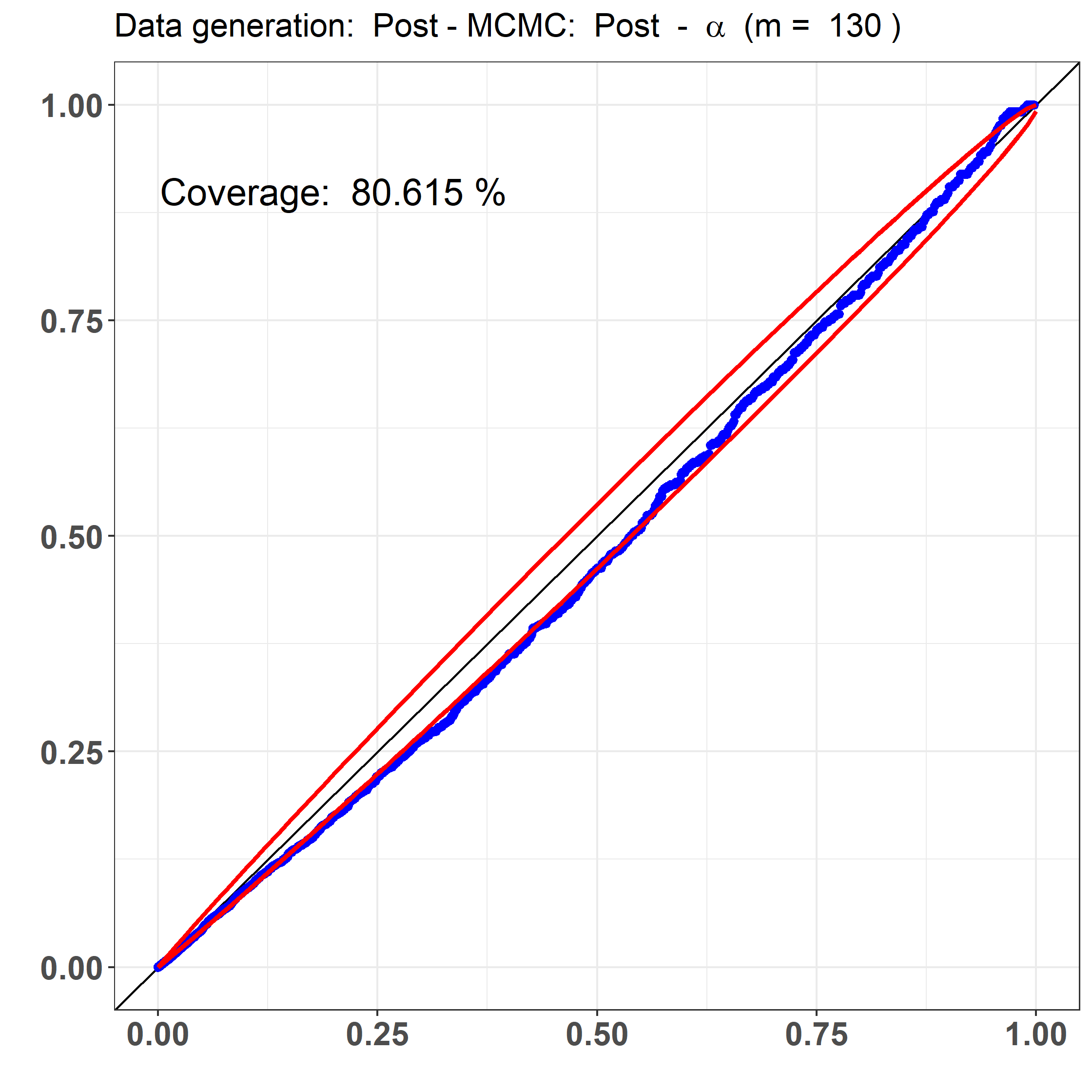}%
			\includegraphics[width=\wdIm\linewidth]{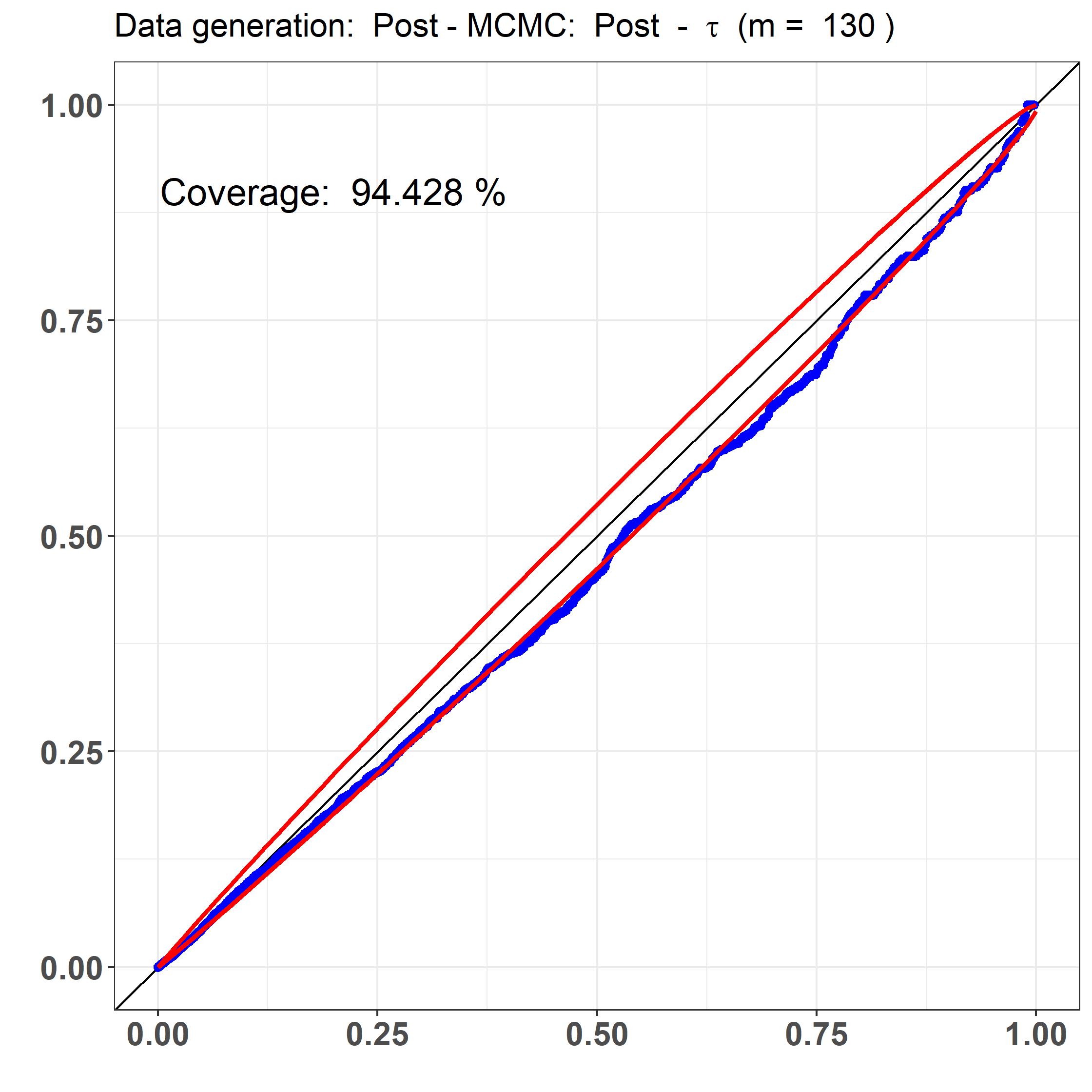}%
			\includegraphics[width=\wdIm\linewidth]{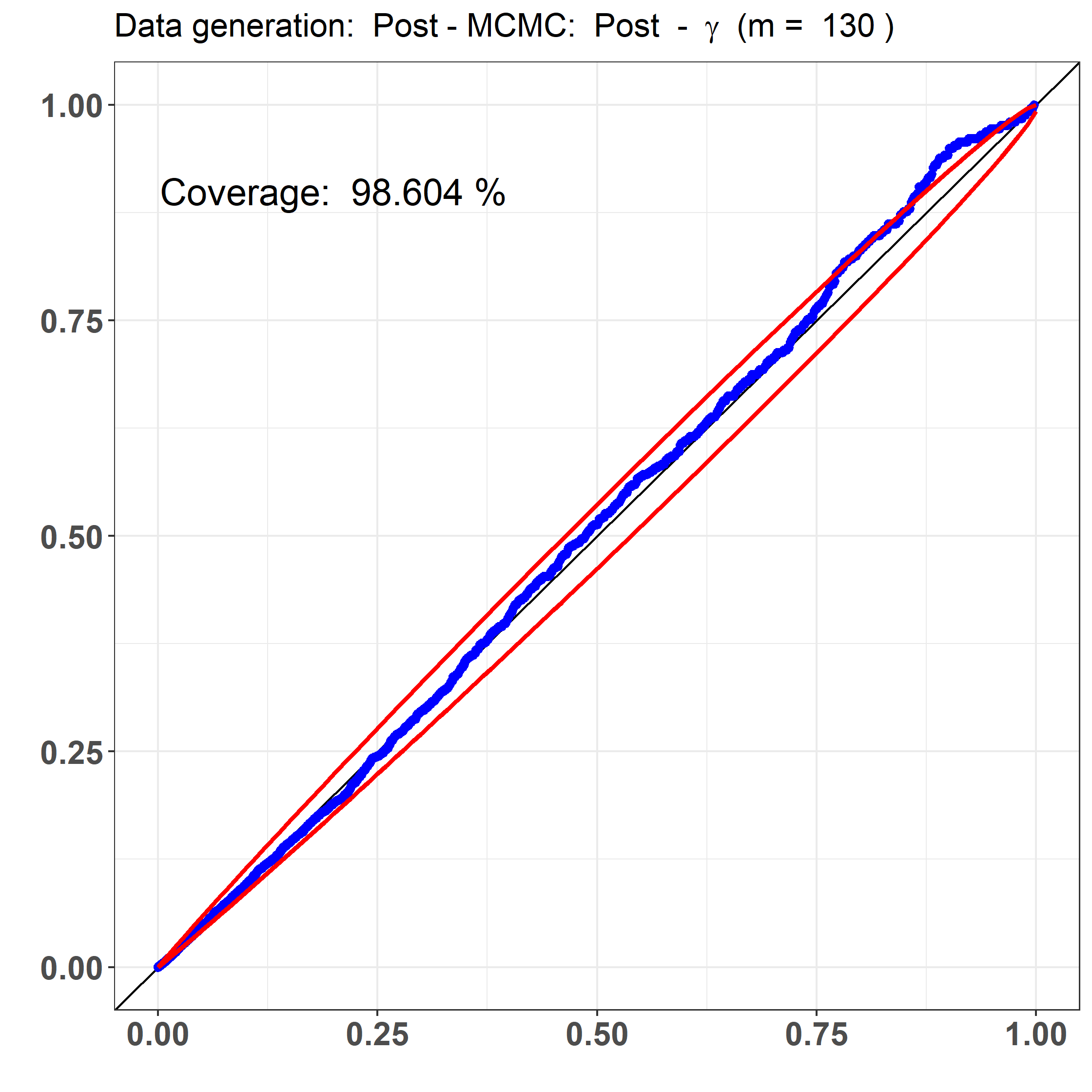} %
			\includegraphics[width=\wdIm\linewidth]{fig/ch5/sbc_post_beta[1]_130}\\%
			\includegraphics[width=\wdIm\linewidth]{fig/ch5/sbc_post_beta[2]_130}%
			\includegraphics[width=\wdIm\linewidth]{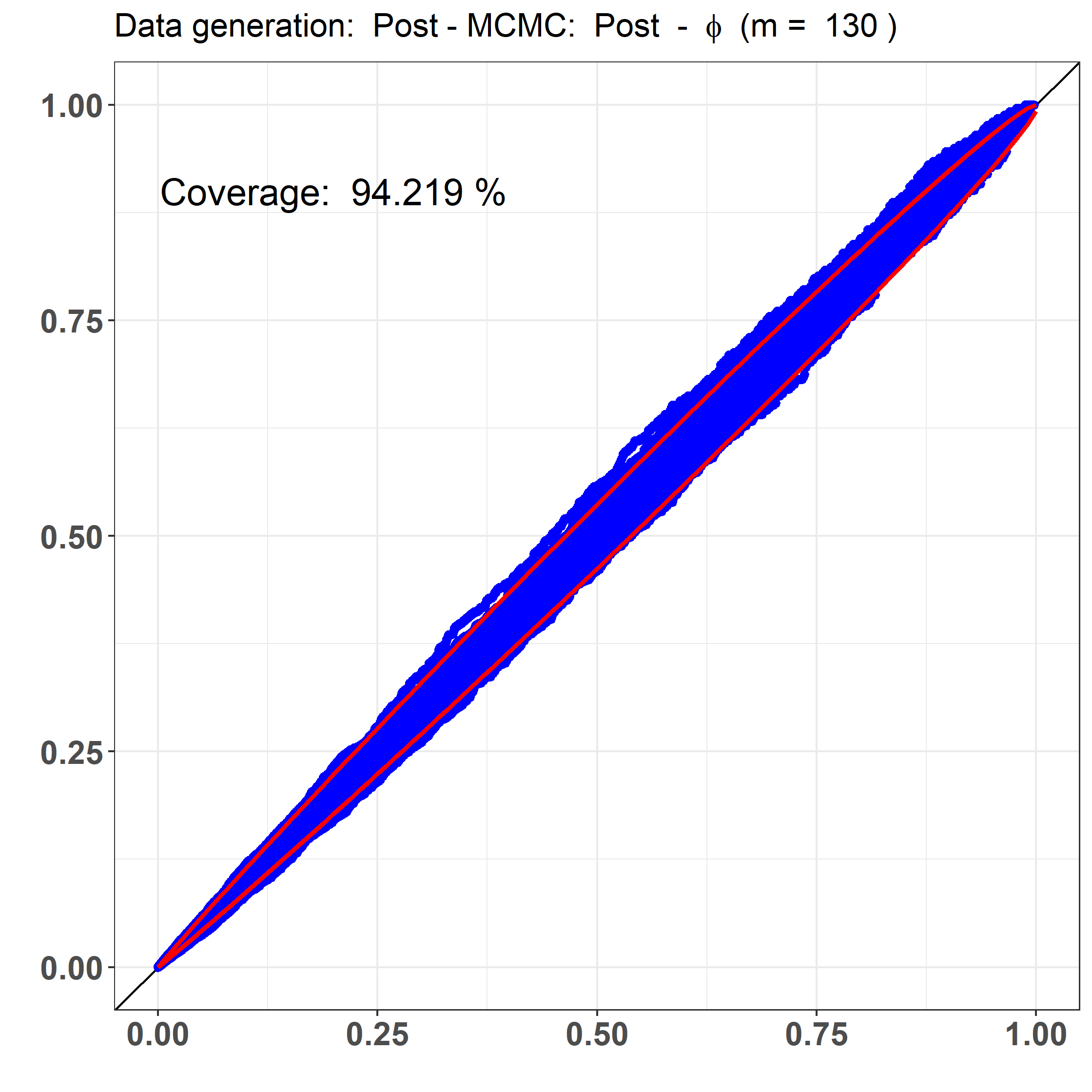}%
			\includegraphics[width=\wdIm\linewidth]{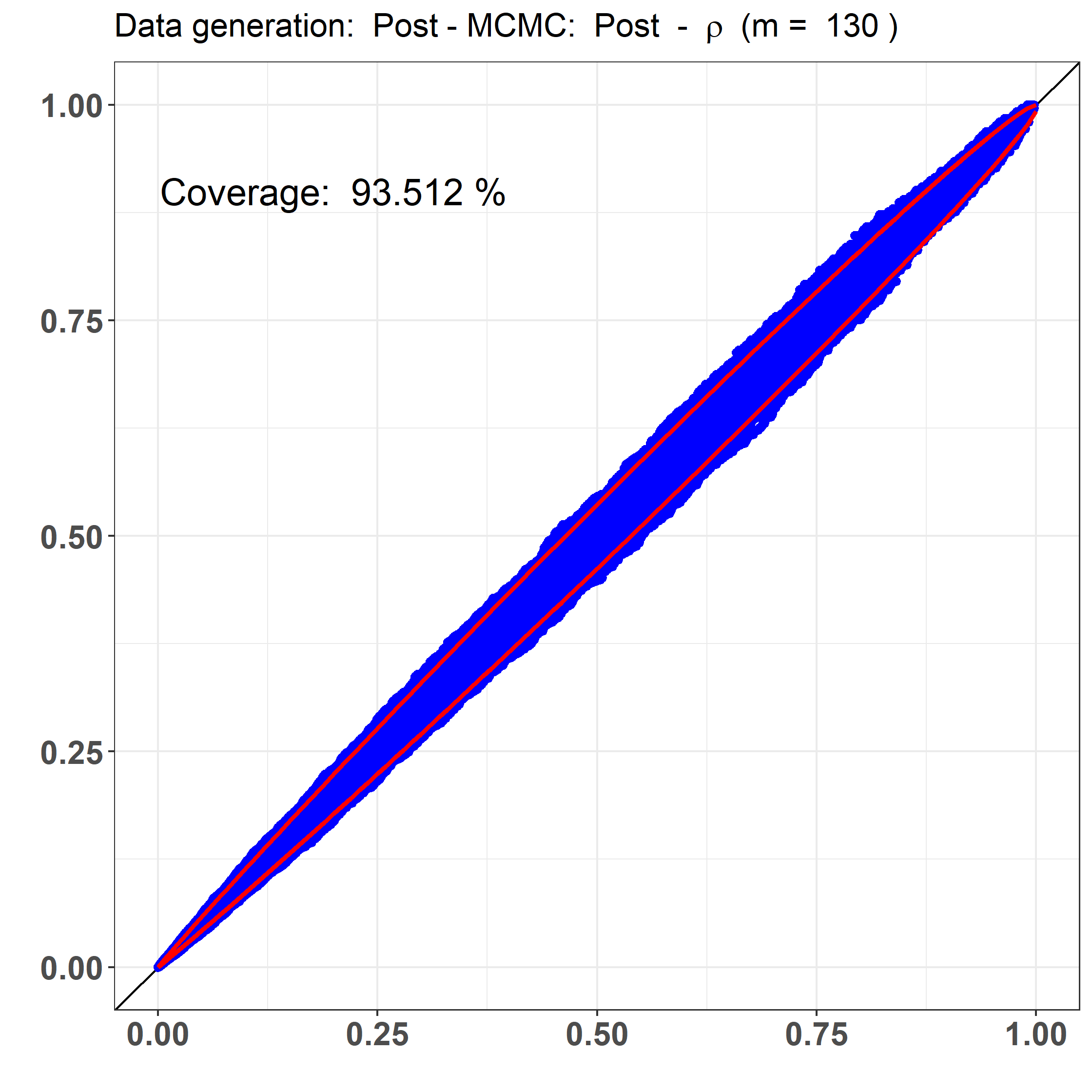}%
			\includegraphics[width=\wdIm\linewidth]{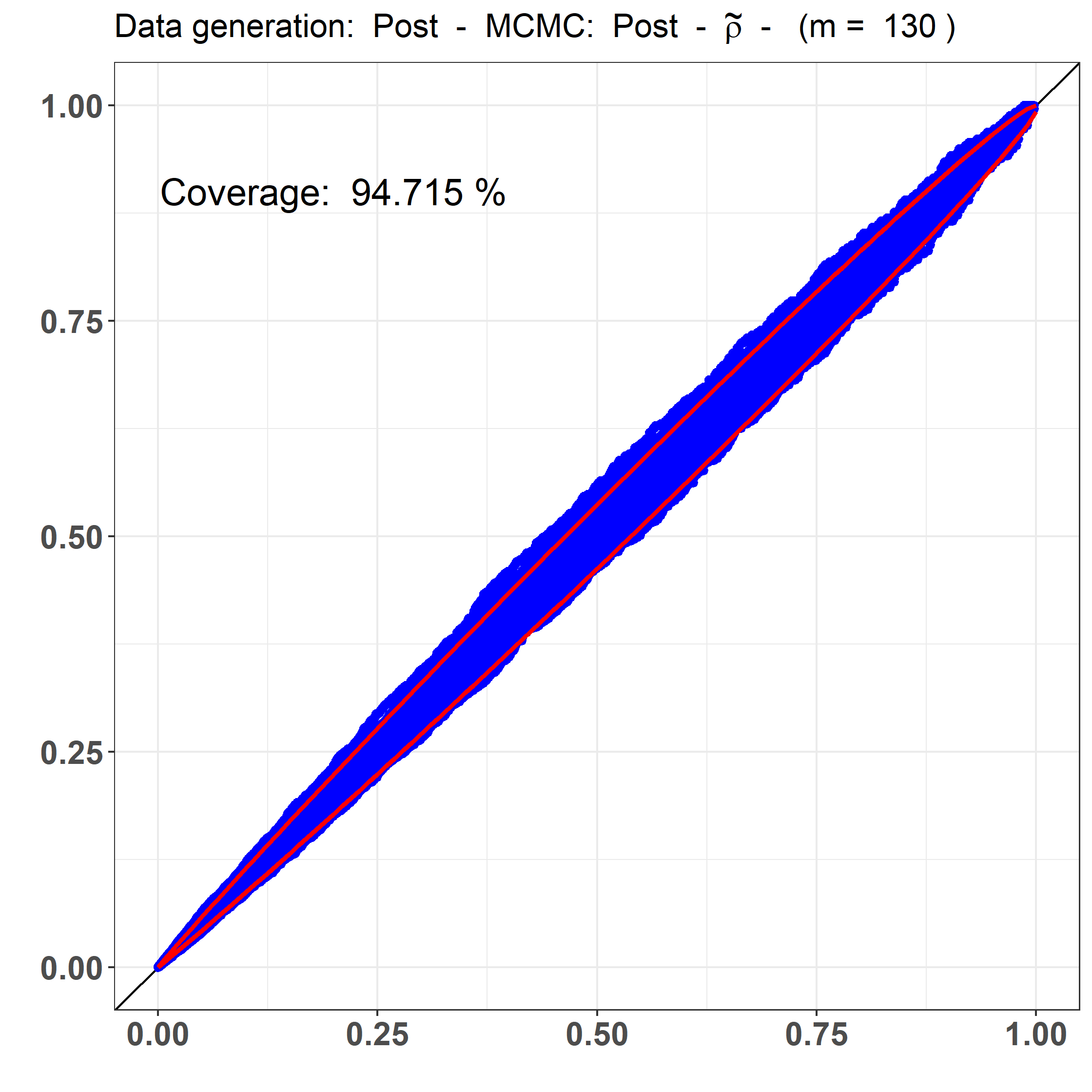}%
		\caption{Data generation: Post - MCMC: Post - $ m = 130 $}
	\end{sidewaysfigure}
	\newpage
	
	\begin{sidewaysfigure}
			\includegraphics[width=\wdIm\linewidth]{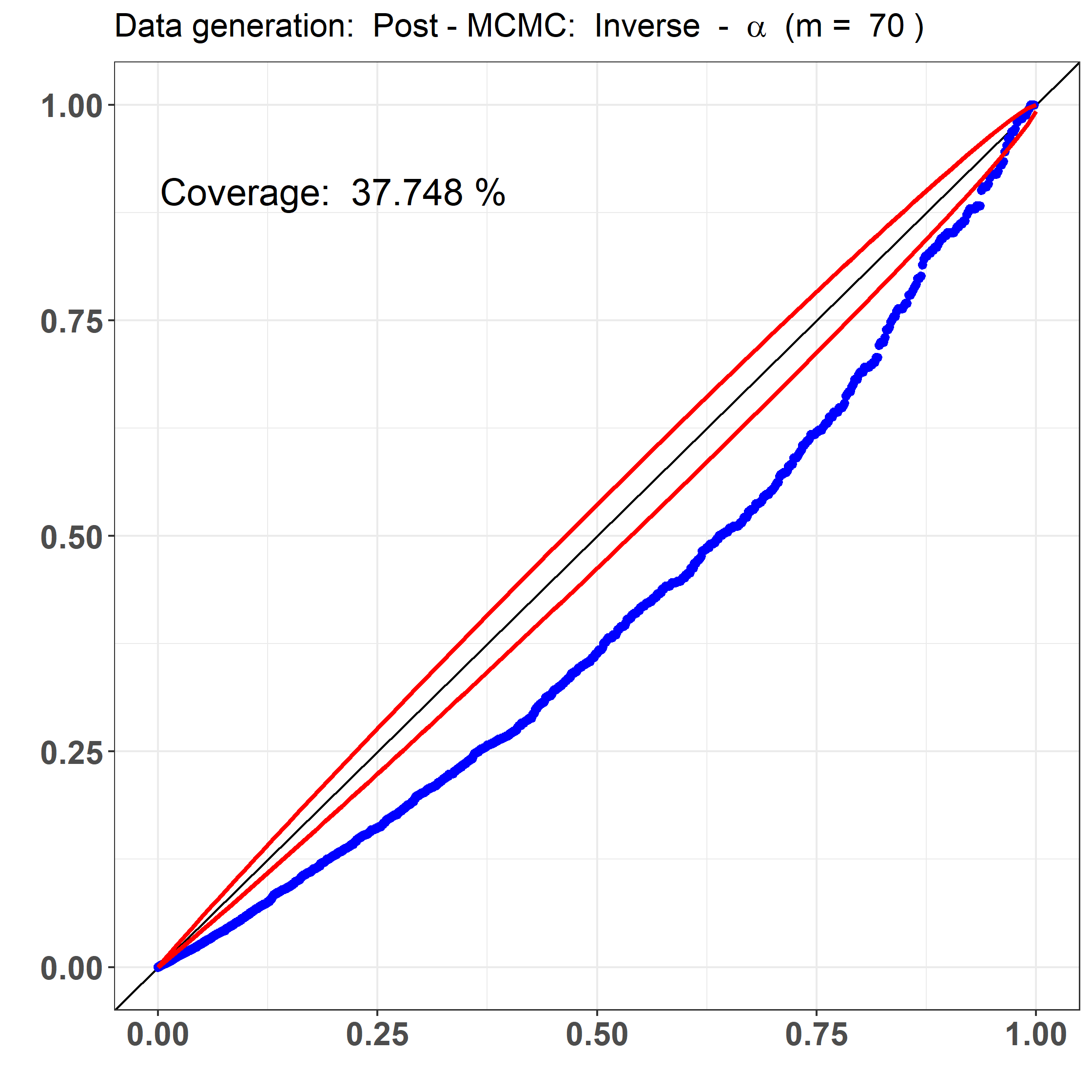}%
			\includegraphics[width=\wdIm\linewidth]{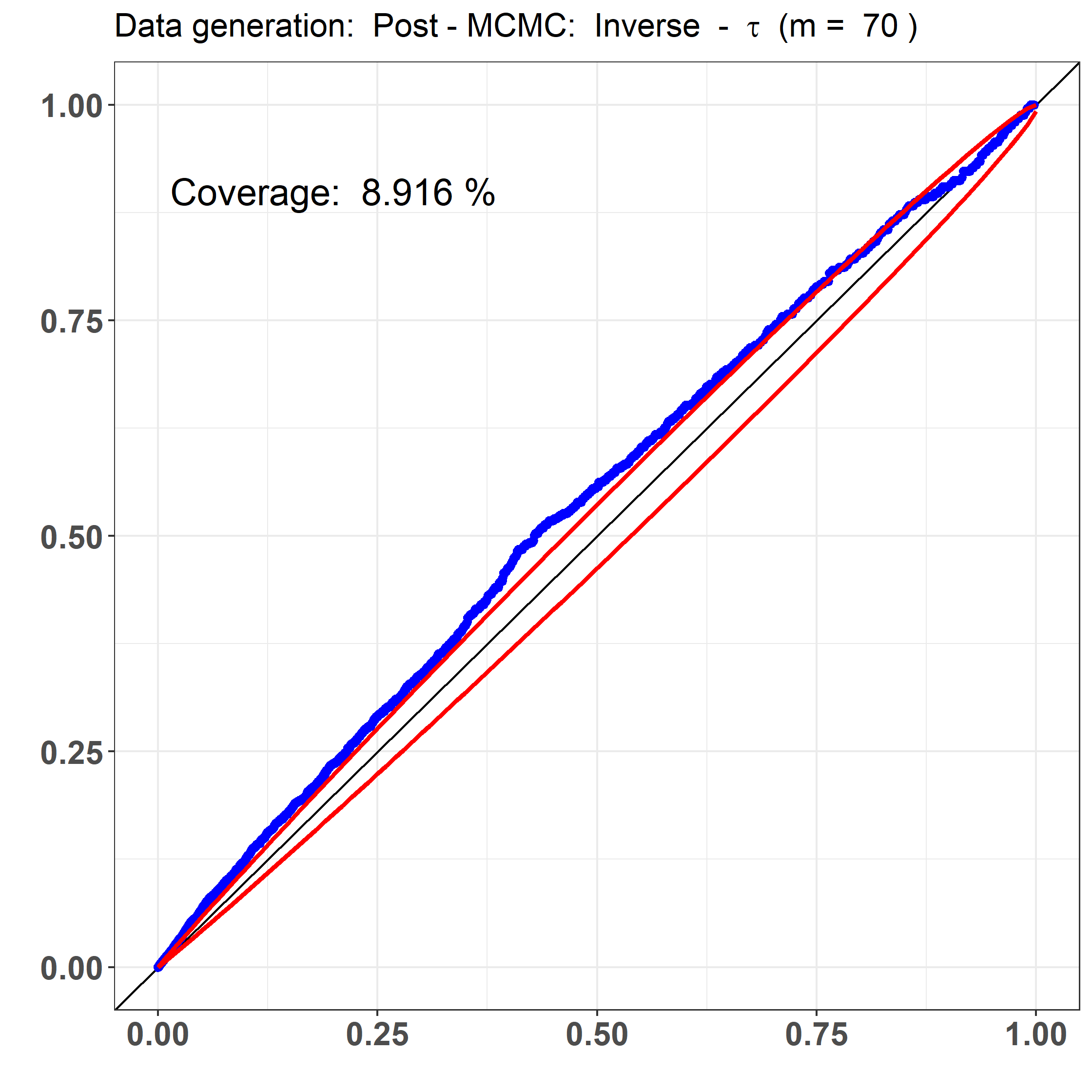}%
			\includegraphics[width=\wdIm\linewidth]{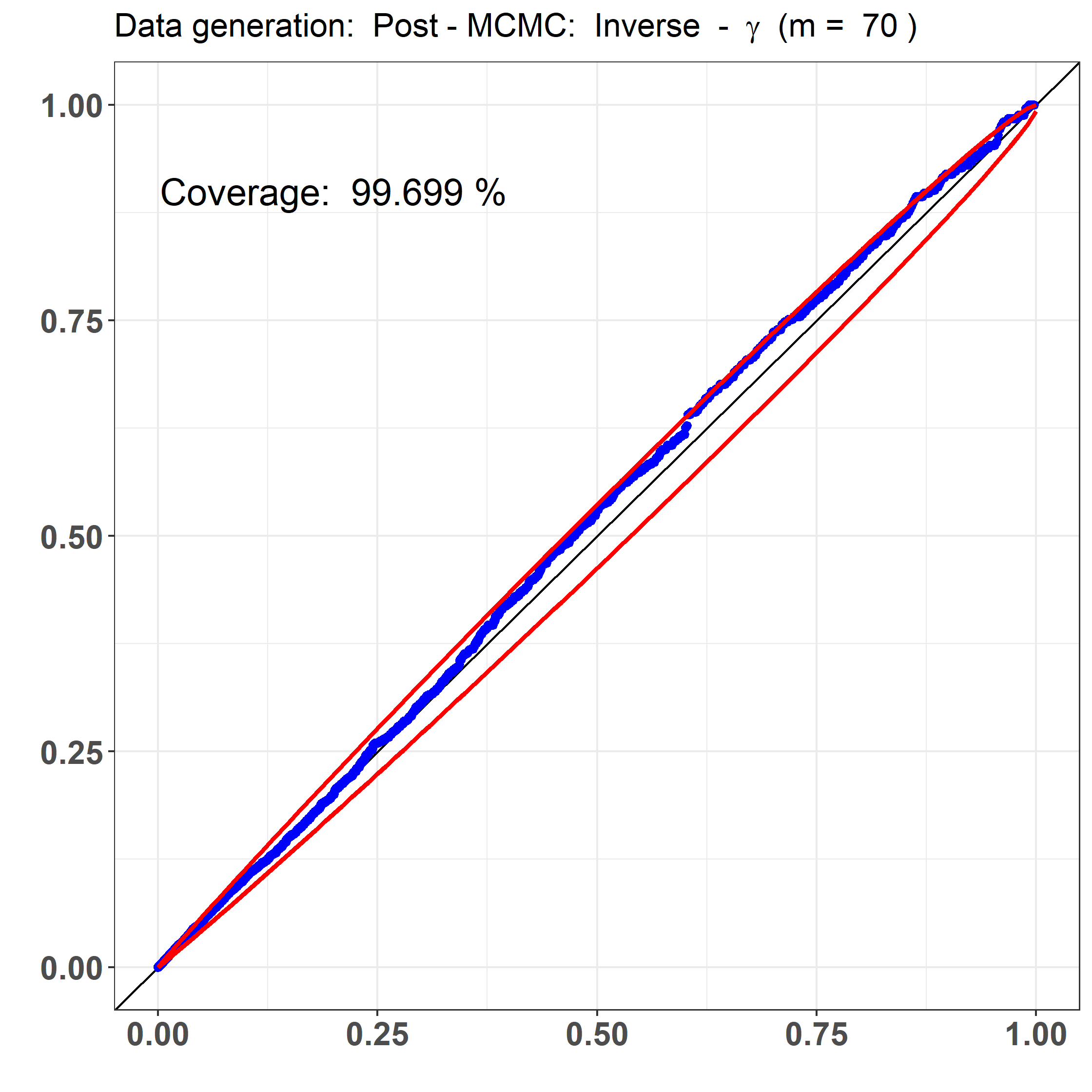} %
			\includegraphics[width=\wdIm\linewidth]{fig/ch5/sbc_inv_beta[1]_70}\\%
			\includegraphics[width=\wdIm\linewidth]{fig/ch5/sbc_inv_beta[2]_70}%
			\includegraphics[width=\wdIm\linewidth]{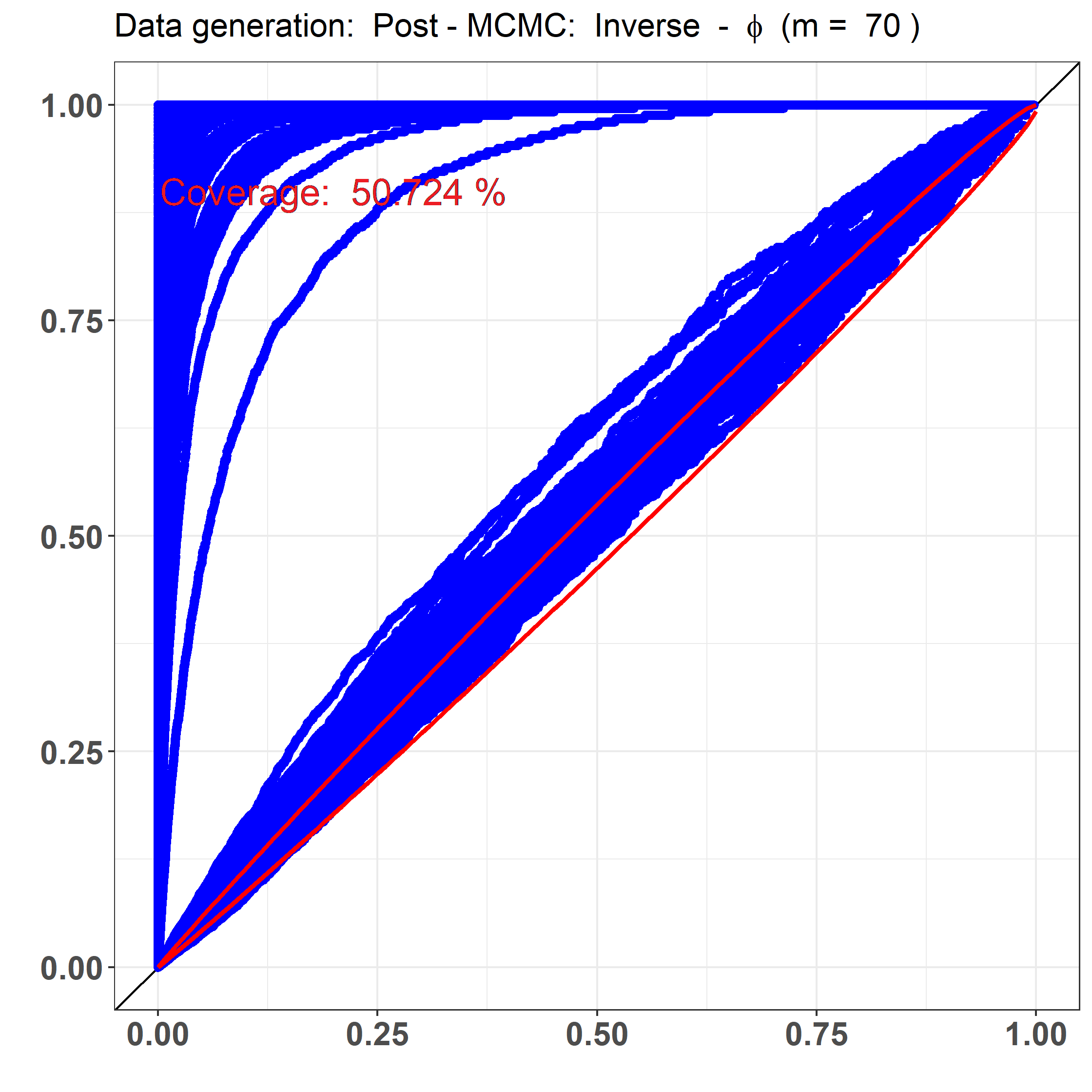}%
			\includegraphics[width=\wdIm\linewidth]{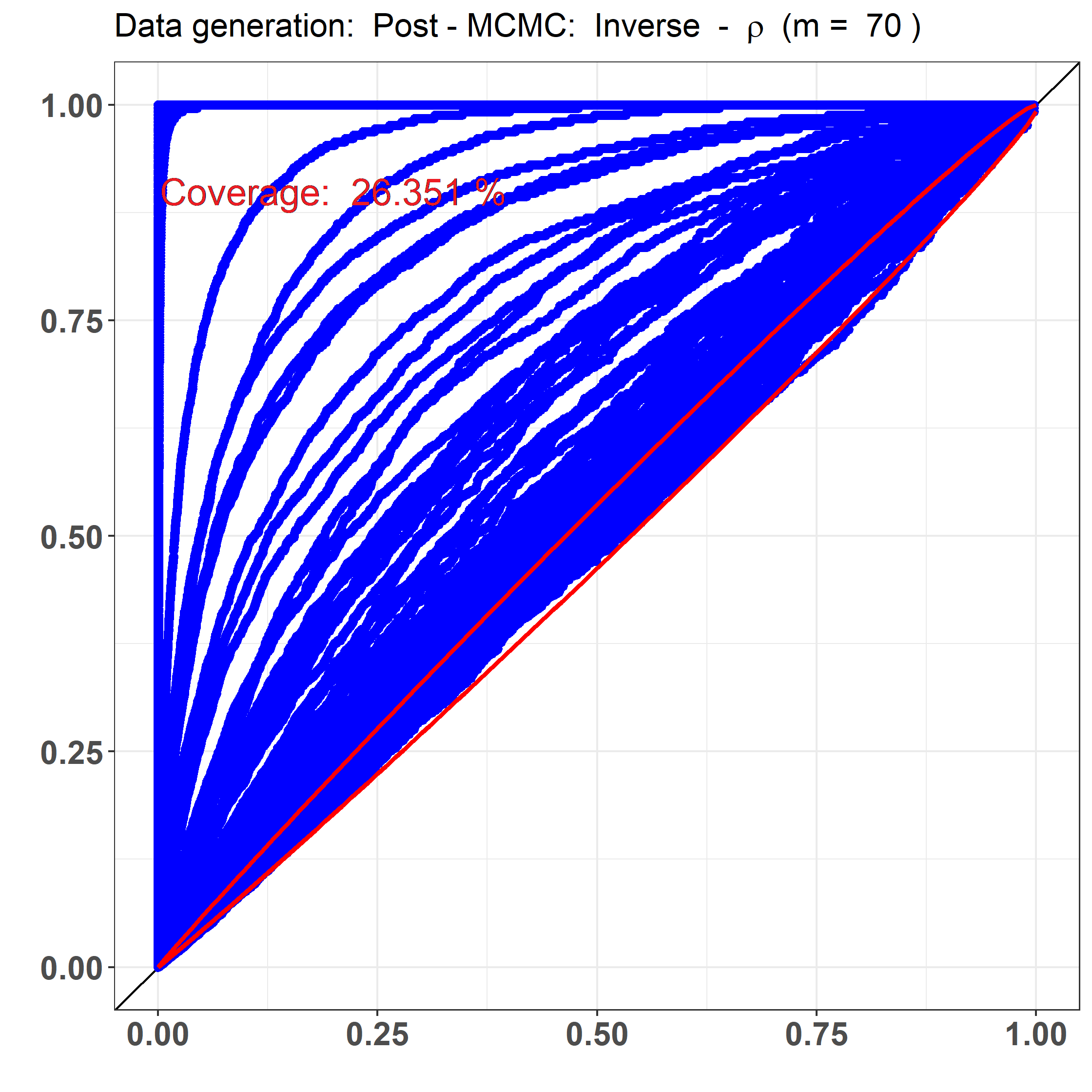}%
			\includegraphics[width=\wdIm\linewidth]{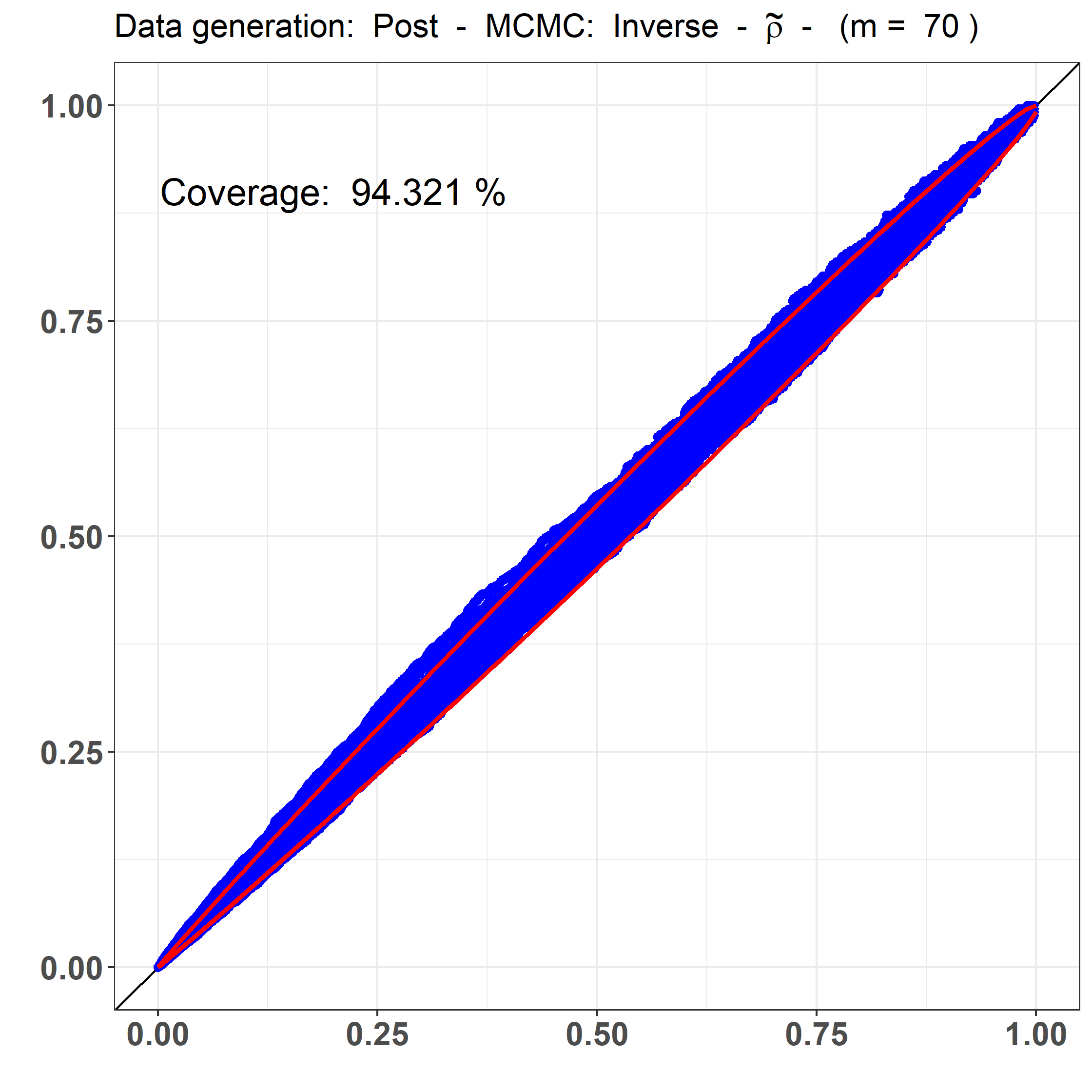}%
		\caption{Data generation: Post - MCMC: Inverse - $ m = 70 $}
	\end{sidewaysfigure}
	\newpage
	
	\begin{sidewaysfigure}
			\includegraphics[width=\wdIm\linewidth]{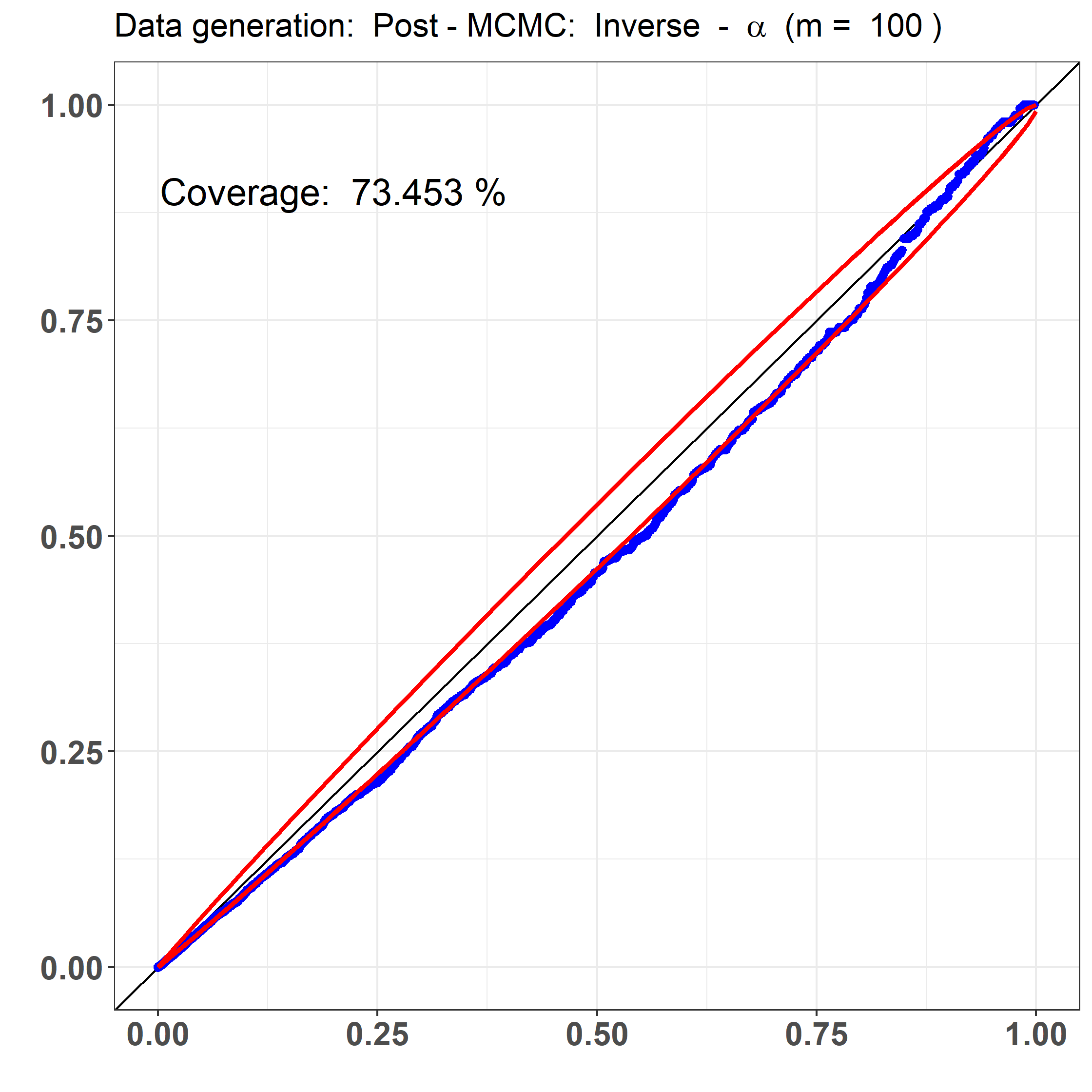}%
			\includegraphics[width=\wdIm\linewidth]{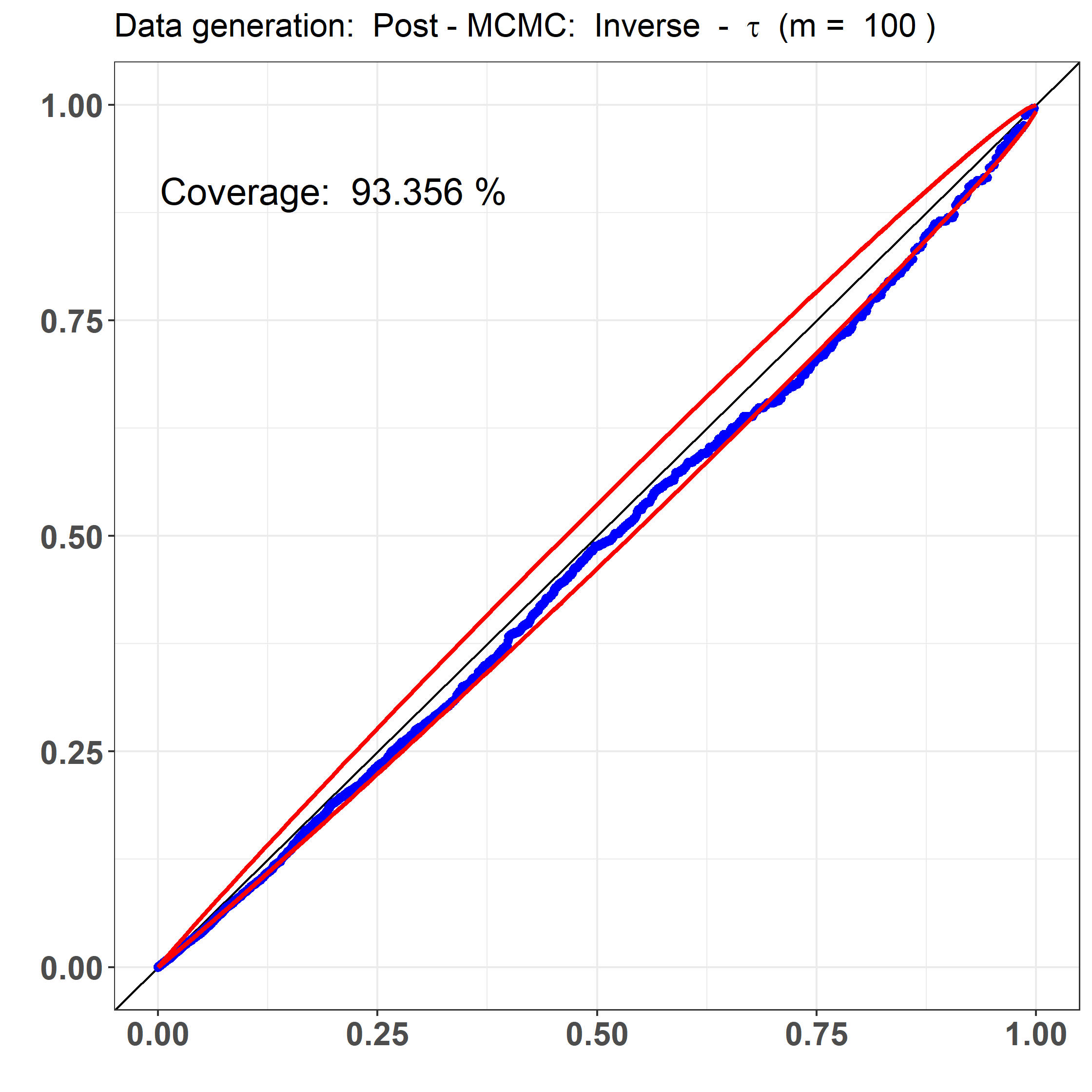}%
			\includegraphics[width=\wdIm\linewidth]{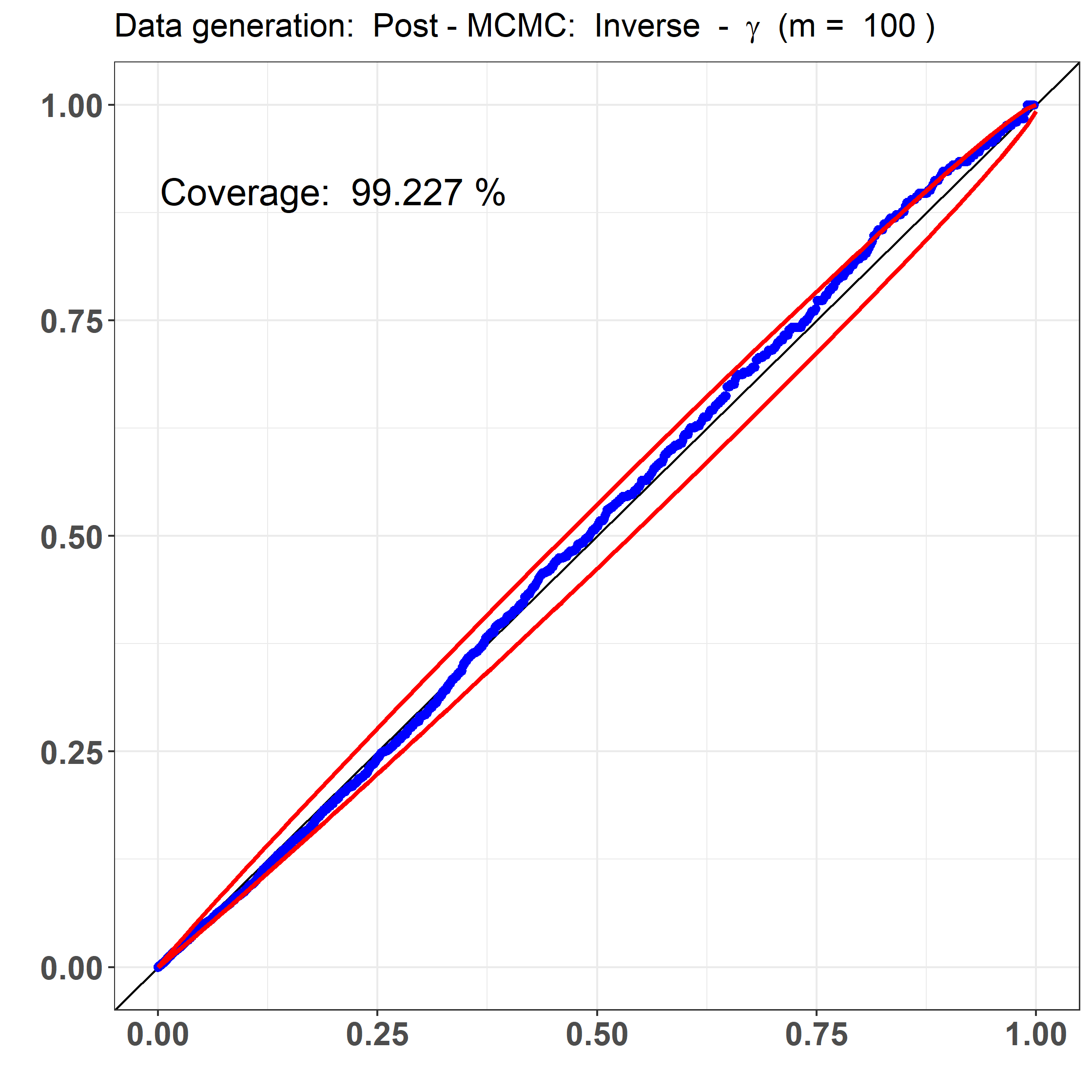} %
			\includegraphics[width=\wdIm\linewidth]{fig/ch5/sbc_inv_beta[1]_100}\\%
			\includegraphics[width=\wdIm\linewidth]{fig/ch5/sbc_inv_beta[2]_100}%
			\includegraphics[width=\wdIm\linewidth]{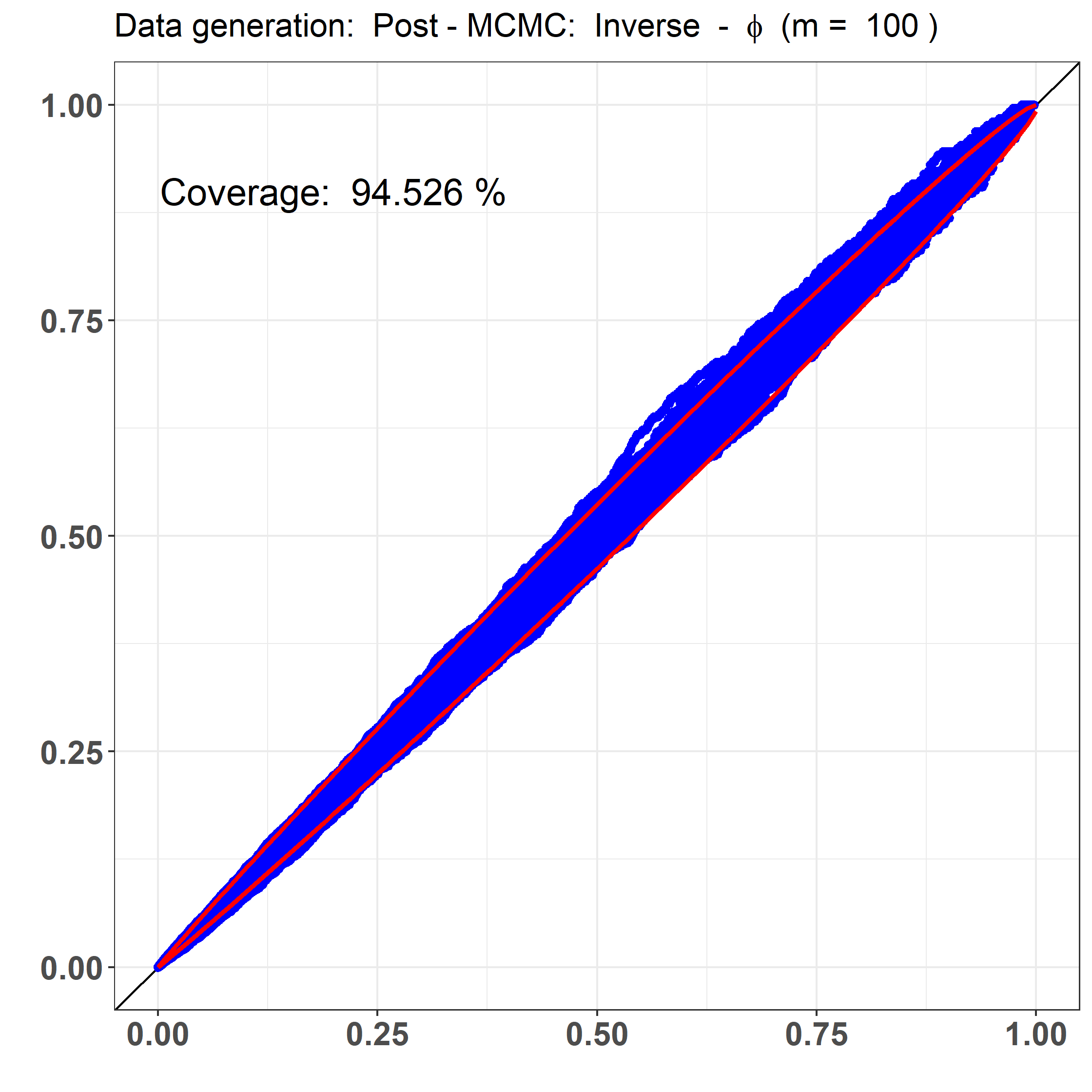}%
			\includegraphics[width=\wdIm\linewidth]{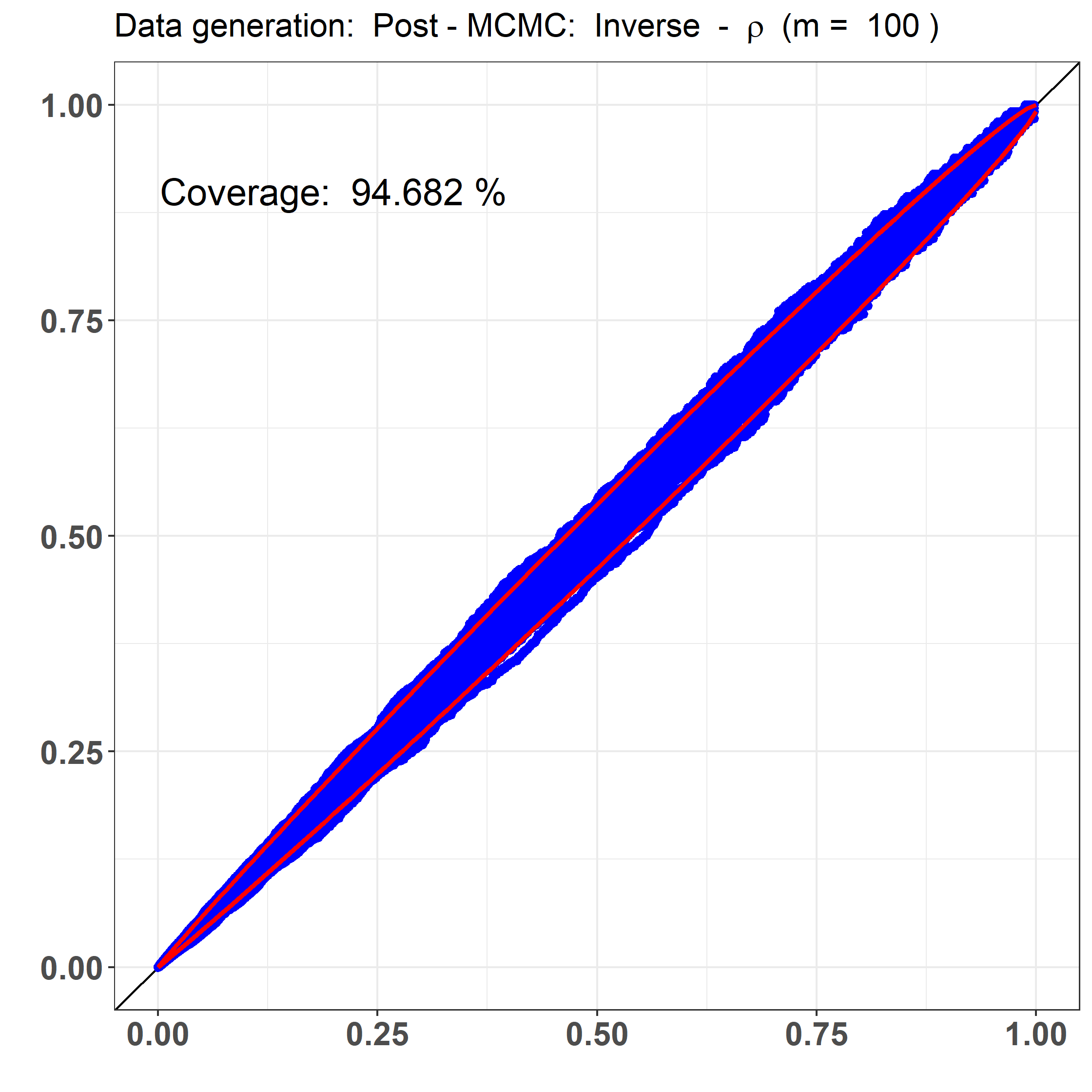}%
			\includegraphics[width=\wdIm\linewidth]{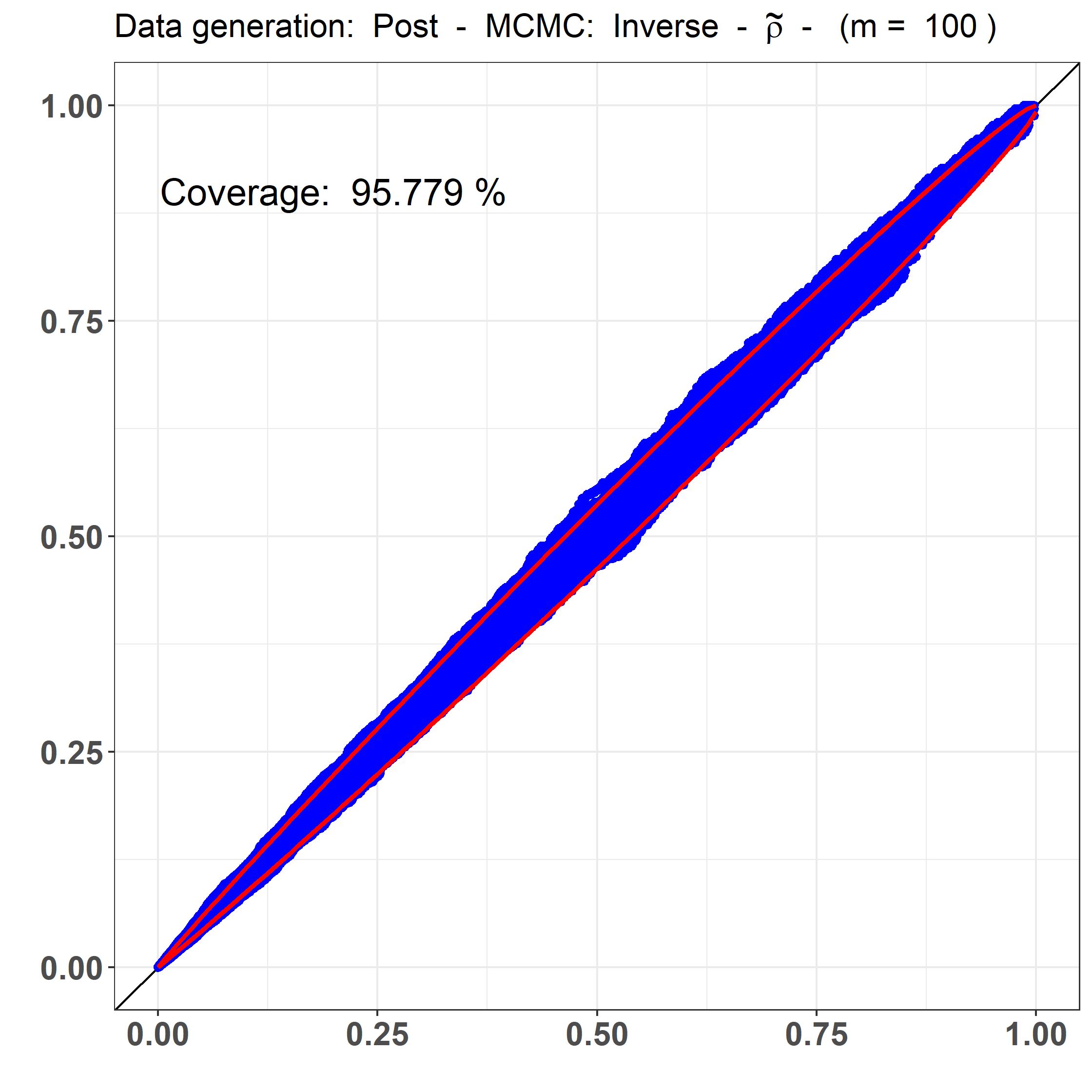}%
		\caption{Data generation: Post - MCMC: Inverse - $ m = 100 $}
	\end{sidewaysfigure}
	\newpage
	
	\begin{sidewaysfigure}
		\includegraphics[width=\wdIm\linewidth]{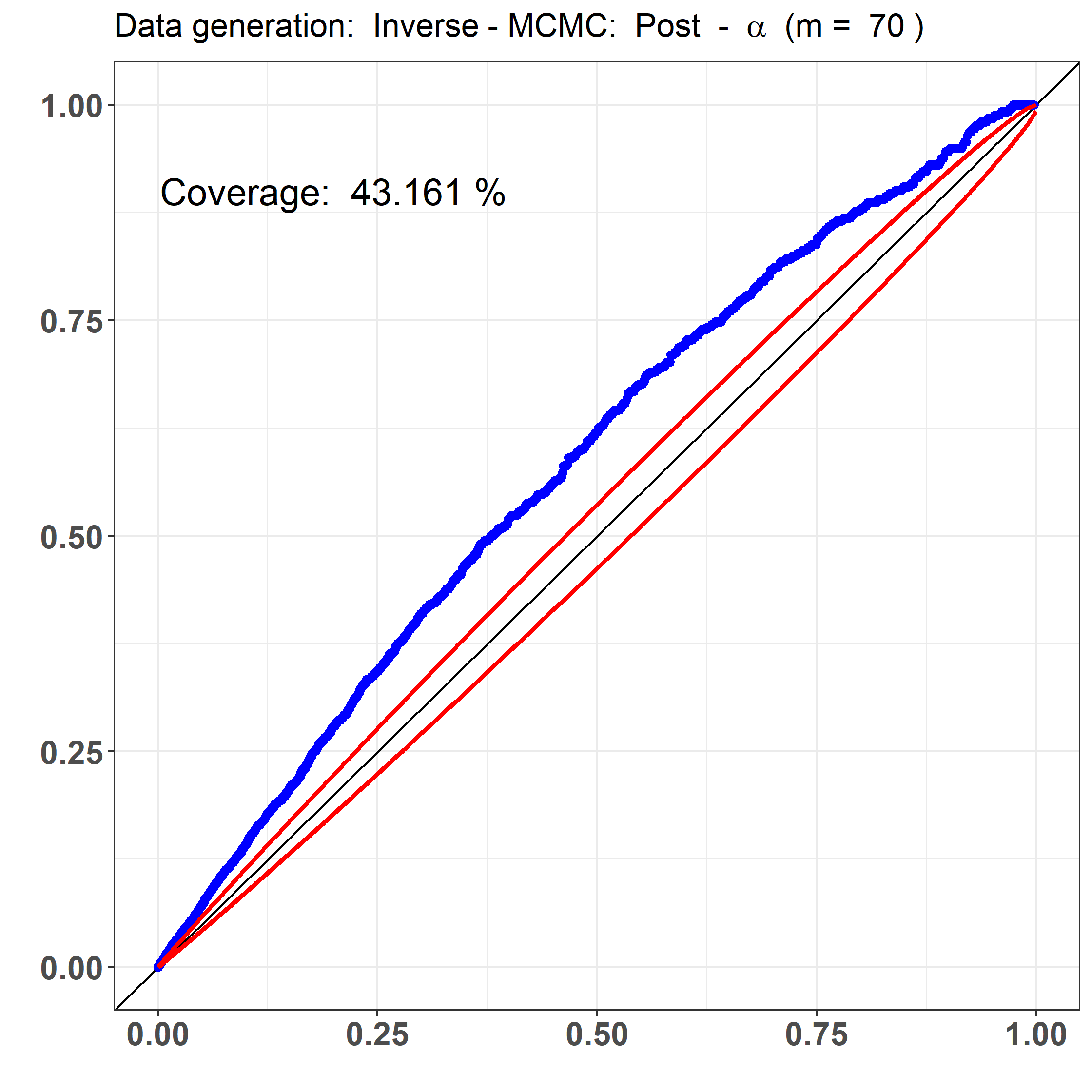}%
		\includegraphics[width=\wdIm\linewidth]{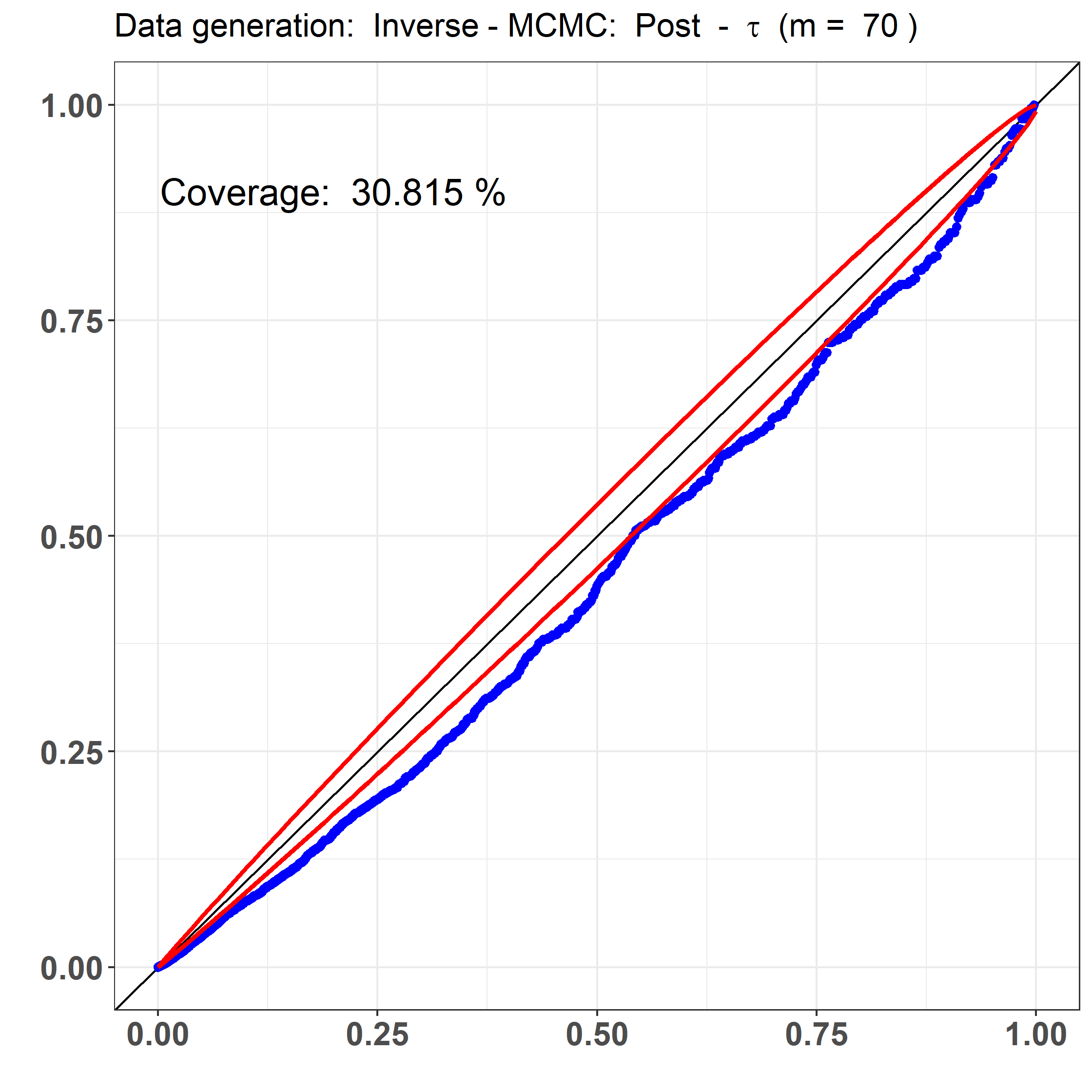}%
		\includegraphics[width=\wdIm\linewidth]{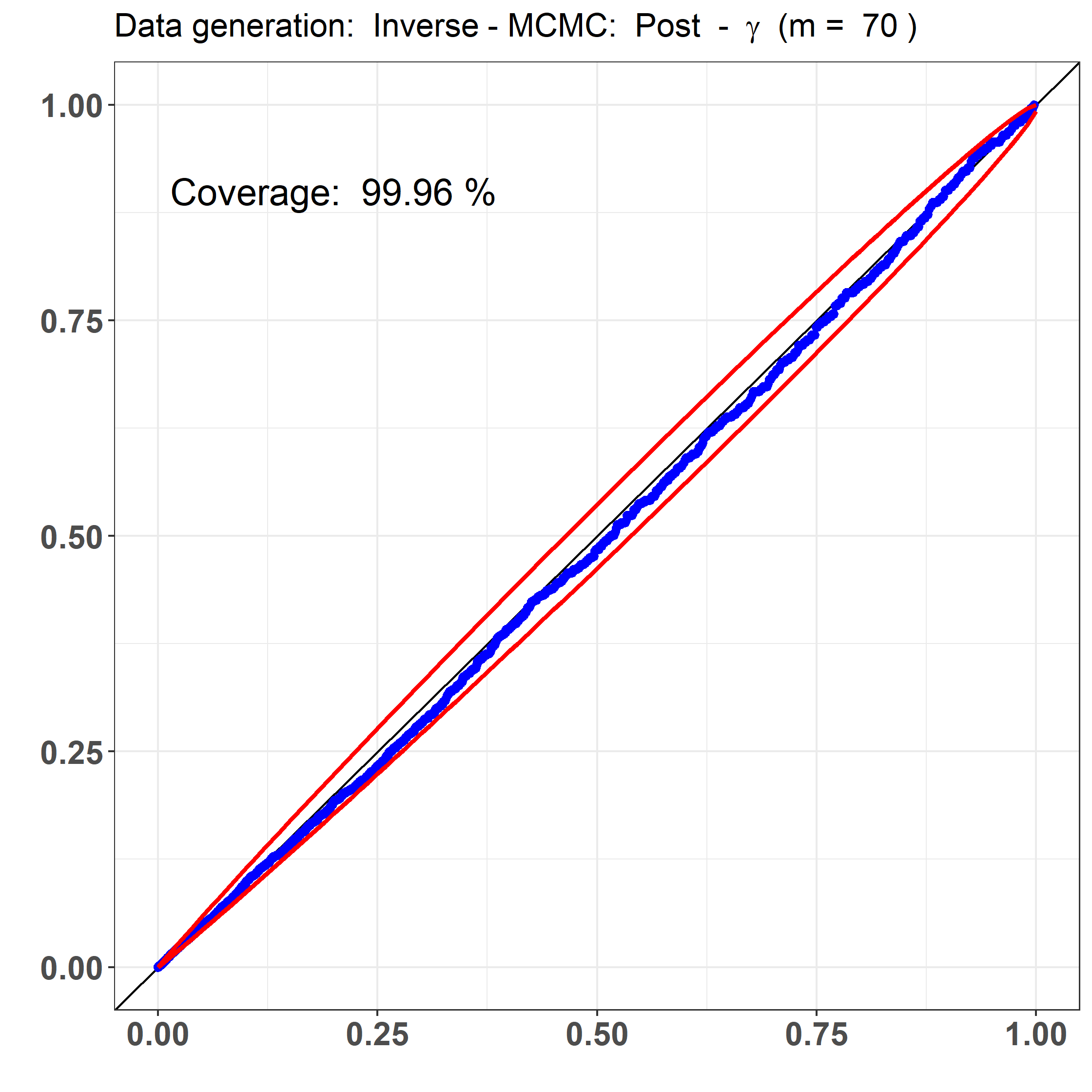} %
		\includegraphics[width=\wdIm\linewidth]{fig/ch5/inv/sbc_post_beta[1]_70}\\%
		\includegraphics[width=\wdIm\linewidth]{fig/ch5/inv/sbc_post_beta[2]_70}%
		\includegraphics[width=\wdIm\linewidth]{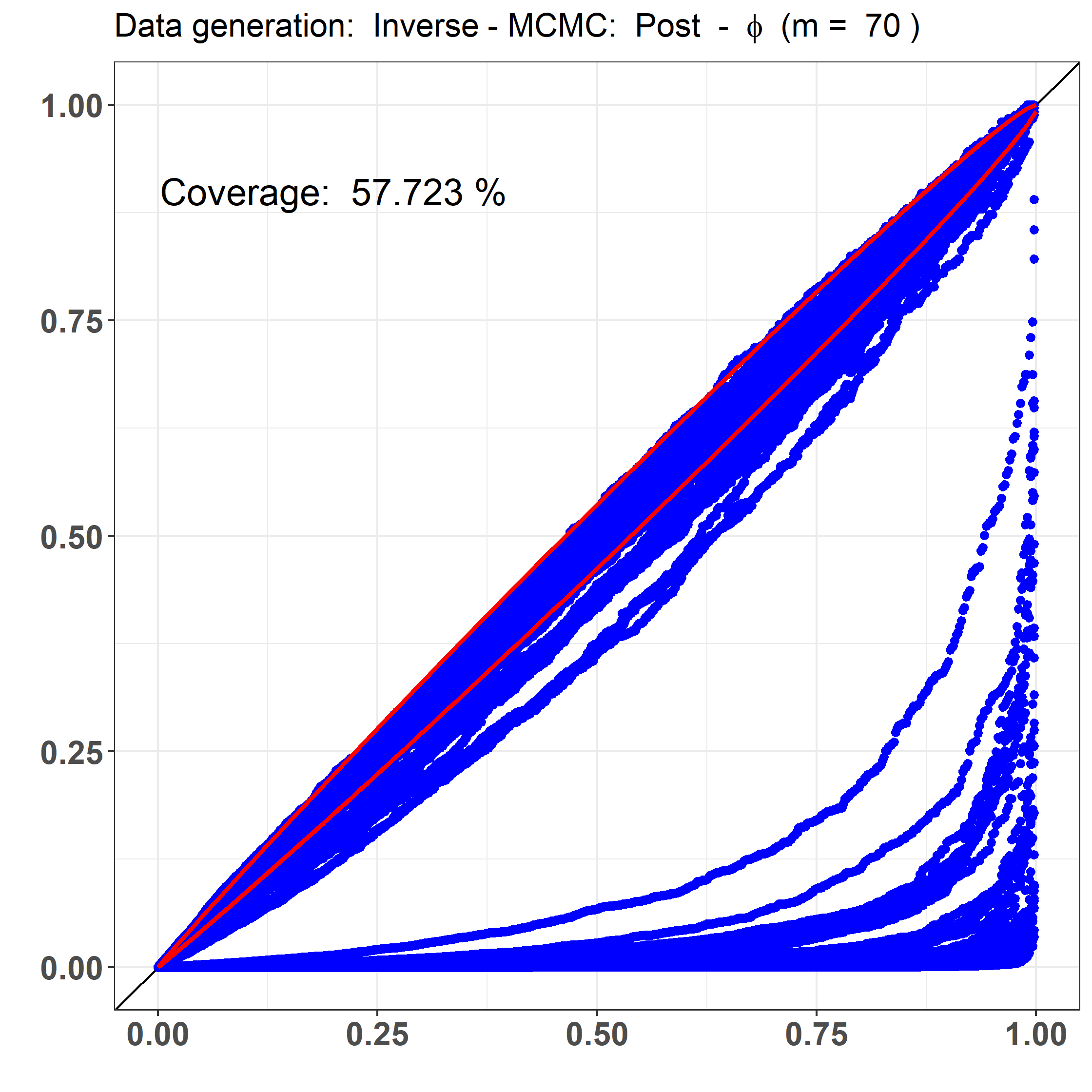}%
		\includegraphics[width=\wdIm\linewidth]{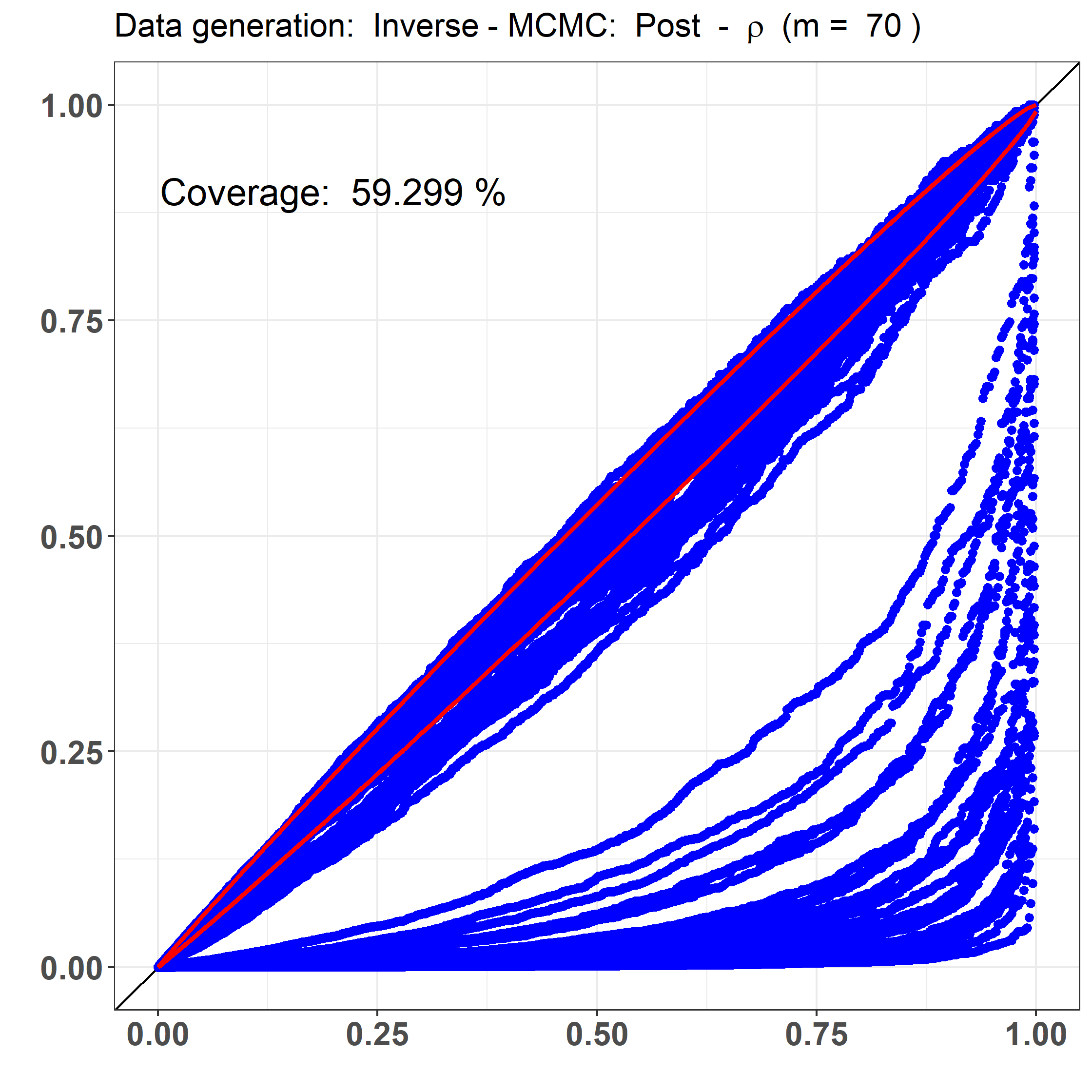}%
		\includegraphics[width=\wdIm\linewidth]{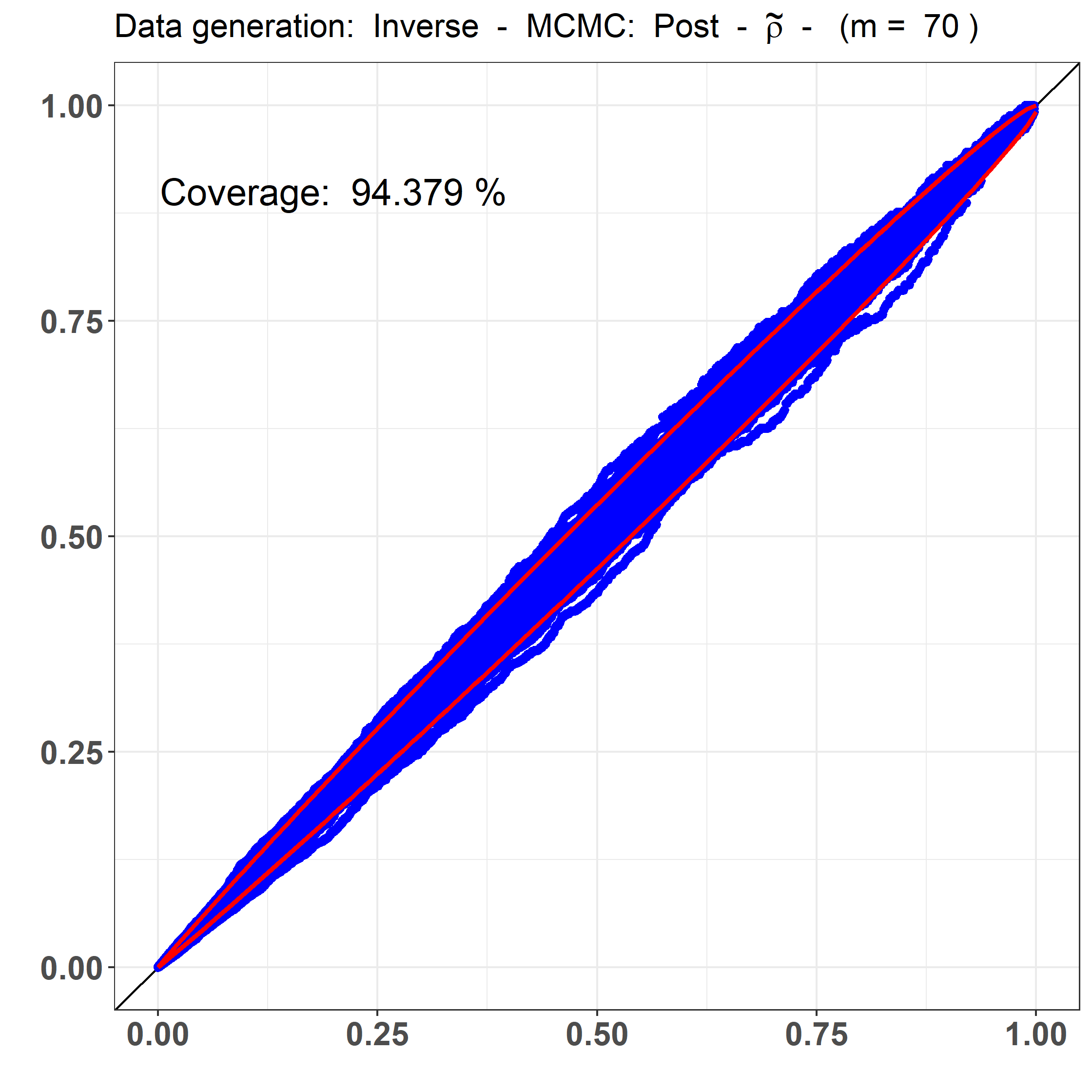}%
	\caption{Data generation: Post - MCMC: Inverse - $ m = 70 $}
\end{sidewaysfigure}
\newpage

\begin{sidewaysfigure}
		\includegraphics[width=\wdIm\linewidth]{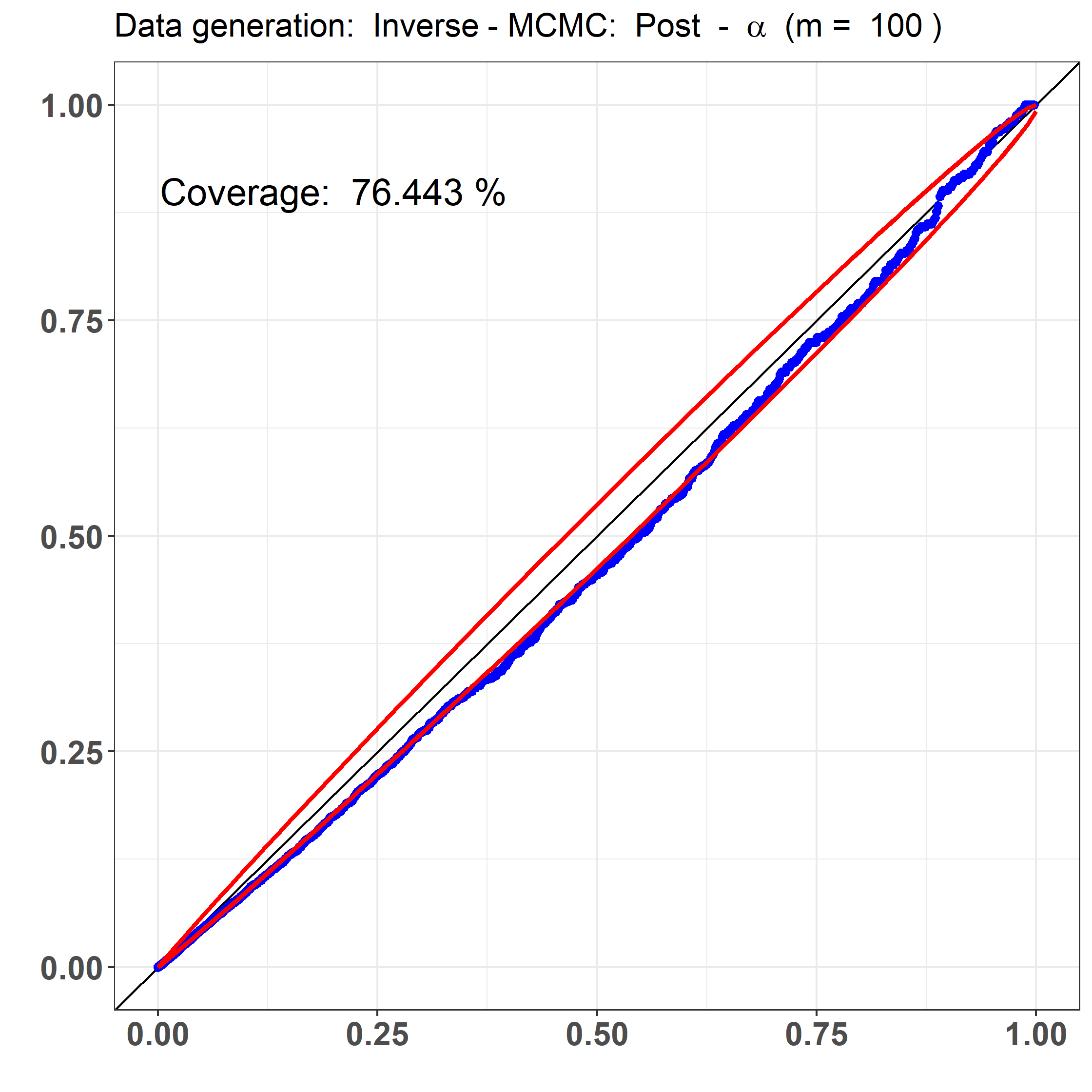}%
		\includegraphics[width=\wdIm\linewidth]{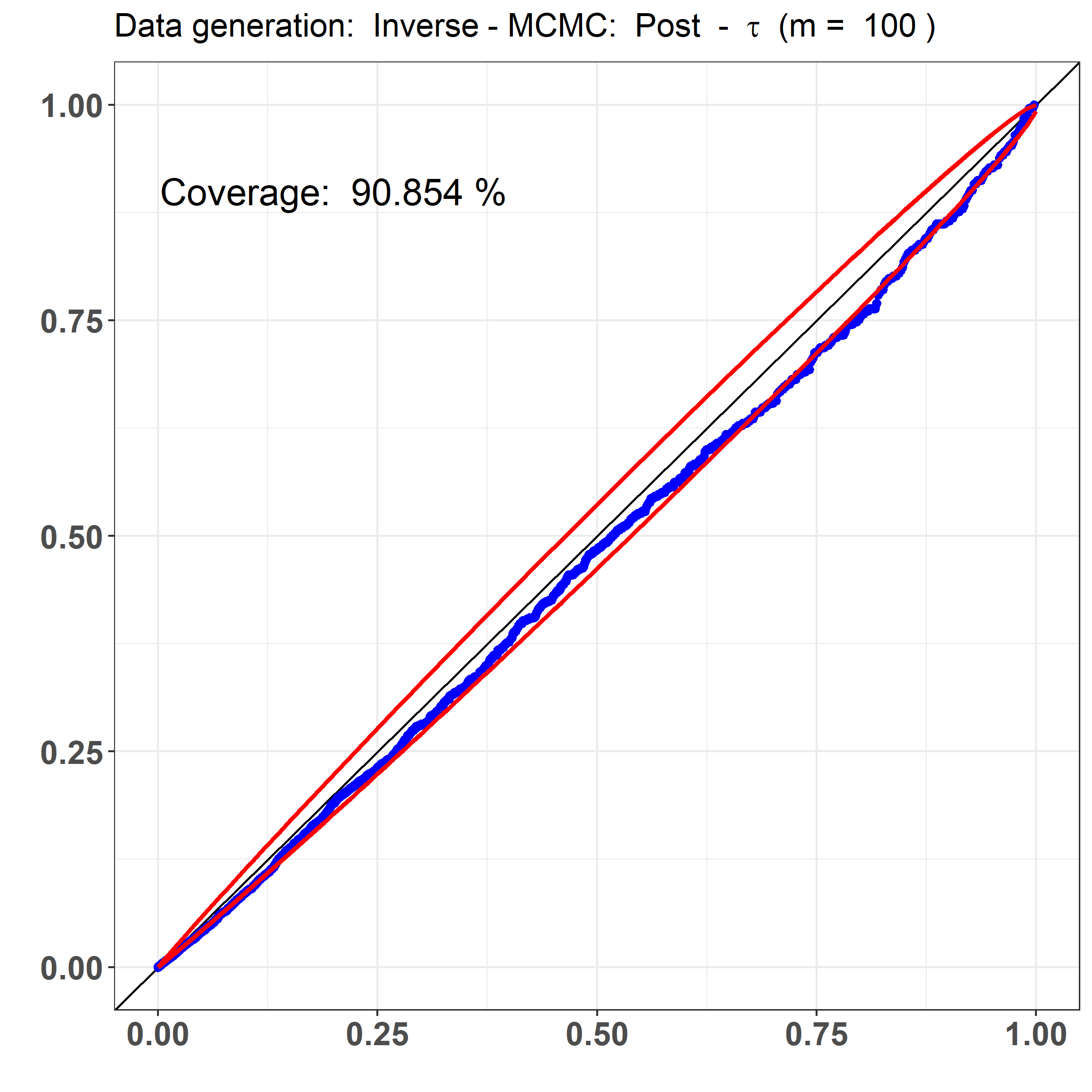}%
		\includegraphics[width=\wdIm\linewidth]{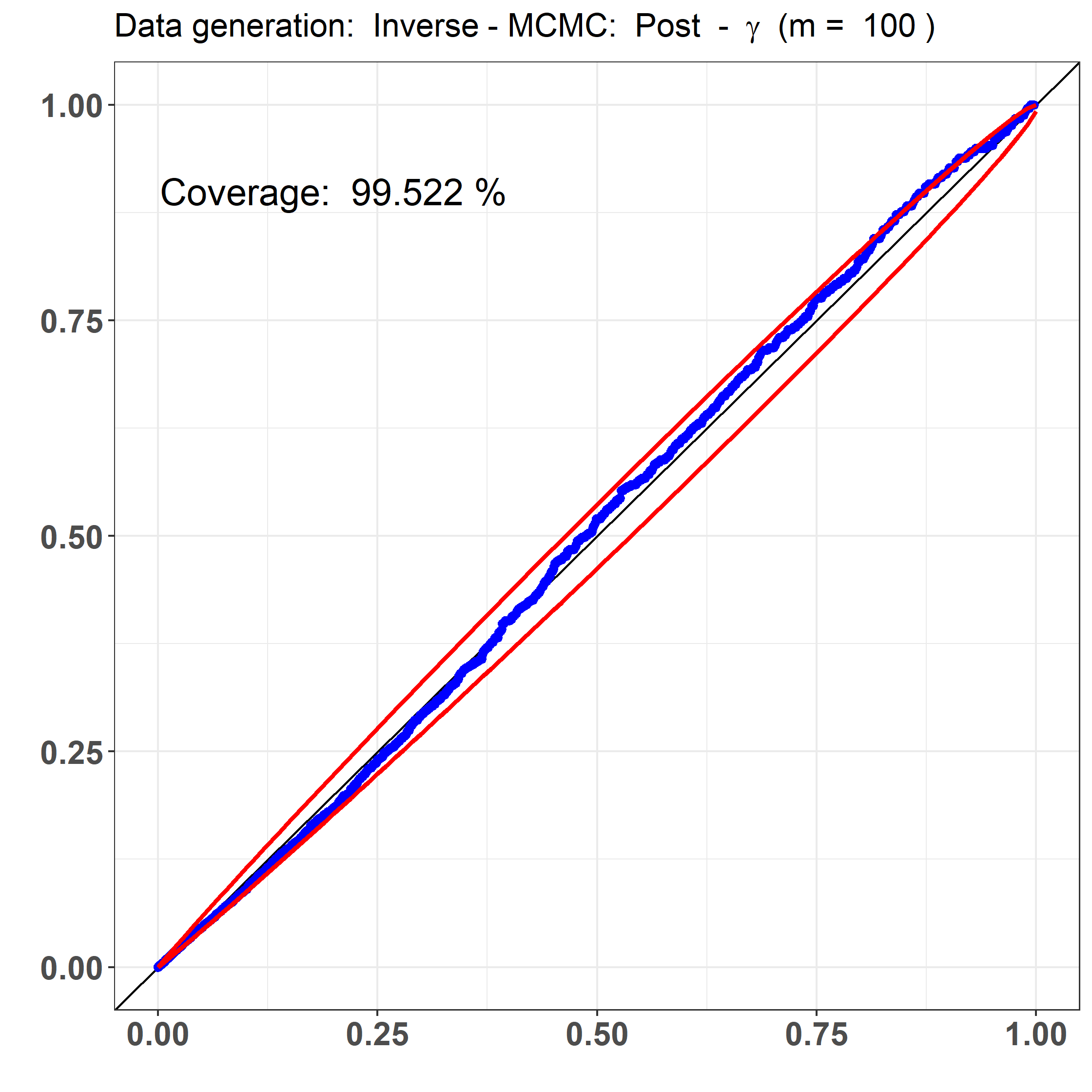} %
		\includegraphics[width=\wdIm\linewidth]{fig/ch5/inv/sbc_post_beta[1]_100}\\%
		\includegraphics[width=\wdIm\linewidth]{fig/ch5/inv/sbc_post_beta[2]_100}%
		\includegraphics[width=\wdIm\linewidth]{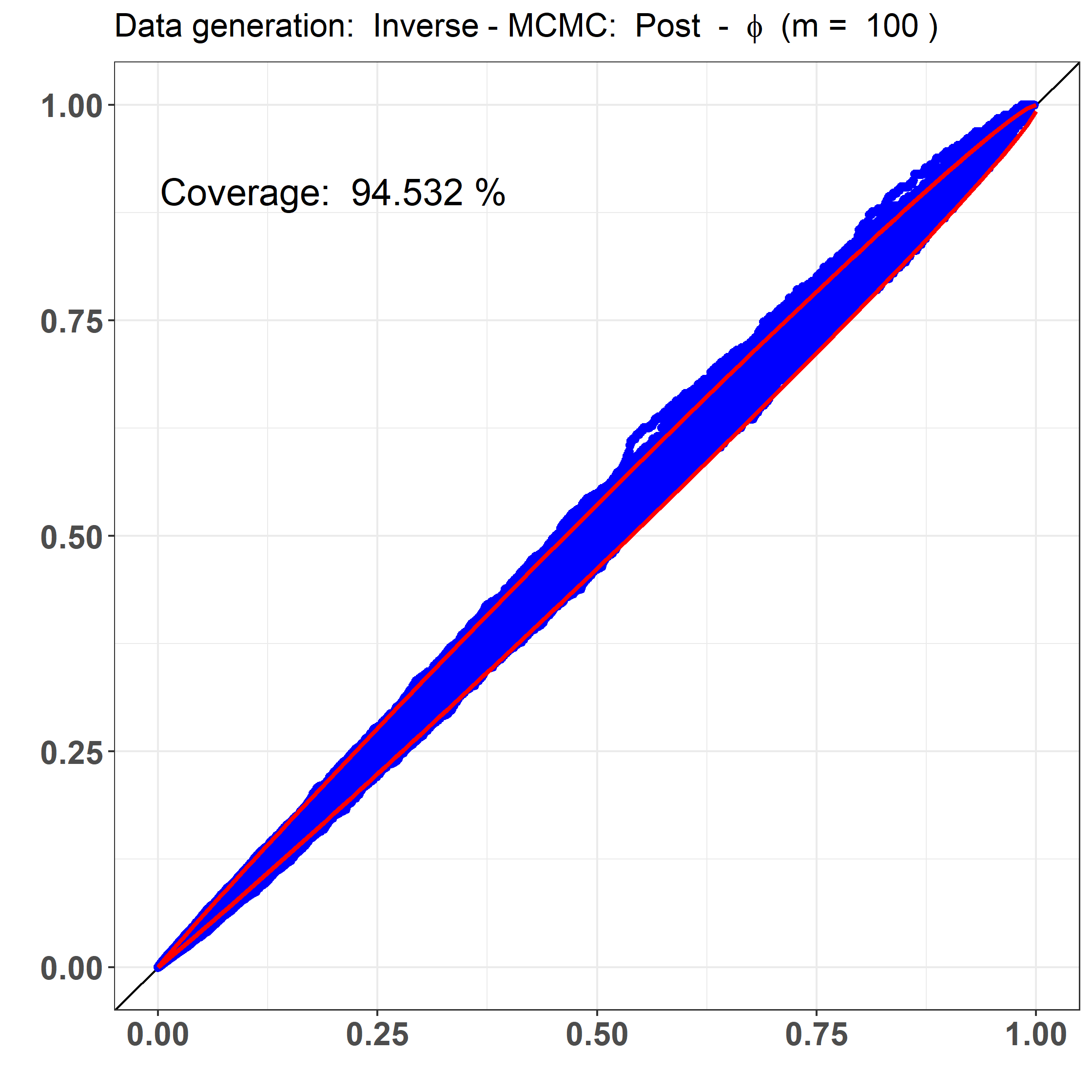}%
		\includegraphics[width=\wdIm\linewidth]{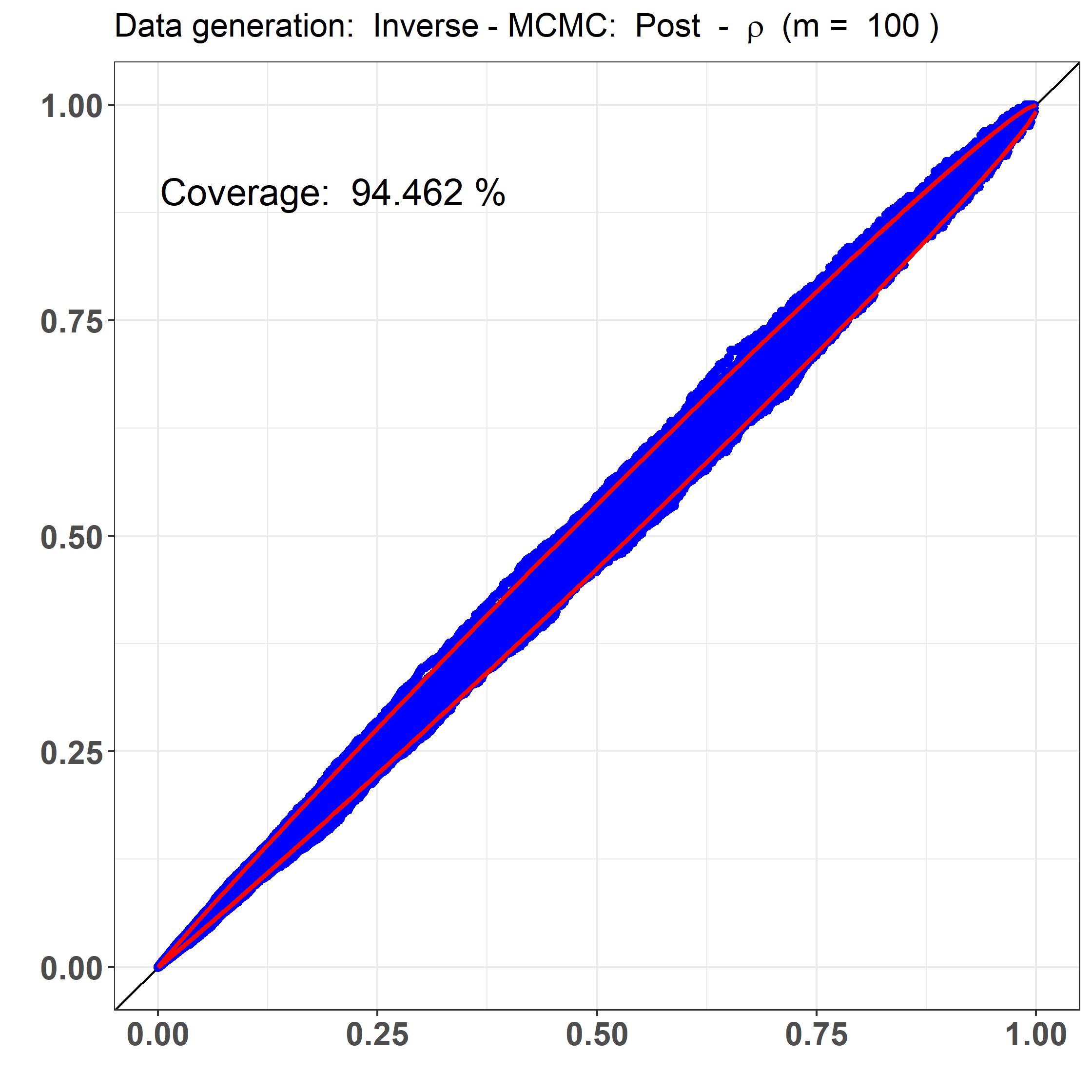}%
		\includegraphics[width=\wdIm\linewidth]{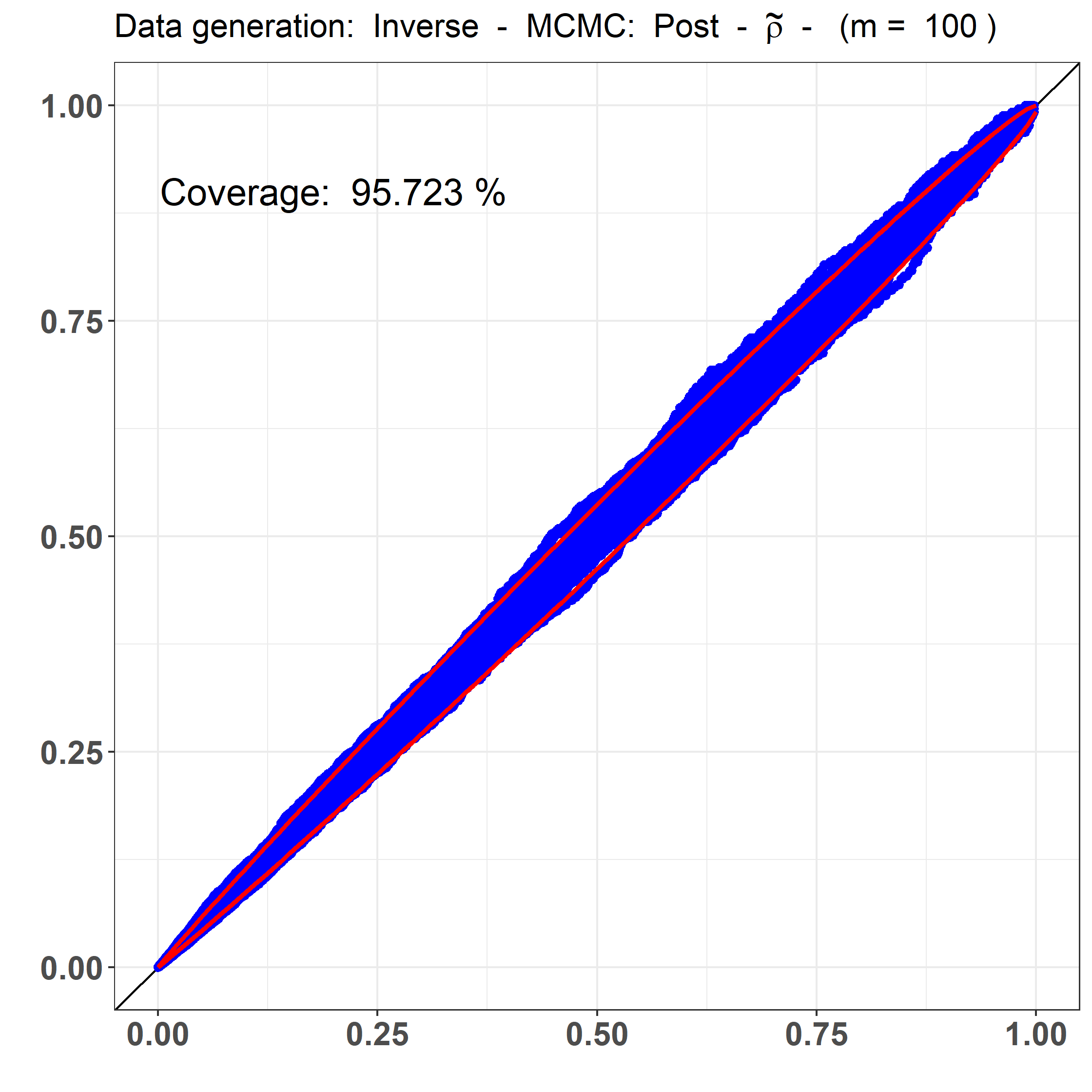}%
	\caption{Data generation: Inverse - MCMC: Post - $ m = 100 $}
\end{sidewaysfigure}
\newpage

\begin{sidewaysfigure}
		\includegraphics[width=\wdIm\linewidth]{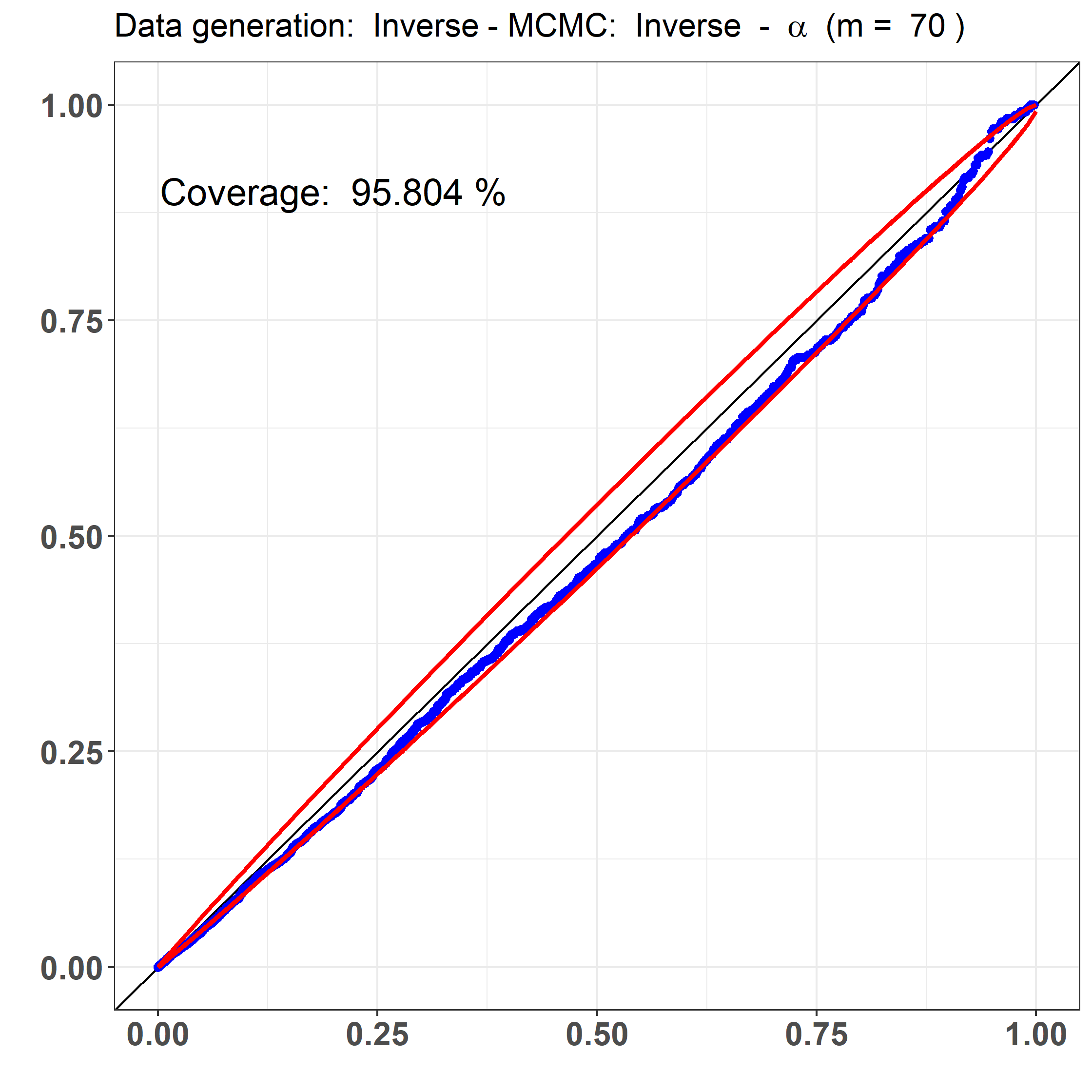}%
		\includegraphics[width=\wdIm\linewidth]{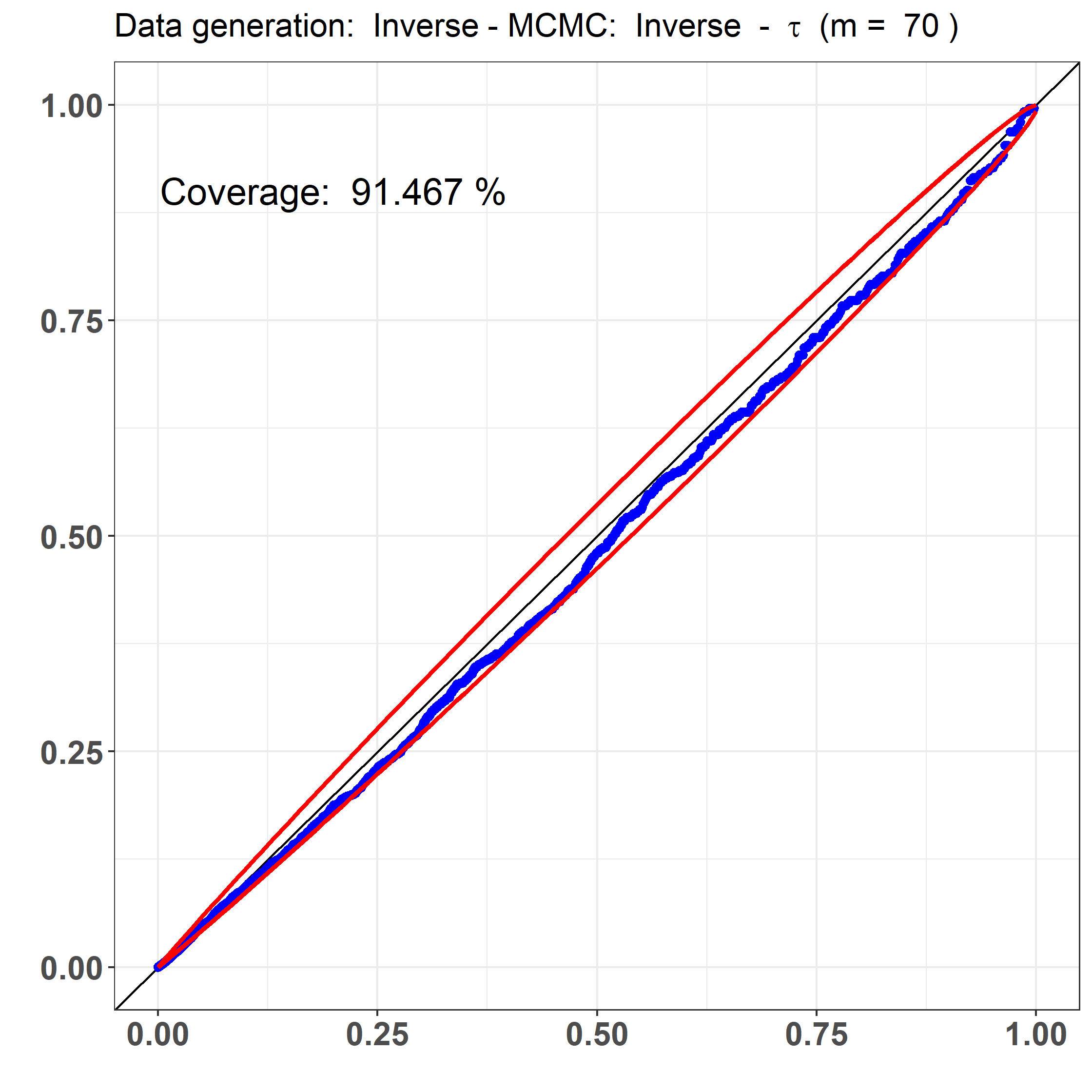}%
		\includegraphics[width=\wdIm\linewidth]{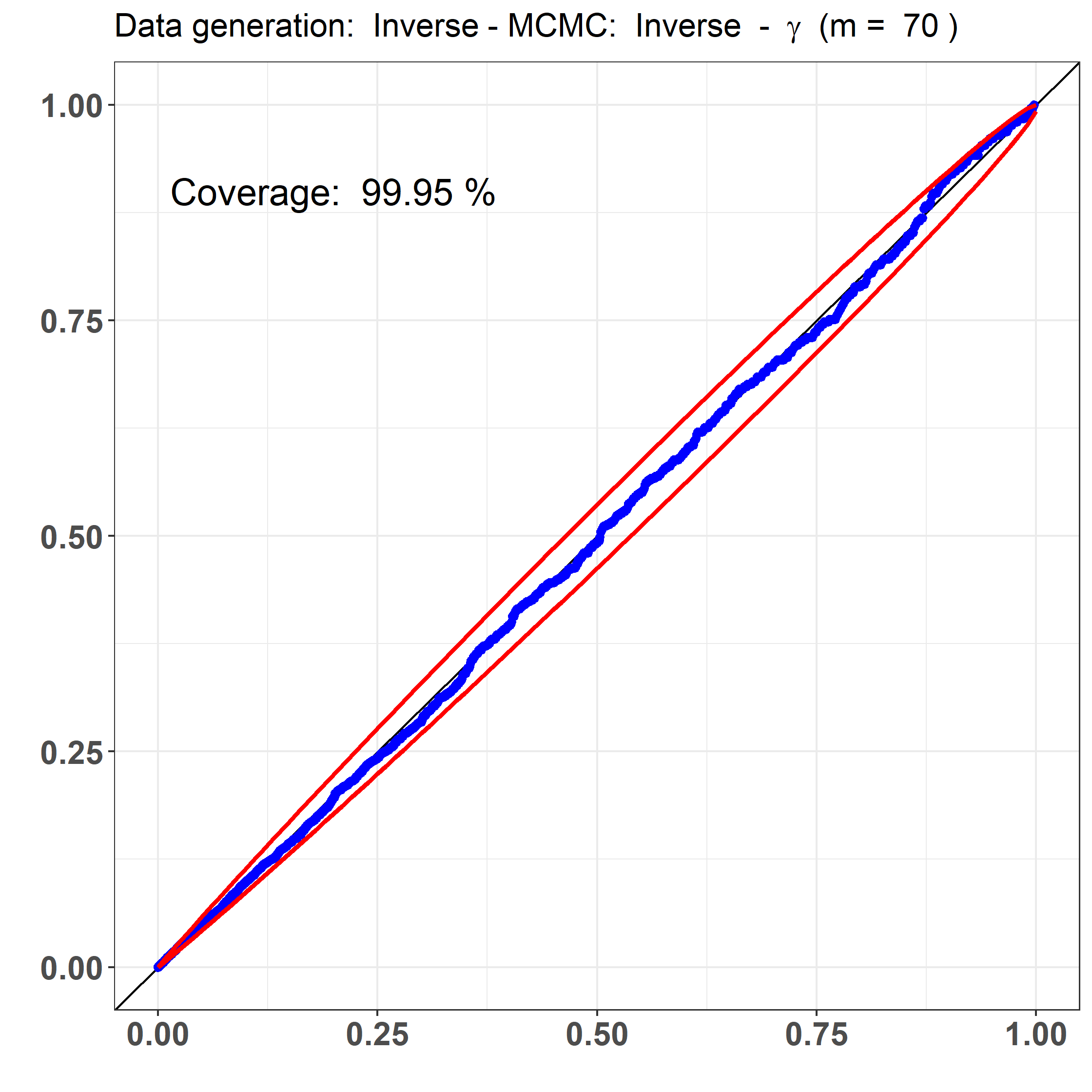} %
		\includegraphics[width=\wdIm\linewidth]{fig/ch5/inv/sbc_inv_beta[1]_70}\\%
		\includegraphics[width=\wdIm\linewidth]{fig/ch5/inv/sbc_inv_beta[2]_70}%
		\includegraphics[width=\wdIm\linewidth]{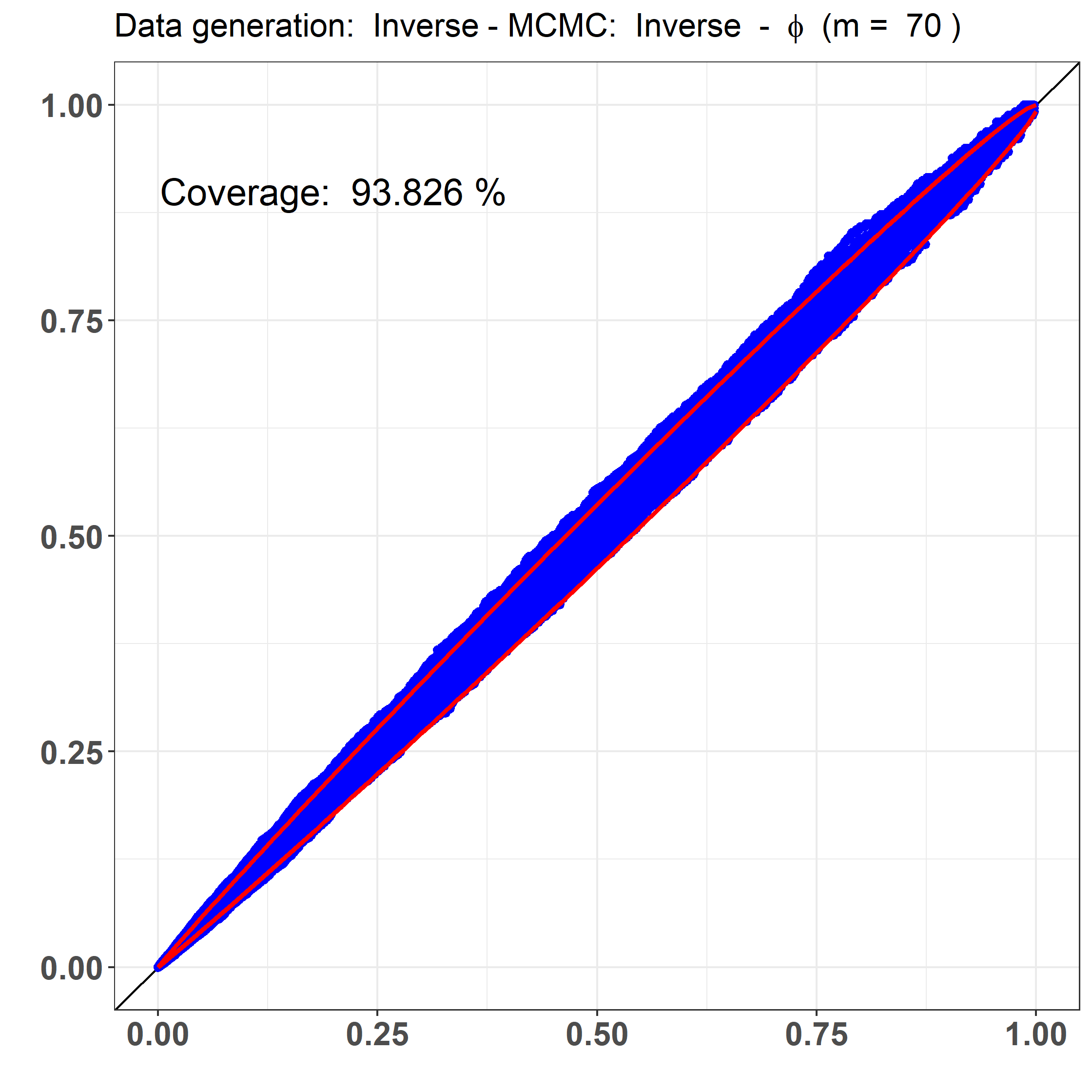}%
		\includegraphics[width=\wdIm\linewidth]{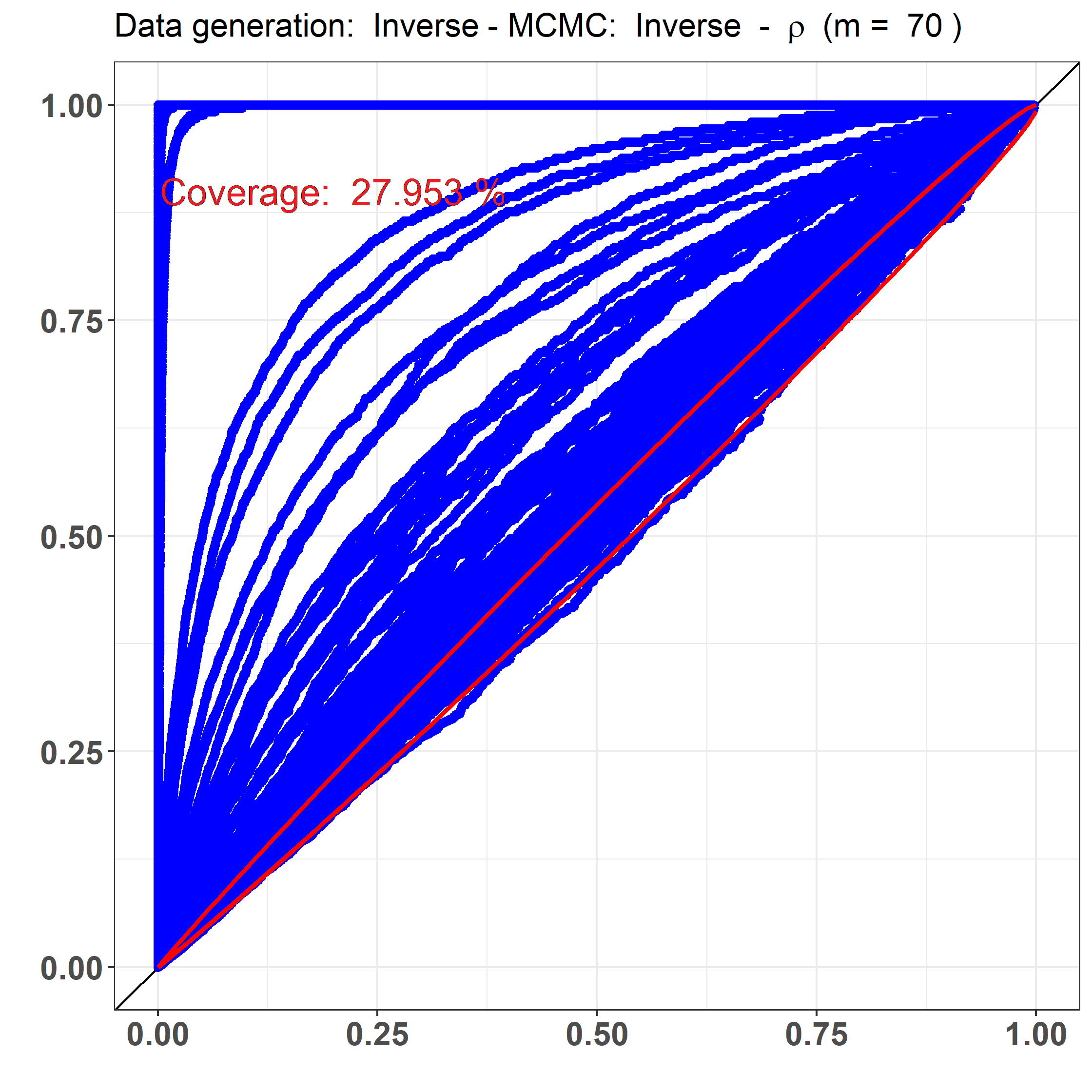}%
		\includegraphics[width=\wdIm\linewidth]{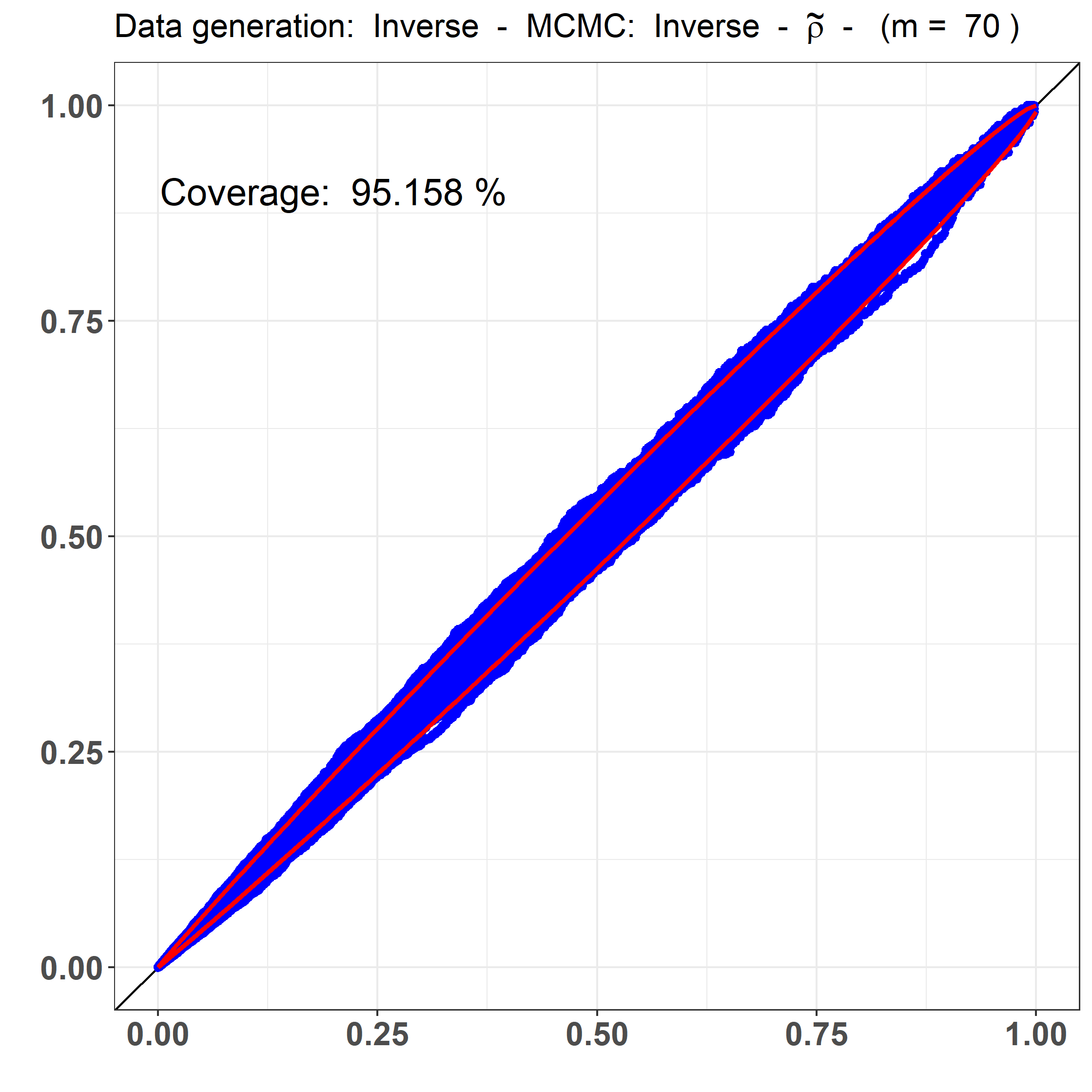}%
	\caption{Data generation: Inverse - MCMC: Inverse - $ m = 70 $}
\end{sidewaysfigure}
\newpage

\begin{sidewaysfigure}
		\includegraphics[width=\wdIm\linewidth]{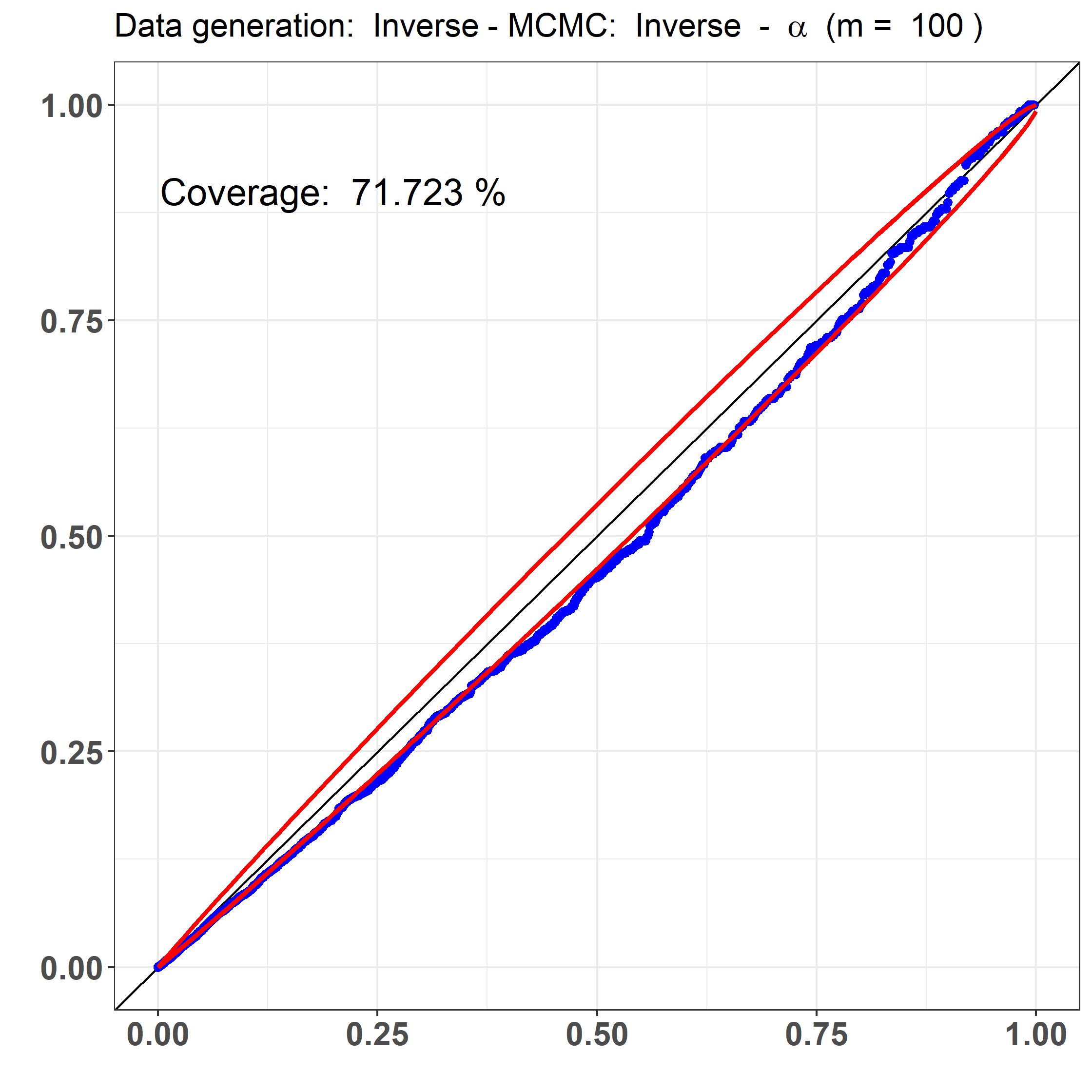}%
		\includegraphics[width=\wdIm\linewidth]{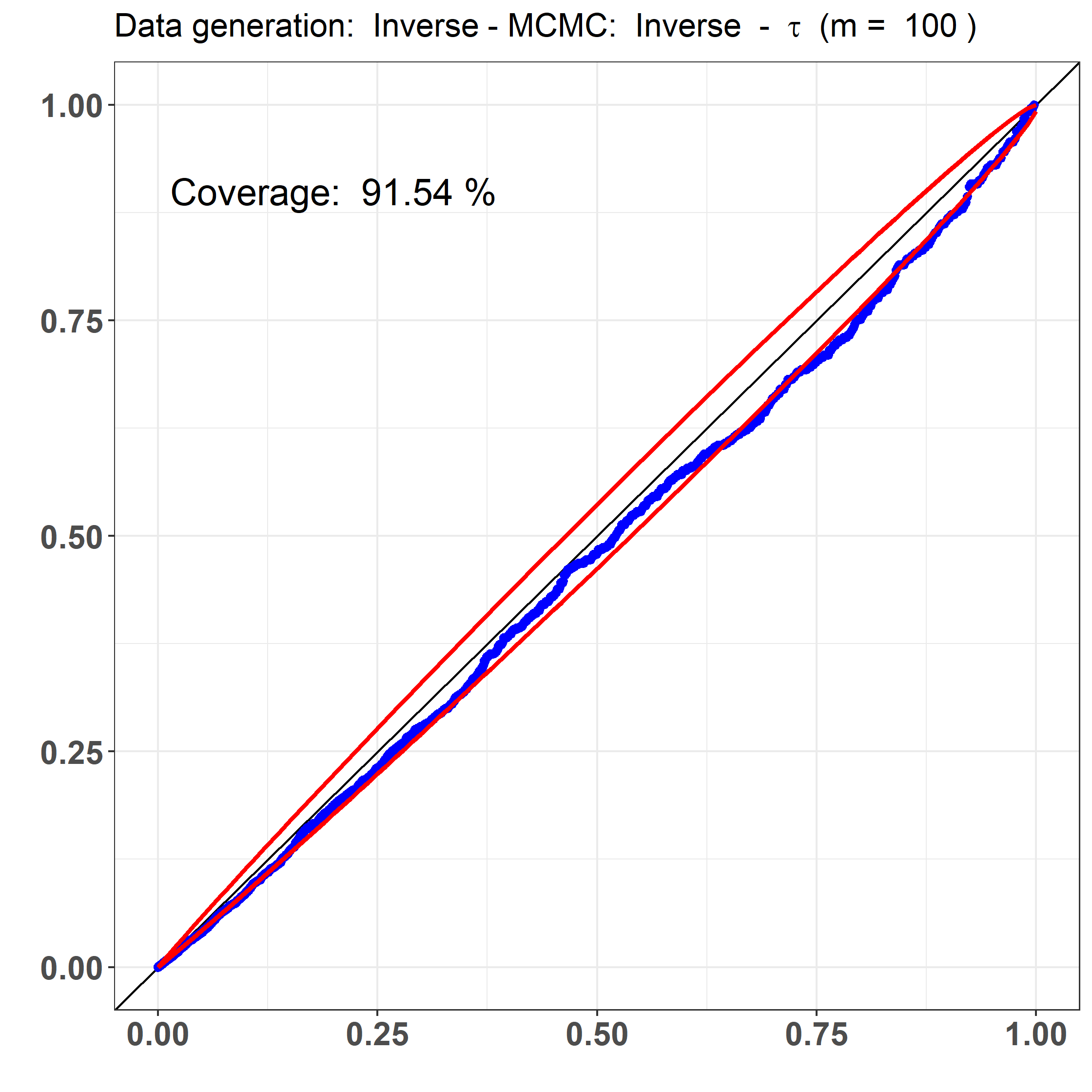}%
		\includegraphics[width=\wdIm\linewidth]{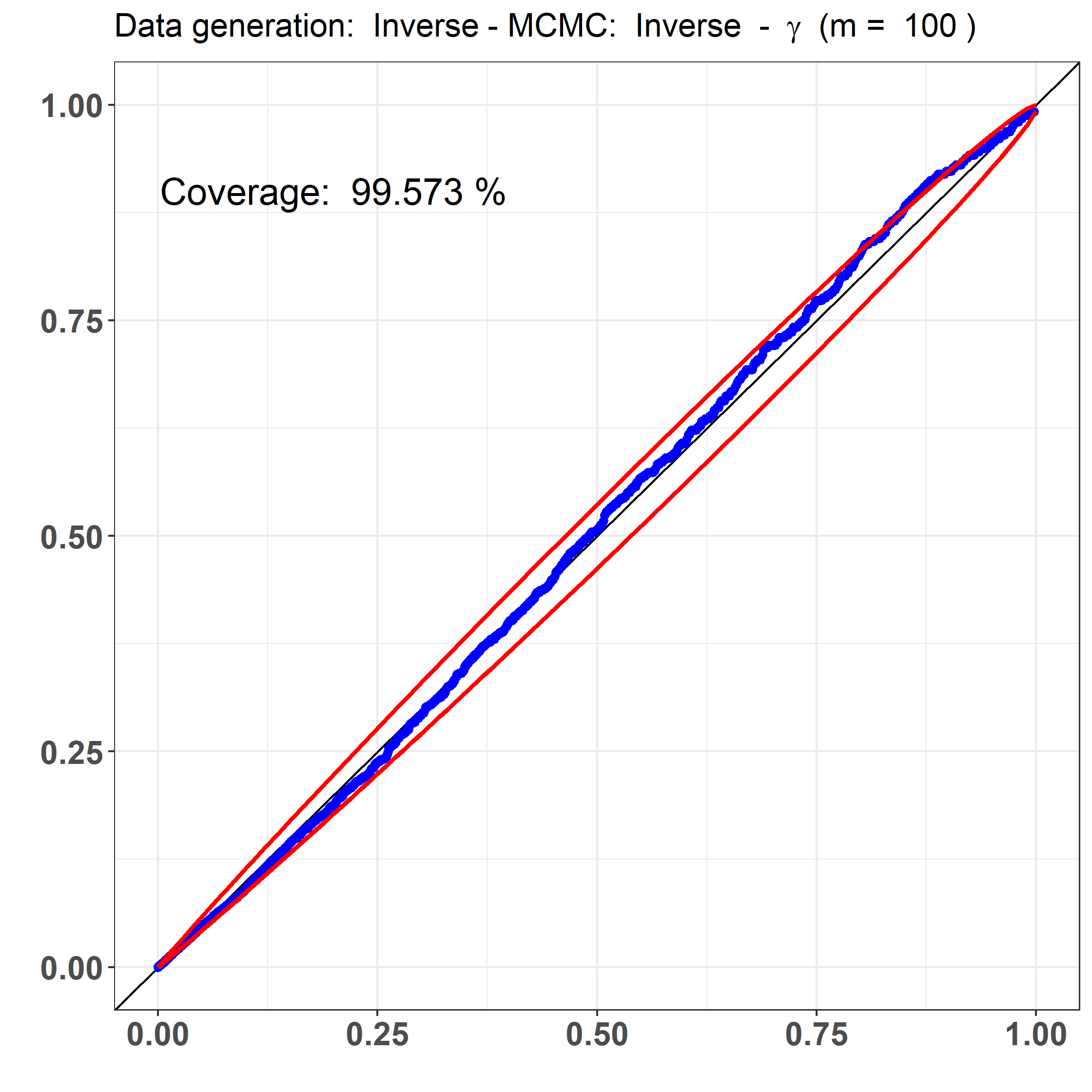} %
		\includegraphics[width=\wdIm\linewidth]{fig/ch5/inv/sbc_inv_beta[1]_100}\\%
		\includegraphics[width=\wdIm\linewidth]{fig/ch5/inv/sbc_inv_beta[2]_100}%
		\includegraphics[width=\wdIm\linewidth]{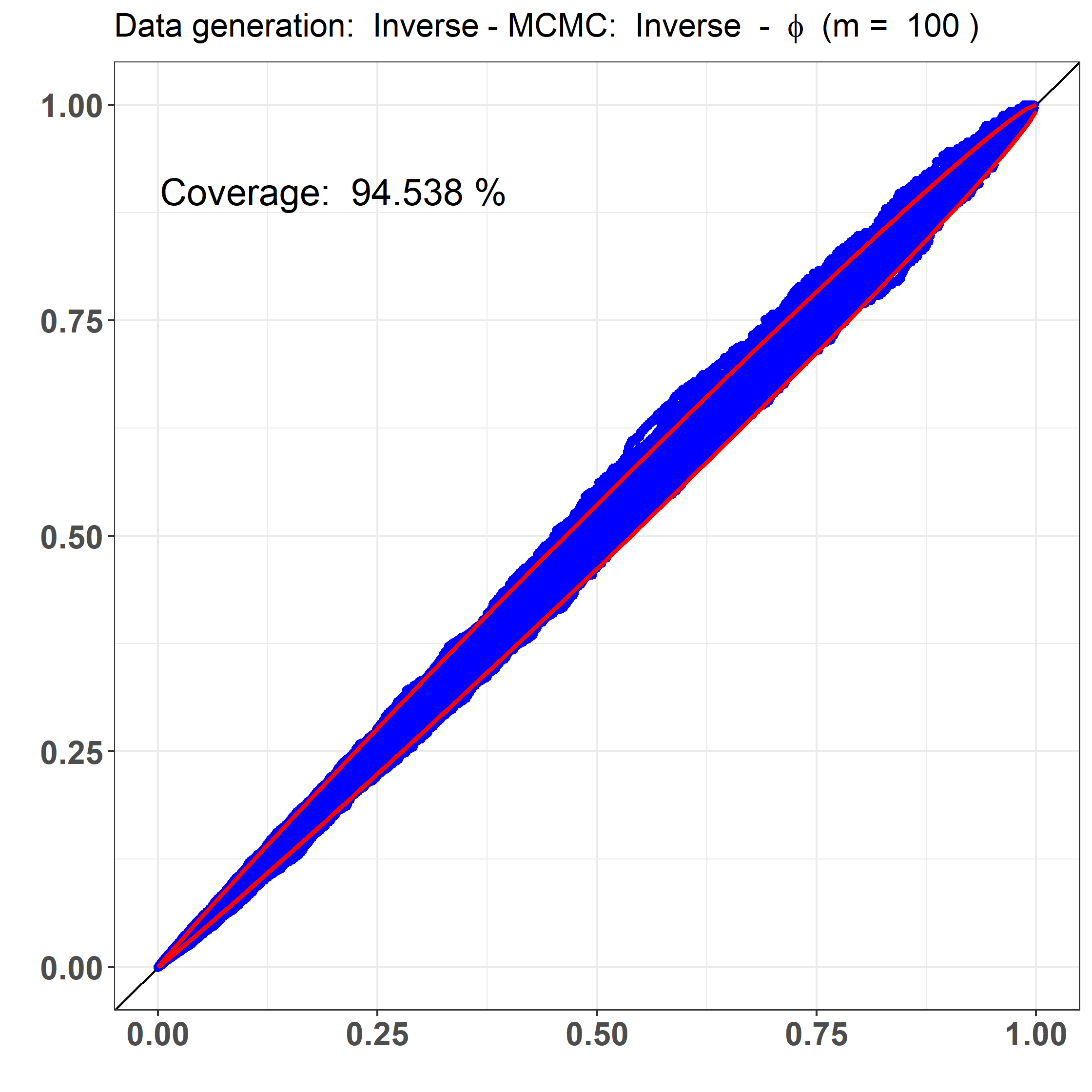}%
		\includegraphics[width=\wdIm\linewidth]{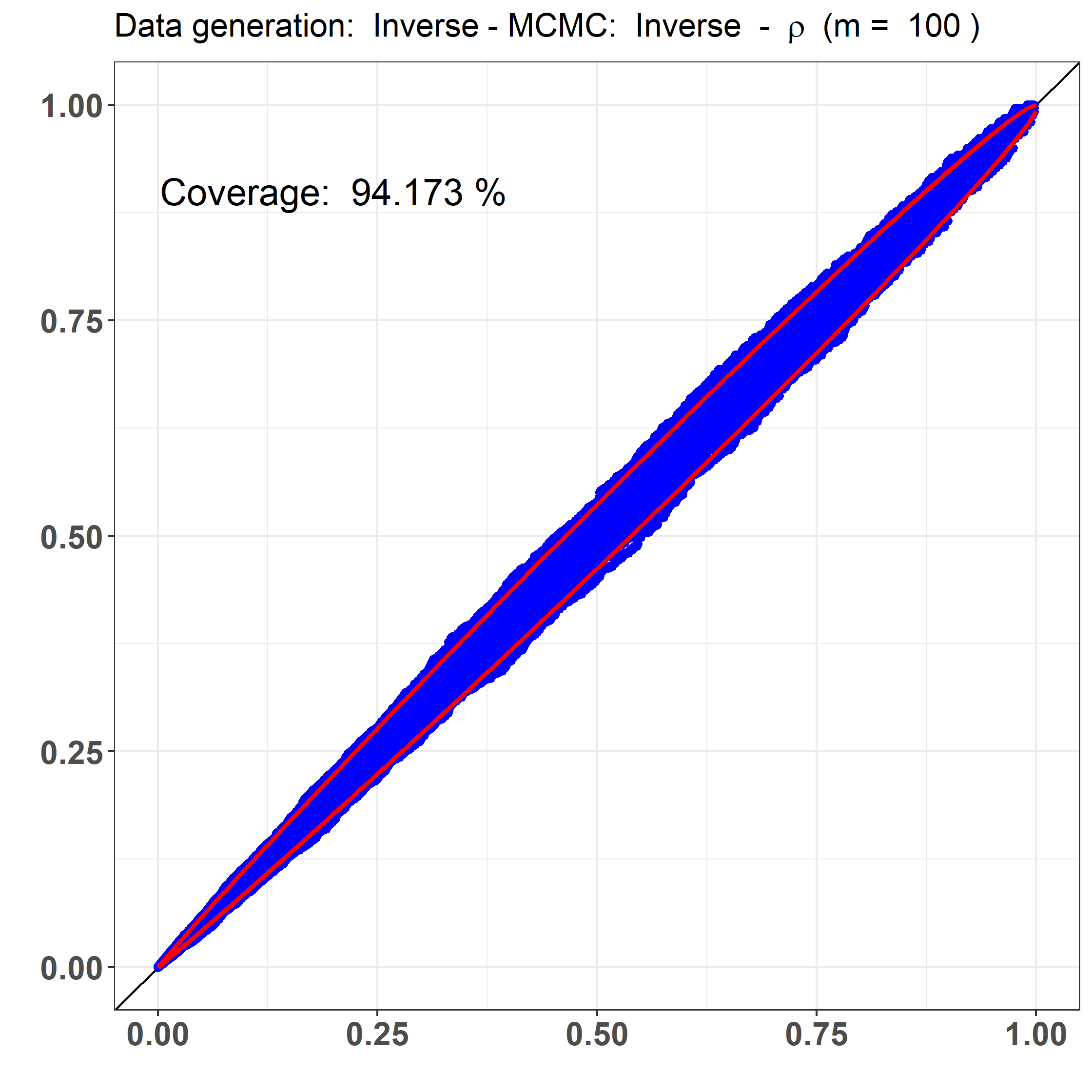}%
		\includegraphics[width=\wdIm\linewidth]{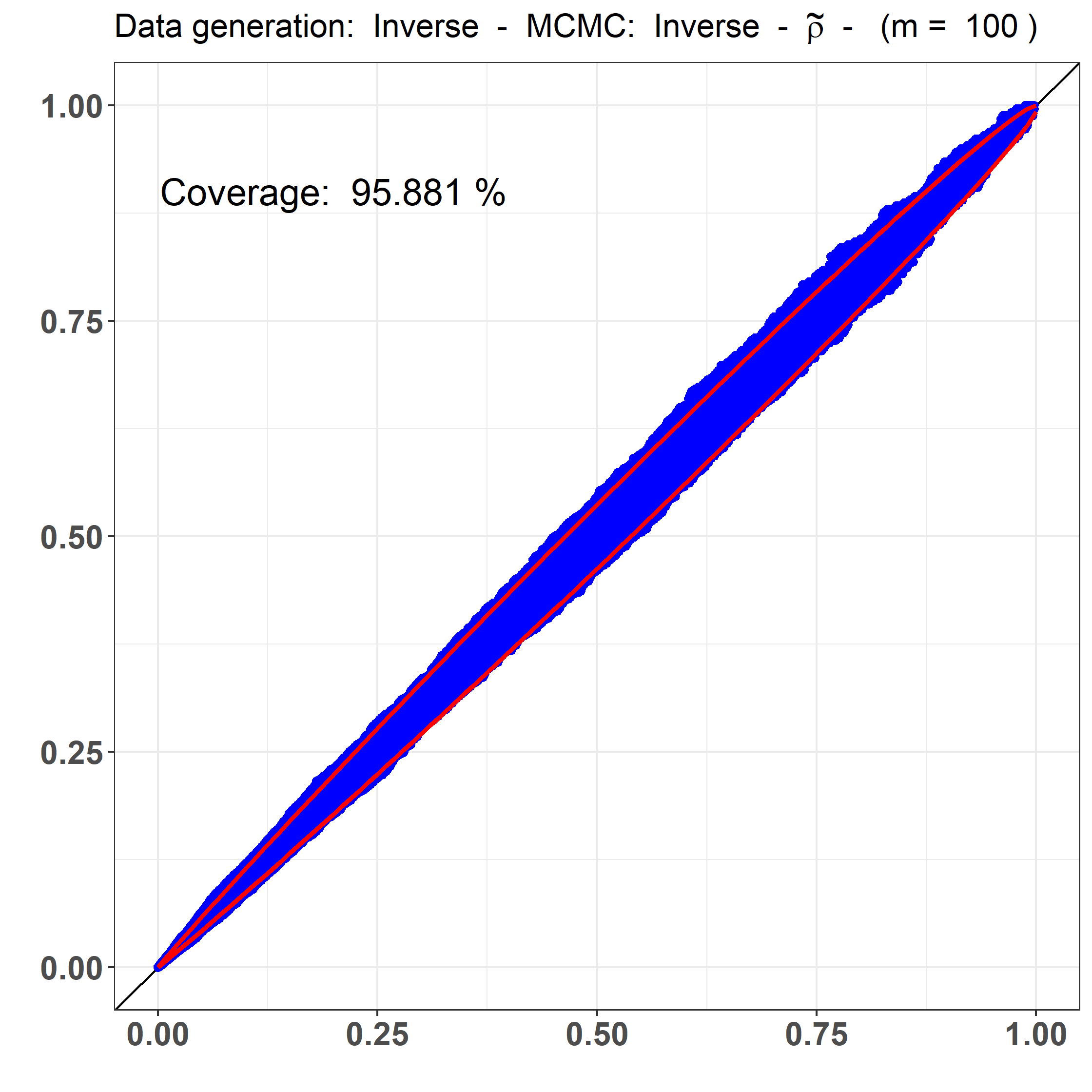}%
	\caption{Data generation: Inverse - MCMC: Inverse - $ m = 100 $}
\end{sidewaysfigure}
\clearpage

\clearpage

\section{Data analysis}
\begin{figure}[h!]
	\caption{Q-Q plots of marginal mixed predictive values versus ordered rank statistic of a uniform distribution for each model run on the SEL dataset}
	\includegraphics[width=\linewidth]{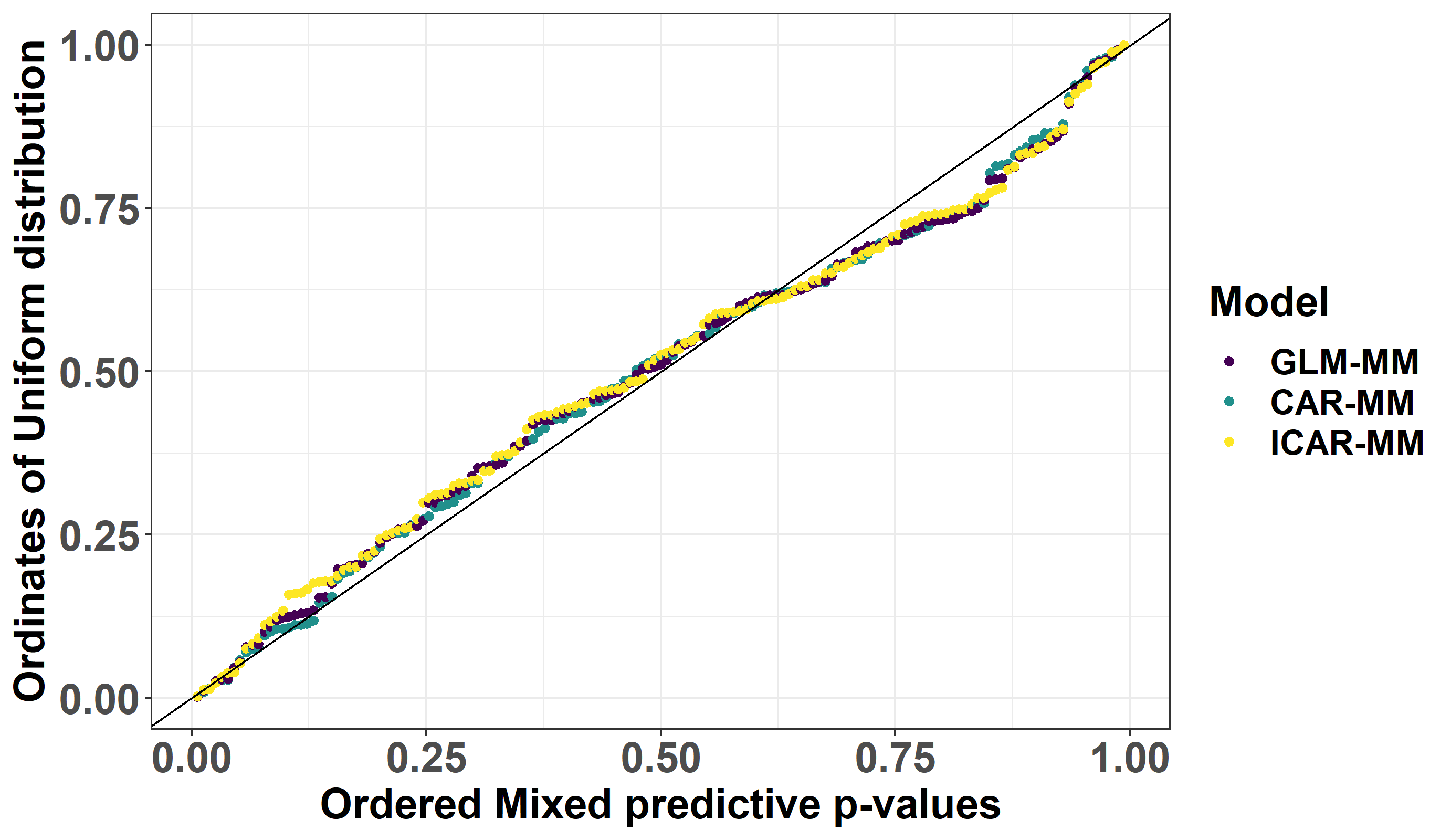}
	\label{fig:post_cal_sel}
\end{figure}

\begin{figure}[h!]
	\centering
	\caption{Posterior means for each areal relative risk ($ \bm{\rho} $ in Equation 7.2) for the ICAR-MM model}
	\includegraphics[width=\linewidth]{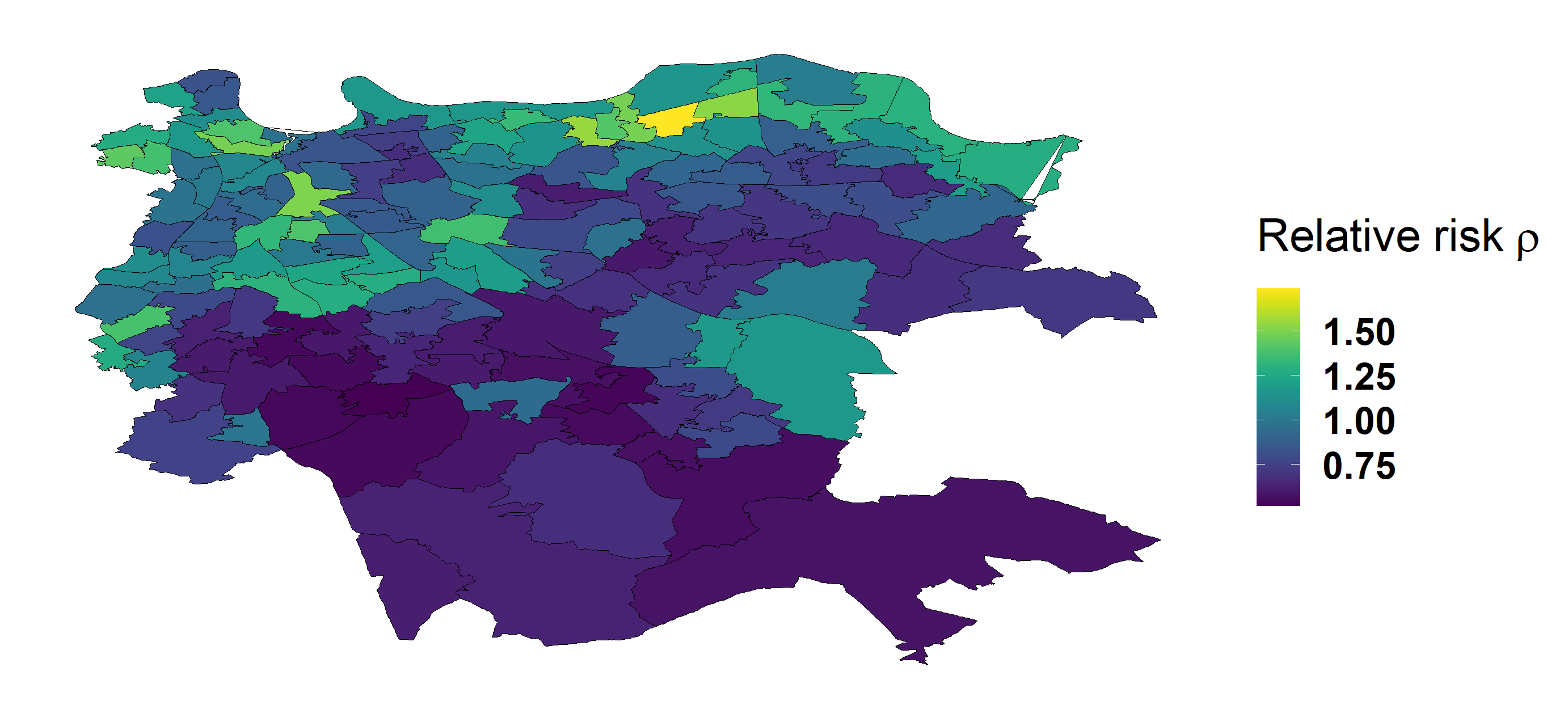}
	\label{fig:rr_map_sel_flat_nb}
\end{figure}

\begin{figure}[h!]
	\centering
	\caption{Posterior means for each areal relative risk ($ \bm{\rho} $ in Equation 7.2) for the GLM-MM model}
	\includegraphics[width=\linewidth]{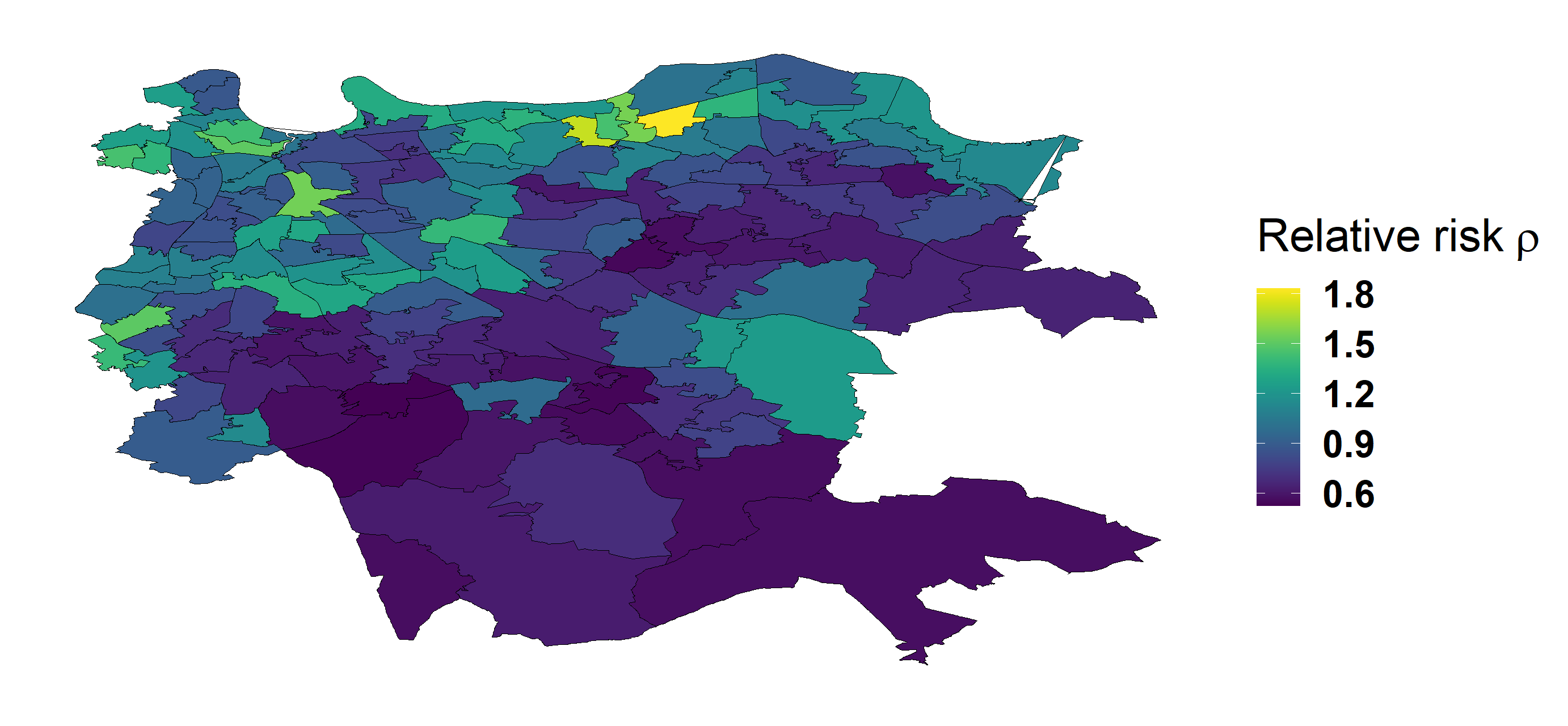}
	\label{fig:rr_map_sel_glm_nb}
\end{figure}
